\documentclass[11pt,a4paper]{article}
\pdfoutput=1 
\usepackage[margin=1in]{geometry}

\usepackage{amsfonts}
\usepackage{amsmath}
\usepackage{amssymb}
\usepackage{amsthm}
\usepackage{float}
\usepackage{xspace}
\usepackage[hyphens]{url}

\usepackage[tt=false]{libertine}
\usepackage[libertine,vvarbb]{newtxmath}

\usepackage{bm}

\usepackage{hyperref}
\usepackage[svgnames]{xcolor}
\hypersetup{colorlinks={true},urlcolor={blue},linkcolor={DarkBlue},citecolor=[named]{DarkGreen}}
\usepackage[authoryear,square]{natbib}

\usepackage{microtype}
\usepackage[capitalise,nameinlink,noabbrev]{cleveref}
\usepackage{enumitem}

\usepackage{subcaption}
\usepackage{tikz}
\usetikzlibrary{math}
\usepackage{doi}

\renewcommand{\paragraph}[1]{\medskip\noindent\textbf{#1.\;}}

\def \R{\mathbb R}

\def \Q{\mathbb Q}
\def \dcal{\mathcal{D}}

\def \pcal{\mathcal{P}}

\newcommand{\sset}[1]{\left\{ #1\right\}}
\newcommand{\ssets}[1]{\{ #1\}}
\newcommand{\fwh}[1]{\; \left| \; #1 \right.}

\newcommand{\card}[1]{\left| #1 \right|}
\newcommand{\cards}[1]{| #1 |}

\DeclareMathOperator*{\expectation}{\mathbb E}
\newcommand{\expect}[2][]{\expectation_{#1}\nolimits\left[#2\right]}

\DeclareMathOperator*{\probability}{\mathrm{Pr}}
\newcommand{\prob}[1]{\probability\left[#1\right]}

\DeclareMathOperator*{\argmax}{argmax}

\DeclareMathOperator{\trunc}{\textup{T}}
\newcommand{\val}{\vec{\mathrm{v}}}
\newcommand{\gceps}{\tilde{\varepsilon}}
\newcommand{\nul}{\textup{null}}

\renewcommand\vec{\bm}

\newcommand{\ppad}{\textup{PPAD}\xspace}
\newcommand{\pls}{\textup{PLS}\xspace}
\newcommand{\fixp}{\textup{FIXP}\xspace}
\newcommand{\gcircuit}{\textup{\textsc{Gcircuit}}\xspace}
\newcommand{\exactgcircuit}{\textup{\textsc{exact-Gcircuit}}\xspace}
\newcommand{\epsfpa}{\textup{\textsc{$\varepsilon$-BNE-FPA}}\xspace}
\newcommand{\exactfpa}{\textup{\textsc{exact-BNE-FPA}}\xspace}
\newcommand{\sz}{\textup{size}}

\newcommand{\nn}{\ensuremath{\mathfrak{m}}} 
\newcommand{\mm}{\ensuremath{n}} 

\theoremstyle{definition}
\newtheorem{definition}{Definition}

\theoremstyle{plain}
\newtheorem{theorem}{Theorem}[section]
\newtheorem{lemma}[theorem]{Lemma}
\newtheorem{corollary}[theorem]{Corollary}
\newtheorem{proposition}[theorem]{Proposition}

\newtheorem{claim}{Claim}
\Crefname{claim}{Claim}{Claims}
\newtheorem{inftheorem}{Informal Theorem}
\newtheorem*{opproblem}{Open Problem}

\theoremstyle{definition}

\newtheorem{example}{Example}

\title{On the Complexity of Equilibrium Computation\\ in First-Price
Auctions\thanks{A preliminary version of this paper appeared in EC'21~\citep{fghlp2021_ec}.}}

\author{
\begin{tabular}{c c}
& \\ \textbf{Aris Filos-Ratsikas} & \textbf{Yiannis Giannakopoulos\thanks{Part of this work was done while the author was a member of the Operations Research group at Technical University of Munich, School of Management, supported by the Alexander von Humboldt Foundation with funds from the German Federal Ministry of Education and Research (BMBF).}}\\
\small{University of Edinburgh, United Kingdom} & \small{FAU Erlangen-Nürnberg, Germany} \\
\href{mailto:Aris.Filos-Ratsikas@ed.ac.uk}{\small{\texttt{Aris.Filos-Ratsikas@ed.ac.uk}}} & \href{mailto:yiannis.giannakopoulos@fau.de}{\small{\texttt{yiannis.giannakopoulos@fau.de}}}\\
& \\
\textbf{Alexandros Hollender\thanks{Supported by an EPSRC doctoral studentship (Reference 1892947)}} & \textbf{Philip Lazos\thanks{Partially supported by the ERC Advanced Grant 788893 AMDROMA ``Algorithmic and Mechanism Design Research in Online Markets'' and MIUR PRIN project ALGADIMAR ``Algorithms, Games, and Digital Markets''.}}\\
\small{University of Oxford, United Kingdom} & \small{IOHK} \\
\href{mailto:alexandros.hollender@cs.ox.ac.uk}{\small{\texttt{alexandros.hollender@cs.ox.ac.uk}}} & \href{mailto:philip.lazos@iohk.io}{\small{\texttt{philip.lazos@iohk.io}}}\\
& \\
\multicolumn{2}{c}{\textbf{Diogo Po\c{c}as}\thanks{Supported by FCT via LASIGE Research Unit, ref.\ UIDB/00408/2020.}}\\
\multicolumn{2}{c}{\small{LASIGE, Faculdade de Ciências, Universidade de Lisboa, Portugal}}\\
\multicolumn{2}{c}{\href{mailto:dmpocas@fc.ul.pt}{\small{\texttt{dmpocas@fc.ul.pt}}}}\\
& \\
\end{tabular}
}

\date{}

\begin{document}
\maketitle

\begin{abstract}
We consider the problem of computing a (pure) Bayes-Nash equilibrium in the first-price auction with continuous value distributions and discrete bidding space. We prove that when bidders have independent \emph{subjective} prior beliefs about the value distributions of the other bidders, computing an $\varepsilon$-equilibrium of the auction is PPAD-complete, and computing an \emph{exact} equilibrium is FIXP-complete. We also provide an efficient algorithm for solving a special case of the problem, for a fixed number of bidders and available bids.
\end{abstract}

\section{Introduction}

Auctions are prime examples of economic environments in which the element of
strategic behavior is prevalent. The associated theory can be traced back to as
early as the 1960s and the seminal work of \citet{vickrey1961counterspeculation}.
Over the years, auction theory and mechanism design have produced some of the most
celebrated results in economics, as can be evidenced, e.g., by the relevant 1996,
2007 and 2020 Nobel
Prizes.\footnote{For the official Nobel Prize announcements see \href{https://www.nobelprize.org/prizes/economic-sciences/1996/summary/}{here}, \href{https://www.nobelprize.org/prizes/economic-sciences/2007/summary/}{here} and \href{https://www.nobelprize.org/prizes/economic-sciences/2020/summary/}{here}.}
Among the plethora of auction formats that this rich literature has proposed, some stand out, such as the second-price auction of \citet{vickrey1961counterspeculation} or the revenue-maximizing auction of \citet{myerson1981optimal}.

Arguably, though, the most fundamental auction format is that of the
\emph{first-price auction}, in which the highest bidder wins and is charged an
amount equal to her bid. Compared to its counterparts mentioned above, the
first-price auction does not enjoy the same desirable incentive properties:
participants may have an incentive to misreport their true bids. At the same time,
however, the first-price auction is very natural and simple to describe, implement
and participate in, making it very suitable for a range of important applications.
As a matter of fact, several online ad exchanges, including Google Ad Manager, have
adopted this auction format for selling their ads, which has been coined ``the
first-price movement'' (see, e.g., \citep{adexchange,paes2020competitive}).

There has been a large body of work studying incentives and bidding behavior in
first-price auctions, dating back to the original paper of
\citet{vickrey1961counterspeculation}. In particular, the literature has studied the
equilibria of the auction in an incomplete information setting where the bidders have
only probabilistic prior beliefs (or simply \emph{priors}) about the values of other
bidders, via the lens of Bayesian game theory~\citep{harsanyi1967games} (see also
\citep{myerson2013game,hartline2012bayesian}). Several different scenarios of
interest have been analyzed; see, e.g.,
\citep{griesmer1967toward,riley1981optimal,plum1992characterization,marshall1994numerical,lebrun1996existence,lebrun1999first,maskin2000equilibrium,lizzeri2000uniqueness,Athey2001,reny2004existence,chawla2013auctions,bergemann2017first}.
It is no exaggeration to say that understanding the Bayes-Nash equilibria of the
first-price auction has historically been one of the most important questions of
auction theory.

The aforementioned literature has been primarily concerned with identifying conditions under which (pure Bayes-Nash) equilibria are guaranteed to exist. Among those, the seminal paper of~\citet{Athey2001} has been pivotal in establishing the existence of equilibria for fairly general settings with continuous priors. A natural follow-up question posed explicitly by \citet{Athey2001}, which was also very much present in earlier works, is whether these equilibria can also be ``found''; in the context of the related literature, this is usually interpreted as coming up with closed-form solutions that describe them.

One of the most significant contributions of computer science to the field of game
theory is to formalize and systematically study this notion of ``finding'' or
``computing'' equilibria in games. Roughly speaking, an equilibrium can be
efficiently computed if it can be found using a limited number of standard
operations that can be performed by a computer, where ``limited'' here typically
means a number which is a polynomial function of the size of the input parameters.\footnote{We
remark that contrary to earlier works in economics, Athey's interpretation of
``finding'' an equilibrium was very much of a computational nature.} In perhaps the
most important result in computational game theory, \citet{daskalakis2009complexity}
proved that in all likelihood, Nash equilibria of general games cannot always be
computed efficiently. In particular, they proved that the problem of computing a
Nash equilibrium is complete for the class PPAD \citep{papadimitriou1994complexity},
which is widely believed to include problems that are computationally hard to
solve.

In this paper, we study the complexity of computing an equilibrium of the first-price auction, in settings with continuous priors and discrete bids. We offer the following main result.

\begin{inftheorem}\label{infthm:main-result}
Computing a (pure, Bayes-Nash) equilibrium of a first-price auction with continuous subjective priors and discrete bids is \ppad-complete.
\end{inftheorem}

This result can be interpreted intuitively as justification of why research in
economics has only had limited success in providing closed forms or characterizations
for the equilibria of the first-price auction. In addition, we consider it to be 
a quite valuable addition to the literature of
total search problems \citep{megiddo1991total}, as it concerns the computation
of equilibria of one of the most fundamental
games in auction theory.

\subsection{Discussion and Further Results}

Below, we provide a more in-depth discussion of our main result and its assumptions, as well as some other related results that we obtain along the way. 

\paragraph{Continuous Priors, Discrete Bids}
\cref{infthm:main-result} applies to the case where the bidders' beliefs about the values of other bidders are continuous distributions, whereas the bidding space is a discrete set. The former assumption is standard in auction theory (see, e.g., \citep[Sec.~3.11]{myerson2013game} or \citep{krishna2009auction}).
From a technical standpoint, this also guarantees the existence of
equilibria~\citep{Athey2001}.\footnote{It is important to note here that in some versions of the problem, even \emph{mixed} Bayes-Nash equilibria are
not guaranteed to exist; see, e.g., \citep{lebrun1996existence}.} The assumption of
the discrete bidding space is clearly motivated by any real-world scenario, in which
the bids will be increments of some minimum monetary amount, e.g., 1 dollar or 1
cent, depending on the application. This setting has in fact been studied in several
works for first-price auctions in particular (see, e.g.,
\citep{chwe1989discrete,Athey2001,escamocher2009existence,cai2010note,rasooly2020importance}).

\paragraph{Subjective Priors} In \cref{infthm:main-result} we assume that the
priors are subjective, meaning that two different bidders might have different
beliefs about the values of some other bidder. In the auction theory literature, it
is often assumed that a ``universal'' prior exists, which is common knowledge among
all players; this is known as the \emph{independent private values} model. Indeed,
such common priors are quite convenient in settings where there is an aggregate
objective that needs to be optimized in expectation (e.g., the social welfare or the
seller's revenue), since they can be used by the designer to tune the parameters of
the auction in a way that works best for the optimization goal at hand; this is the
case, e.g., for Myerson's revenue-maximizing auction~\citep{myerson1981optimal}.

From our perspective however, where the goal is to study the players' incentives and
compute an equilibrium, we believe it is natural to make the more general assumption
that priors are still independent, but subjective: this is enough for the bidders to come up with their best
responses. As a matter of fact, Harsanyi's original paper~\citep{harsanyi1967games},
as well as classic textbooks in economics (e.g., \citep{myerson2013game,Jehle2001a})
introduce Bayesian games directly in the context of subjective
beliefs.\footnote{These works also usually provide discussions on ``consistency''
conditions, e.g., see \citep{harsanyi1967games} and
\citep[Sec.~2.8]{myerson2013game}. See also a related discussion in \cref{sec:conclusion} of our work.} Similar notions of subjective priors and
``subjective equilibria'' have also been studied rather extensively for general
Bayesian games in economics
\citep{hahn1973notion,fudenberg1986limit,battigalli1988conjectural,battigalli1992learning,kalai1993rational,kalai1995subjective,rubinstein1994rationalizable}
and computer science \citep{witkowski2012peer,frongillo2016geometric}.

The subjective priors assumption is necessary for our PPAD-hardness result, but we
would of course be very interested in settling the complexity for the case of common
priors as well. In fact, as we explain in \cref{sec:conclusion}, we consider this to
be one of the most important open problems in computational game theory. Thus,
besides being of standalone interest, one can also see our result for subjective
priors as an important first step in the quest of answering this question. We remark
that our PPAD-membership result obviously applies to common priors, as this is just
a special case of subjective beliefs.

\paragraph{Approximate Equilibria} While \cref{infthm:main-result}
states the PPAD-completeness of computing an equilibrium of the first-price auction,
the formal statement is in fact about $\varepsilon$-equilibria, i.e., stable states
in which bidders do not wish to unilaterally deviate unless they are better off by
some small positive quantity $\varepsilon$. As we explain in \cref{sec:prelims},
this is very much necessary: there are examples where the equilibrium is
\emph{irrational}, and therefore cannot be computed exactly in many standard models
of computation. As a matter of fact, this is a common theme in most papers in
equilibrium computation; see, e.g.,
\citep{daskalakis2009complexity,chen2009settling} or the survey of
\citet{goldberg2011survey} for a related discussion.

Of course, the focus on $\varepsilon$-equilibria is only relevant for the membership result in PPAD; the computational hardness result for approximate equilibria is clearly stronger. In fact, we show that under some standard assumptions (see \cref{sec:prelims}), the problem is PPAD-hard even when $\varepsilon$ is allowed to be a (sufficiently small) constant, independent of the input parameters. This is the strongest type of PPAD-hardness one could hope for. 
For the computation of \emph{exact} equilibria, \citet{etessami2010complexity}
defined the computational class FIXP. At a high level, this class contains problems
that can be stated as computations of (possibly irrational) fixed points of functions defined by means of algebraic circuits
(see~\citep{yannakakis2009equilibria}). We complement our main result about
$\varepsilon$-equilibria with the following analogous result on exact ones:

\begin{inftheorem}\label{infthm:FIXP}
Computing an \emph{exact} (pure, Bayes-Nash) equilibrium of a first-price auction with continuous subjective priors and discrete bids is \fixp-complete.
\end{inftheorem}
One way to interpret a FIXP-completeness result in the standard computational (Turing) model is in terms of \emph{strong} vs \emph{weak} approximations. A weak approximation is an $\varepsilon$-equilibrium as defined above and is captured by our PPAD-completeness result. A strong approximation is a set of strategies represented by rational numbers, which are ``$\varepsilon$-close'' to an exact equilibrium (in terms of the max norm), and is captured by our FIXP-completeness result. We remark that this is completely analogous to the computation of Nash equilibria in general games, see \citep{etessami2010complexity,garg2016dichotomies} for a more in-depth discussion.

\paragraph{The Meaning of \ppad-completeness} As we mentioned earlier, a
PPAD-hardness result is interpreted as an indication that the problem
cannot be solved in polynomial time. In particular, it is as hard as finding Nash
equilibria in general games
\citep{daskalakis2009complexity,chen2009settling,mehta2018constant,rubinstein2018inapproximability}, market equilibria in Arrow-Debreu markets
\citep{vazirani2011market,chen2013complexity} or solutions to fixed point theorems
\citep{papadimitriou1994complexity,goldberg2019hairy}. Additionally, PPAD has been
shown to be hard under various cryptographic assumptions (e.g., see
\citep{bitansky2015cryptographic,garg2016revisiting,choudhuri2019finding,rosen2017can}),
meaning that solving a PPAD-hard problem would ``break'' those assumptions as well.
On the other hand, an ``in PPAD'' result can be interpreted as the existence of an
(inefficient) algorithm that uses a path-following argument to reach a solution.

\paragraph{An Efficient Algorithm} 
Besides our main PPAD- and FIXP-completeness results, we identify a special case of the problem which can be solved efficiently, namely when the number of bidders and the size of the bidding space are constant, and the value distributions are ``sufficiently smooth'', in the sense that they are given by piecewise polynomial functions. To this end, we have the following theorem.

\begin{inftheorem}\label{infthm:polytime}
A (pure, Bayes-Nash) equilibrium of the first-price auction can be computed in polynomial time when there is a constant number of bidders, a constant-size bidding space, and continuous (subjective) priors which are piecewise polynomial functions.
\end{inftheorem}

\cref{infthm:polytime} complements our PPAD- and FIXP-hardness results rather tightly, as our reductions use a constant bidding space and very simple, piecewise constant distributions, but a large number of bidders. 

\subsection{Related Work}

As we mentioned earlier, there is a significant amount of work in economic theory on
the equilibria of the first-price auction
\citep{griesmer1967toward,riley1981optimal,plum1992characterization,marshall1994numerical,lebrun1996existence,lebrun1999first,maskin2000equilibrium,lizzeri2000uniqueness,Athey2001,reny2004existence,bergemann2017first,Cheng2006}.
Among those, the most relevant work to us is that of~\citet{Athey2001}, who
established the existence of pure Bayes-Nash equilibria in games with discontinuous
payoffs which satisfy the \emph{single crossing property} of
\citet{milgrom1994monotone}, of which the first-price auction is a special case.
Athey's proof applies to both discrete and continuous bidding spaces, and in fact
the latter is established through the former, via a limit argument similar in spirit
to \citep{lebrun1996existence,maskin2000equilibrium}.

To the best of our knowledge, there are only a few prior works on the computational
complexity of equilibria in first-price auctions. \citet{escamocher2009existence}
study the problem of computing equilibria when \emph{both} the priors and the
bidding space are discrete. In that case, it is not hard to construct
counter-examples that show that pure equilibria may not exist, and therefore they are
concerned with the question of \emph{deciding} their existence. Their results do not
provide a conclusive answer (i.e., neither NP-hardness nor polynomial-time
solvability is proven), except for the very special case of two bidders with
bi-valued distributions. \citet{wang2020bayesian} very recently studied the equilibrium
computation problem in settings with \emph{discrete priors} and \emph{continuous bids}
(in a sense, the opposite of what we do here), and under the \emph{Vickrey
tie-breaking rule} for deciding the winner of the auction in case of a tie.
According to this rule, ties are resolved by running an auxiliary second-price
(Vickrey) auction among the potential winners of the first-price auction;
effectively this allocates the item to the bidder with highest true valuation. This
tie-breaking rule was introduced by \citet{maskin2000equilibrium} primarily as a
technical tool in proving their existence results for the \emph{uniform tie-breaking rule},
where ties are broken uniformly at random among the bidders with the highest bid. Our
results are proven for the uniform tie-breaking rule, which is the standard rule in
the literature of the
problem~\citep{lebrun1996existence,maskin2000equilibrium,Athey2001,krishna2009auction}.

Finally, we remark that while we consider an equilibrium computation setting, our
results are markedly different from other works on such problems, e.g.,
\citep{daskalakis2009complexity}. This is because our paper concerns a much more specific
and structured game, and crucially, a game which is \emph{Bayesian}, which is not
the case for most prior work. Conceptually closer to our work is the paper by \citet{cai2014simultaneous} who study the complexity of Bayesian \emph{combinatorial} auctions, a more complicated auction format which typically involves multiple items for sale and more complex agent valuations over subsets of items.
The complexity of \emph{general} Bayesian games
(beyond auctions) has been studied in the literature, primarily resulting in
NP-hardness results for several cases of interest, e.g., see
\citep{Gottlob:2007aa,Conitzer:2008aa}.

\section{Model and Notation}\label{sec:prelims}

In a (Bayesian) \emph{first-price auction (FPA)}, there is a set
$N=\{1,2,\ldots,n\}$ of \emph{bidders} (or \emph{players}) and one item for sale. Each
player $i$ submits a \emph{bid} $b_i \in B$, where the \emph{bidding space} $B
\subseteq [0,1]$ is a finite set. 
We will also make the standard assumption (often referred to as the ``null bid'' in the literature) that $0\in B$, which can be interpreted as
the option of the bidders to not participate in the auction (see, e.g., \citep{maskin2000equilibrium,Athey2001}).

The item is allocated to the player with the highest bid, who is charged a payment equal to her bid. If there are multiple
players submitting the same highest bid, the winner is determined based on the
\emph{uniform tie-breaking} rule.
Formally, for a \emph{bid profile} $\vec{b}=(b_1,\ldots,b_n)$, the \emph{ex-post utility} of player $i$ with true value $v_i$ is given by
\begin{equation}
\label{eq:ex_post_utilities}
\tilde{u}_i(\vec{b};v_i) \equiv 
\begin{cases}
\frac{1}{\cards{W(\vec{b})}}(v_i-b_i), & \text{if}\;\; i\in W(\vec{b}), \\
0, & \text{otherwise}, 
\end{cases}
\qquad\text{where}\;\; W(\vec{b})=\argmax_{j\in N} b_j
\end{equation}

For each pair of players $i, j\in N$, $i\neq j$, there is a continuous value distribution
$F_{i,j}$ over $[0,1]$; we call this
distribution the \emph{prior} of bidder $i$ over the values of bidder $j$. The
\emph{subjective belief} of player $i$ for the values
$\vec{v}_{-i}=(v_1,\dots,v_{i-1},v_{i+1},\dots, v_n)$ of the other bidders is then
given by the product distribution $\vec{F}_{-i}\equiv\times_{j\neq i} F_{i,j}$. In
other words, from the perspective of bidder $i$, the values $v_j$ for $j\neq i$ are
drawn \emph{independently} from distributions $F_{i,j}$.
Notice that the special case where $F_{i,j} = F_{i',j}$ for all $j \in N$ and $i, i'
\in N\setminus\ssets{j}$ corresponds to the classic \emph{independent private
values} model of auction theory, where the value of each bidder is drawn
(independently of the others) from a single distribution. More formally, simplifying
the notation by using $F_j$ instead of $F_{i,j}$, $\vec{v}$ is drawn from the
\emph{common prior} distribution $\vec{F}=\times_{j\in N} F_j$.
Obviously, while our hardness results rely on the fact that priors are subjective,
all of our positive results trivially extend to the case of common priors as well.

The FPA described above naturally induces a game in which each bidder $i$ selects
her bid based on her own (true) value $v_i$, and her beliefs $\vec{F}_{-i}$. A \emph{strategy} of bidder $i$ is a function
$\beta_i: [0,1] \rightarrow B$ mapping values to bids. 
Given a strategy profile $\vec{\beta}_{-i}$ of the other players, the (interim) \emph{utility} of player $i$ with true value $v_i$ when bidding $b\in B$ is
\[
u_i(b,\vec{\beta}_{-i};v_i) \equiv\expect[\vec v_{-i}\sim
\vec{F}_{-i}]{\tilde{u}_i(b,\vec{\beta}_{-i}(\vec{v}_{-i});v_i)},
\] 
where $\vec{\beta}_{-i}(\vec{v}_{-i})$ is a shorthand for
$\left(\beta_1(v_1),\dots,\beta_{i-1}(v_{i-1}),\beta_{i+1}(v_{i+1}),\dotsm,\beta_n(v_n)\right)$.
Intuitively, the player calculates her (expected) utility by drawing a value $v_j$
for each bidder $j \neq i$ from her corresponding subjective prior distribution
$F_{i,j}$, and then using the strategy ``rules'' $\vec{\beta}_{-i}$ of the others to
map their values to actual bids in $B$.

We are interested in ``stable'' states of the FPA, i.e., strategy profiles from
which no bidder would like to unilaterally deviate to a different strategy. Formally,
we have the following definition.

\begin{definition}[$\varepsilon$-Bayes-Nash equilibrium of the FPA]\label{def:bayes-nash-equilibrium}
Let $\varepsilon \geq 0$.
A strategy profile $\vec{\beta}=(\beta_1, \ldots, \beta_n)$ is a (pure, interim) $\varepsilon$-Bayes-Nash equilibrium ($\varepsilon$-BNE) of the FPA if for any bidder $i \in N$ and any value $v_i \in [0,1]$, 
\[
u_i(\beta_i(v_i),\vec{\beta}_{-i};v_i) \geq u_i(b,\vec{\beta}_{-i};v_i) - \varepsilon \qquad \text{for all}\;\; b\in B.
\]
\end{definition}
Given a fixed strategy profile $\vec{\beta}_{-i}$ of the other bidders, we will denote the set of $\varepsilon$-\emph{best responses} of player $i$ by 
\[
BR_{i}^{\varepsilon}(\vec{\beta}_{-i}) = \left\{\beta_i\fwh{ u_i(\beta_i(v_i),\vec{\beta}_{-i};v_i) \geq \max_{b\in B} u_i(b,\vec{\beta}_{-i};v_i) - \varepsilon\quad\text{for all}\;\; v_i \in [0,1]} \right\}
\]
Using this, the condition in \cref{def:bayes-nash-equilibrium} can be equivalently
written as $\beta_i\in BR_{i}^{\varepsilon}(\vec{\beta}_{-i})$ for all players $i$. For the special case of $\varepsilon=0$, i.e.\ \emph{exact} best-responses, we will drop the $\varepsilon$ superscript.

Notice that, in \cref{def:bayes-nash-equilibrium} we define a relaxed equilibrium
concept, in which the bidder does not want to change to a different strategy unless
it increases her utility by an additive factor larger than $\varepsilon$; obviously,
when $\varepsilon=0$ we recover the standard definition of the (exact) Bayes-Nash
equilibrium.

\paragraph{No Overbidding}
As part of our model, we will make the assumption that bidders will never submit a
bid $b_i$ which is higher than their valuation $v_i$. This is a standard
assumption in the 
literature of the first-price auction \citep{maskin2000equilibrium,maskin2003uniqueness,lebrun2006uniqueness,escamocher2009existence,wang2020bayesian}
and auctions in general
\citep{caragiannis2015bounding,lucier2010price,bhawalkar2011welfare,feldman2013simultaneous,christodoulou2008bayesian,leme2010pure}.
The rationale behind it stems from the fact that, given the format of the utilities
in the FPA (see~\eqref{eq:ex_post_utilities}), it is arguably unreasonable to overbid,
as bidding $0$ will \emph{always} result in at least the same utility. In
game-theoretic terms, the overbidding strategy is \emph{weakly dominated} by bidding
$0$, which can be interpreted as abstaining from the auction. These strategies are
typically excluded from consideration to rule out unnatural equilibria (see
\citep{feldman2013simultaneous} for a discussion).
\medskip

\noindent We are now ready to formally define our computational problem of finding an equilibrium of the FPA: \\

\noindent\fbox{%
\colorbox{gray!10!white}{
    \parbox{0.965\textwidth}{%
\noindent\textbf{\underline{$\bm{\varepsilon}$-\textsc{Bayes-Nash Equilibrium in the First-Price Auction ($\bm{\varepsilon}$-BNE-FPA)}}} 

\smallskip

\noindent \textsc{Input:} 
\begin{itemize}
\item[-] a set of bidders $N=\{1,2,\ldots,n\}$;
    \item[-] a finite bidding space $B\subseteq [0,1]$;
    \item[-] for each pair of bidders $i, j\in N$, a continuous value distribution $F_{i,j}$ over $[0,1]$.
\end{itemize}

\smallskip

\noindent \textsc{Output:} 
An  $\varepsilon$-Bayes-Nash equilibrium $\vec{\beta}=(\beta_1,\ldots,\beta_n)$.
}}} \medskip

We will use the term \exactfpa instead of $0$-BNE-FPA to denote the computational problem of finding an exact Bayes-Nash equilibrium of the auction. Some remarks related to the definition above are in order.

\paragraph{The Input Model for the Distributions} We have intentionally vaguely
stated that the distributions $F_{i,j}$ should be provided as input to the problem,
but we have not specified exactly how. Our positive results hold even when the
functions $F_{i,j}$ are fairly general, and can be concisely and efficiently
represented in a form that is appropriate for computation. In the interest of
clarity, we omit the technical details here, and we refer the reader to
\cref{app:inputs} where we provide all the details of the input model. For the
negative results on the other hand, we use fairly simple distributions $F_{i,j}$ -- this only makes our results stronger. In particular, we use \emph{piecewise-constant} density functions, which can be represented by the endpoints and the value for
each interval.

\paragraph{Explicit Bidding Space} We assume that the bidding space is
explicitly given as part of the input. This assumption is required in \cref{sec:bestresponses} in order to show that we can compute best-responses efficiently. Even in the mildest of settings where the bidding space is given implicitly, computing best-responses turns out to be computationally and
information-theoretically hard. We show this in \cref{app:bidding-space}.

\paragraph{Equilibrium Representation} Besides the representation of the input, the output of our computational problem, i.e., the equilibrium of the FPA, should also be represented in some concise and efficient way. 
Following the standard literature of the problem, we will consider equilibria for which the strategy $\beta_i(v_i)$ of each bidder is a non-decreasing function of her value $v_i$ (e.g., see \citep{Athey2001,maskin2000equilibrium,reny2004existence} and \citep[Appendix G]{krishna2009auction}) for which the existence of an equilibrium is always guaranteed \citep{Athey2001}. These equilibria are in a sense the only ``natural'' ones, as, similar to the case of overbidding (see earlier discussion), any bidder's strategy is weakly dominated by a non-decreasing strategy. 

Based on this, there is a straightforward and computationally efficient way of representing 
the best response of each player, as a step function with a
finite set of ``jump points'', corresponding to the values at which the bidder
``jumps'' from one bid to the next \citep{Athey2001}.
Formally, we define \begin{equation}\label{eq:jump_points}
\alpha_i(b)=\sup\sset{v\fwh{\beta_i(v)\leq b}}.
\end{equation}
Intuitively,
$\alpha_i(b)$ is the largest value for which player $i$ would bid $b$ or lower. With
a slight abuse of notation, we can write $\alpha_i=\beta_i^{-1}$, that is,
$\alpha_i$ can be interpreted as an \emph{inverse bidding} strategy. In that way, we
can also rework $\beta_i$ from $\alpha_i$, as $\beta_i(v)=b$, where
$v\in(\alpha_i(b^-),\alpha_i(b)]$ for any $b \in B$. Here we let $b^-$ denote the previous bid, i.e., the largest $b' \in B$ with $b' < b$. Finally, to be able to handle the corner cases in a unified way, we set $a_i(b^-)=0$ when $b=0$. Notice also, that $\alpha_i(b)=1$ when $b=\max B$.

In particular, this implies that bidding strategies are left-continuous (which is without loss of generality given our value distributions), as shown in \cref{fig:inversebidding}.

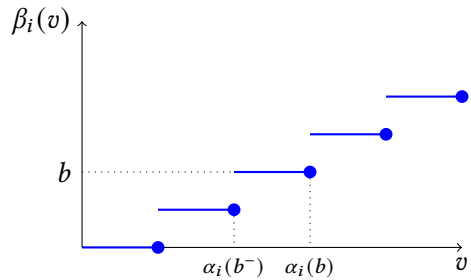
\begin{figure}[ht]\begin{center}
\begin{tikzpicture}
\draw[->] (0,0) -- (5,0);
\draw[->] (0,0) -- (0,3);
\node[below] at (5,0) {$v$};
\node[below] at (3,0) {\scriptsize $\alpha_i(b)$};
\node[below] at (2,0) {\scriptsize $\alpha_i(b^-)$};
\node[left] at (0,3) {$\beta_i(v)$};
\node[left] at (0,1.0) {$b$};
\coordinate (b1l) at (0,0.0) {};
\coordinate (b1r) at (1,0.0) {};
\coordinate (b2l) at (1,0.5) {};
\coordinate (b2r) at (2,0.5) {};
\coordinate (b3l) at (2,1.0) {};
\coordinate (b3r) at (3,1.0) {};
\coordinate (b4l) at (3,1.5) {};
\coordinate (b4r) at (4,1.5) {};
\coordinate (b5l) at (4,2.0) {};
\coordinate (b5r) at (5,2.0) {};
\draw[blue,thick] (b1l) -- (b1r);
\draw[blue,thick] (b2l) -- (b2r);
\draw[blue,thick] (b3l) -- (b3r);
\draw[blue,thick] (b4l) -- (b4r);
\draw[blue,thick] (b5l) -- (b5r);
\draw[dotted] (0,1.0) -- (b3l);
\draw[dotted] (b3r) -- (3,0.0);
\draw[dotted] (b2r) -- (2,0.0);
\draw[blue,thick,fill=blue] (b1r) circle (2pt);
\draw[blue,thick,fill=blue] (b2r) circle (2pt);
\draw[blue,thick,fill=blue] (b3r) circle (2pt);
\draw[blue,thick,fill=blue] (b4r) circle (2pt);
\draw[blue,thick,fill=blue] (b5r) circle (2pt);
\end{tikzpicture}
\end{center}\caption{A monotone bidding strategy $\beta_i(\cdot)$ can be succinctly represented by its jump points, $\alpha_i(b)$ for $b\in B$.\label{fig:inversebidding}}\end{figure}

\paragraph{Irrational Equilibria}
As discussed in our
introduction, for our PPAD-completeness result, 
we will be looking for an $\varepsilon$-approximate equilibrium,
rather than an exact one. Of course, this only makes our hardness results even
stronger; but besides that, it is actually very much necessary for our membership
result in PPAD as well. In particular, as demonstrated by the example below, the FPA
may have \emph{only irrational} equilibria, even when all input parameters are
rational numbers.

\begin{example}\label{ex:irrational}
Consider a FPA with $n=3$ bidders and common priors, whose values are independently
and identically distributed according to the uniform distribution on $[0,1]$; that
is, $F_i(x)=x$ for $i=1,2,3$. Let the bidding space be $B=\{0,1/2\}$. Clearly, this
auction can be represented with piecewise-constant density functions (with a single
piece) and with a finite number of rational quantities. It can be verified that the
auction has a unique equilibrium, where a bidder bids $0$ iff her valuation is below $\frac{-1+\sqrt{5}}{2}\approx 0.618$; therefore, the unique equilibrium is irrational. We
provide the detailed derivation in \cref{app:irrational}.
\end{example}

The appropriate setting for studying the computation of \emph{exact} equilibria is
the class FIXP of \citet{etessami2010complexity}. In \cref{sec:inFIXP,sec:hardness} we show
that the problem of exact equilibrium computation of the FPA is FIXP-complete.

\paragraph{Further Notation} We conclude the section with the following terminology which will be useful in multiple sections of our paper. For $t_1 < t_2$, we will let $\trunc_{[t_1,t_2]}$ denote the \emph{truncation} of a value $x$ to $[t_1,t_2]$, i.e., $\trunc_{[t_1,t_2]}(x) = \max\{t_1,\min\{t_2,x\}\}$. Furthermore, for $k \in \mathbb{N}$ we sometimes use $[k]$ to denote $\{1,2,\dots,k\}$.

\subsection{Outline}
In \cref{sec:bestresponses} we provide a useful characterization of BNE and then show how to compute the best responses in polynomial time. In \cref{sec:membership}, first we provide a new existence proof via Brouwer's fixed point theorem, and then proceed to prove the membership of the equilibrium computation problems in PPAD and FIXP. In \cref{sec:hardness} we show the computational hardness for these classes. In \cref{sec:positive} we present an efficient algorithm for a natural special case. We conclude with some interesting future directions in \cref{sec:conclusion}.

\section{Equilibrium Characterization and Best Response Computation}\label{sec:bestresponses}

In this section we begin by presenting a useful characterization of $\varepsilon$-BNE that is crucial for many parts of the paper. Then, we show how best-responses of bidders can be checked and computed in polynomial time. We remark that the reductions that we will construct in \cref{sec:membership} to show the \ppad-membership and the \fixp-membership of the problem do not technically require the computation of the whole best-response function, but rather only the probabilities of winning the item given the bidder's bid and the bidding strategies of the other bidders. However, the best-response computation is interesting in its own right, and that is why we present this here.

\paragraph{Characterization}
The following lemma essentially states that an $\varepsilon$-BNE is characterized by the behavior of the bidding function at the jump points. Recall that for any bid $b$, we let $b^-$ denote the previous bid, and we use the convention $\alpha_i(b^-)=0$ when $b=0$. Notice that $\alpha_i(b)=1$ when $b = \max B$. Furthermore, when $\alpha_i(b^-) = \alpha_i(b)$, the corresponding strategy $\beta_i$ does not use bid $b$, but instead jumps from the previous bid directly to the following one. We will call all bids where this does not happen, i.e, we have $\alpha_i(b^-) < \alpha_i(b)$, \emph{non-degenerate} for bidder $i$.

\begin{lemma}[Characterization of $\varepsilon$-BNE]\label{lem:charepsilonbne}
	Fix an $\varepsilon \geq 0$. A strategy profile $\vec{\beta}$ is an $\varepsilon$-BNE of the FPA, if and only if, for every bidder $i$ and every non-degenerate bid $b$,
	\begin{equation}\label{eq:charepsilonbne1}
	u_i(b,\vec{\beta}_{-i};\alpha_i(b^-)) \geq
	u_i(b',\vec{\beta}_{-i};\alpha_i(b^-)) -\varepsilon
	\qquad\text{for all}\;\; b'<b
	\end{equation}
	and
	\begin{equation}\label{eq:charepsilonbne2}
	u_i(b,\vec{\beta}_{-i};\alpha_i(b)) \geq
	u_i(b',\vec{\beta}_{-i};\alpha_i(b)) -\varepsilon
	\qquad\text{for all}\;\; b'>b.
	\end{equation}
\end{lemma}

\paragraph{The $\bm{H}$-functions}
Before proving this characterization, we introduce some useful notation. We use the term $H_i(b,\vec{\beta}_{-i})$ to denote the (perceived) probability that bidder $i$ wins the item with bid $b$, when the other bidders use bids according to the bidding strategy $\vec{\beta}_{-i}$, i.e., 
\begin{equation*}
H_i(b,\vec{\beta}_{-i}) = \prob{\text{bidder } i \text{ wins}|b, \vec{\beta}_{-i}}
\end{equation*}
The utility can easily be expressed in terms of this function, namely $u_i(b,\vec{\beta}_{-i};v_i) = (v_i-b) \cdot H_i(b,\vec{\beta}_{-i})$.

\begin{proof}[Proof of \cref{lem:charepsilonbne}]
	($\Rightarrow)$: Fix a bidder $i$ and a bid $b$ with $\alpha_i(b^-)<\alpha_i(b)$. Since bidder $i$ bids $b$ inside the non-empty interval $(\alpha_i(b^-),\alpha_i(b)]$, and $\vec{\beta}$ is an $\varepsilon$-BNE, we get that $u_i(b,\vec{\beta}_{-i};v_i)\geq u_i(b',\vec{\beta}_{-i};v_i)-\varepsilon$ for every $v_i \in(\alpha_i(b^-),\alpha_i(b)]$ and $b'\neq b$. Since the utilities are continuous functions on $v_i$, the inequalities must also hold at the interval endpoints.
	
	($\Leftarrow$): Suppose (\ref{eq:charepsilonbne1},~\ref{eq:charepsilonbne2}) hold. Take any bidder $i$ and any valuation $v_i$, and let $(\alpha_i(b^-),\alpha_i(b)]$ be the interval containing $v_i$. Notice that the utilities $u_i(b,\vec{\beta}_{-i};v_i)$,  $u_i(b',\vec{\beta}_{-i};v_i)$ are linear functions on $v_i$, with slopes given by $H_i(b,\vec{\beta}_{-i})$,  $H_i(b',\vec{\beta}_{-i})$ respectively. For $b'<b$, we know that $H_i(b',\vec{\beta}_{-i})\leq H_i(b,\vec{\beta}_{-i})$ and $u_i(b,\vec{\beta}_{-i};v)\geq u_i(b',\vec{\beta}_{-i};v)-\varepsilon$ holds at $v=\alpha_i(b^-)$; therefore it must hold also at $v=v_i$. Similarly for $b'>b$, we know that $H_i(b',\vec{\beta}_{-i})\geq H_i(b,\vec{\beta}_{-i})$ and $u_i(b,\vec{\beta}_{-i};v)\geq u_i(b',\vec{\beta}_{-i};v)-\varepsilon$ holds at $v=\alpha_i(b)$; therefore it must hold also at $v=v_i$. We thus conclude that $\vec{\beta}$ is an $\varepsilon$-BNE.
\end{proof}

We now consider the basic computational problems of checking and computing best-responses of bidders. We assume throughout that bidding strategies provided in the input are given via rational quantities corresponding to the jump points $\alpha_j(b)$, as defined in \cref{sec:prelims}. The first step to be able to check or compute best-responses is the efficient computation of the $H$-functions defined above.

\paragraph{Computation of the $\bm{H}$-functions} Recall that 
$H_i(b,\vec{\beta}_{-i}) = \prob{\text{bidder } i \text{ wins}|b, \vec{\beta}_{-i}}$.
This probability clearly depends on bidder $i$'s prior on the other bidders' distributions, as well as on whether $b$ is the highest bid, and if it is, how many other highest bids there are in the auction, in case of a tie. While the form of the functions $H_i$ can be devised analytically, the expression involves exponentially many terms in the number of bidders $n$; therefore it is not obvious that it can be computed efficiently. The following lemma states that this is in fact possible.

\begin{lemma}\label{lem:H-functions}
	Given a bidder $i$, a bid $b$ and bidding strategies $\vec{\beta}_{-i}$ of the other bidders, the probability $H_i(b,\vec{\beta}_{-i})$ of bidder $i$ winning the item can be computed in polynomial time.
\end{lemma}

\begin{proof}
	For ease of notation, we present the proof for bidder $i=n$. The cases for the other
	bidders are analogous and can be handled, e.g., via an appropriate relabeling. The
	probability that bidder $n$ wins (given her bid and the bidding strategies of the
	other bidders) can be written as
	
	\begin{equation}
	\label{eq:H_functions_sum}
	H_n(b,\vec{\beta}_{-n})=\sum_{k=0}^{n-1}\frac{1}{k+1}T(b,n-1,k),  
	\end{equation}
	where, for $0\leq k \leq \ell \leq n-1 $, we use $T(b,\ell,k)$ to denote the
	probability that \emph{exactly} $k$ out of the first $\ell$ bidders bid exactly $b$,
	and the remaining $\ell-k$ bidders all bid below $b$; in other words, for the special
	case where $\ell=n-1$ in the above expression, $T(b,n-1,k)$ is the probability of
	the highest bid being $b$, with $k+1$ bidders (including bidder $n$) being tied for
	the highest bid. Next, for a given bidder $j$, let
	$$G_{j,b^-}=F_{n,j}(\alpha_j(b^-))=\prob{\beta_j(v_j)<b},\quad
	g_{j,b}=F_{n,j}(\alpha_j(b))-G_{jb^-}=\prob{\beta_j(v_j)=b}$$ denote the (perceived
	from the perspective of bidder $n$) probabilities that bidder $j$ bids below $b$,
	and exactly $b$, respectively. Note that $G_{j,b^-}$ and $g_{j,b}$ can be efficiently computed
	with access to $F_{n,j}$ and $\vec{\alpha}_{-n}$. Moreover, one could write
	
	\begin{equation}\label{eq:tblj}
	T(b,n-1,k)=\sum_{\substack{S\subseteq[n-1]\\ |S|=k}}\prod_{j\in S}g_{j,b}\cdot\prod_{j\not\in S}G_{j,b^-}.
	\end{equation}
	
	Notice that \eqref{eq:tblj} does not yield an efficient way of computing the probabilities, as the number of summands can be exponential in $n$. To bypass this obstacle, we observe that, more generally, the probabilities $T(b,\ell,k)$ can be computed from $G_{\ell,b^-}$ and $g_{\ell,b}$ via dynamic programming, by conditioning on bidder $\ell$'s bid, in the following way:
	\begin{align*}
	T(b,0,0) &=1;              &\\
	T(b,\ell,k) &=0,&\quad\text{for}\;\; k>\ell;\\
	T(b,\ell+1,0  ) &=T(b,\ell,0)G_{\ell+1,b^-};&\\
	T(b,\ell+1,k+1) &=T(b,\ell,k)g_{\ell+1,b} + T(b,\ell,k+1)G_{\ell+1,b^-}; &\text{for}\;\; k\leq\ell.
	\end{align*}
	Thus, all values of $T(b,n-1,k)$, for $k=0,\ldots,n-1$, can be computed with a total number of $O(n^2)$ recursive calls, so that $H_n(b,\vec{\beta}_{-n})$ can be computed in polynomial time.
\end{proof}

\cref{lem:H-functions} implies that the utilities in (\ref{eq:charepsilonbne1},~\ref{eq:charepsilonbne2}) of the characterization (\cref{lem:charepsilonbne}) can be computed in polynomial time. Since there are $O(n|B|^2)$ inequalities to check in \cref{lem:charepsilonbne}, we immediately conclude the following.

\begin{corollary}
	Given $\varepsilon \geq 0$, and a strategy profile $\vec{\beta}$ in a first-price auction with subjective priors, one can determine in polynomial time if $\vec{\beta}$ constitutes an $\varepsilon$-BNE. 
\end{corollary}

Using \cref{lem:H-functions}, we can now also efficiently compute best-responses, and, in fact, even \emph{exact} best-responses (i.e., $\varepsilon$-best-responses for $\varepsilon = 0$).

\begin{theorem}\label{theorem:BR}
	In a first-price auction with subjective priors, the bidders' best-responses can be computed in polynomial time.
\end{theorem}

\begin{proof}
	Given a bidder $i$ and the vector of bidding strategies $\vec{\beta}_{-i}$, one can compute in polynomial time the probabilities $H_i(b,\vec{\beta}_{-i})$ for each bid $b\in B$ using \cref{lem:H-functions}. Now recall that the utility of bidder $i$, when having a valuation of $v_i$ and bidding $b$, is given by $u_i(b,\vec{\beta}_{-i};v_i)=(v_i-b) \cdot H_i(b,\vec{\beta}_{-i})$, which is a linear function on $v_i$ having slope $H_i(b,\vec{\beta}_{-i})$. Thus, maximizing the utility amounts to taking the maximum (or \emph{upper envelope}) of $|B|$ linear functions; the result is a piecewise linear function whose jump points can be efficiently computed by solving linear equations. In particular, given bids $b<b'$, we can compute $\alpha=\tilde\alpha_i(b,b')$ as the solution of $u_i(b,\vec{\beta}_{-i};\alpha)=u_i(b',\vec{\beta}_{-i};\alpha)$, that is,
	\begin{equation*}
	\tilde\alpha_i(b,b')=
	\begin{cases}
	\frac{b'H_i(b',\vec{\beta}_{-i})-bH_i(b,\vec{\beta}_{-i})}{H_i(b',\vec{\beta}_{-i})-H_i(b,\vec{\beta}_{-i})} &\text{if }\;\; H_i(b',\vec{\beta}_{-i})\neq H_i(b,\vec{\beta}_{-i}),\\ 
	+\infty&\text{otherwise}.
	\end{cases}
	\end{equation*}
	Intuitively, $\tilde\alpha_i(b,b')$ is the jump point corresponding to bidding $b$ versus bidding $b'$: bidder $i$ achieves higher utility by bidding $b$ iff $v_i<\tilde\alpha_i(b,b')$. Now the highest value for which bidder $i$ (weakly) prefers bidding $b$ versus any other higher bid is $\min_{b'>b}\tilde\alpha_i(b,b')$; if at this valuation, bidding $b$ also achieves higher utility than bidding any other lower bid, then $\min_{b'>b}\tilde\alpha_i(b,b')$ is indeed one of the desired jump points. Otherwise, $b$ is a degenerate bid, in the sense that there is no valuation for which $b$ is an optimal response. Therefore, the jump points introduced in~\eqref{eq:jump_points} are given by $\alpha_i(b)=\max_{b'\leq b}\min_{b''>b'}\tilde\alpha_i(b',b'')$.\footnote{The maximization over $b'\leq b$ serves to exclude degenerate cases, e.g.\ if $b'<b<b''$ but $\tilde\alpha_i(b,b'')<\tilde\alpha_i(b',b'')<\tilde\alpha_i(b,b')$.} Clearly then, the $\alpha_i(b)$ can be found in polynomial time.
\end{proof}

\section{Existence and Membership in PPAD and FIXP}\label{sec:membership}

The existence of equilibria in our setting can essentially be established by adapting a proof by \citet{Athey2001}, which relies on Kakutani's fixed point theorem. Unfortunately, proofs that are based on this fixed point theorem cannot easily be turned into membership results for computational classes such as \ppad and \fixp. This is especially true for \fixp which is essentially defined as the class of all problems that can be solved by finding a Brouwer fixed point. In order to circumvent this obstacle we present a new proof that uses Brouwer's fixed point theorem. In this section, we first present this proof, and then utilize it to prove membership of our problems of interest in \ppad and \fixp.

\subsection{Existence of Equilibria via Brouwer's Fixed Point Theorem}

\begin{theorem}\label{thm:existence}
Every first-price auction with continuous subjective priors and finite bidding space admits a monotone non-decreasing and non-overbidding pure Bayes-Nash equilibrium.
\end{theorem}

\begin{proof}
Let $N=\{1,2,\dots, n\}$ be the set of bidders, $F_{i,j}$ the continuous subjective priors, and $0=b_0,b_1,\dots,b_m$ be the ordered list of bids, i.e., the elements of $B \subseteq [0,1]$. Recall that a monotone non-decreasing strategy $\beta_i: [0,1] \to B$ can be represented by its jump points $\alpha_i(b)$. Let
$$\mathcal{D} = \{\vec{\alpha} = (\alpha_1,\alpha_2,\dots,\alpha_n) \in ([0,1]^{m})^n \, | \, \forall i \in N, j \in [m]: \alpha_i(b_{j-2}) \leq \alpha_i(b_{j-1}) \land b_j \leq \alpha_i(b_{j-1}) \}$$
where we use the convention $\alpha_i(b_{-1}) := 0$ to keep the notation simple.
The domain $\mathcal{D}$ is the set of all monotone non-decreasing non-overbidding strategy profiles, represented by their jump points. Note that $\mathcal{D}$ is compact and convex.

In what follows we slightly abuse notation by replacing the strategy profile $\vec{\beta}$ by its representation $\vec{\alpha}$ in some terms.
Recall the functions $H_i(b,\vec{\alpha}_{-i})$ defined in \cref{sec:bestresponses}, which represent the probability that bidder $i$ wins the auction, if they bid $b$. By inspecting the proof of \cref{lem:H-functions}, it is easy to see that the quantities $G_{jb^-}$ and $g_{jb}$ are continuous with respect to $\vec{\alpha}_{-i}$, since the distributions are continuous. As a result, the terms $T(b,n-1,j)$ are also continuous in $\vec{\alpha}_{-i}$ (by \eqref{eq:tblj}), which implies that $H_i(b,\vec{\beta}_{-i})$ is also continuous in $\vec{\alpha}_{-i}$. Since the utility functions can be written as $u_i(b,\vec{\alpha}_{-i};v_i) = (v_i-b) \cdot H_i(b,\vec{\alpha}_{-i})$, it follows that the functions $(\vec{\alpha}_{-i}, v_i) \mapsto u_i(b,\vec{\alpha}_{-i};v_i)$ are continuous.

We now construct a function $G: \mathcal{D} \to \mathcal{D}$. For any bidder $i \in N$ and any $j \in [m]$, define the continuous function $\Delta_j^i: \mathcal{D} \to \R$ by
$$\Delta_j^i(\vec{\alpha}) = u_i(b_{j-1},\vec{\alpha}_{-i};\alpha_i(b_{j-1})) - \max_{\ell \geq j} u_i(b_\ell,\vec{\alpha}_{-i};\alpha_i(b_{j-1})).$$
The intuition behind the construction of $\Delta_j^i(\vec{\alpha})$ is as follows. We consider the point $\alpha_i(b_{j-1})$ (the last point where bidder $i$ currently bids $b_{j-1}$) and compare how beneficial it is for bidder $i$ to bid $b_{j-1}$ at this point, compared to using a larger bid. The sign of $\Delta_j^i(\vec{\alpha})$ essentially encodes the result of this comparison. As a result, it encodes whether bidder $i$: (a) is happy with the current value of $\alpha_i(b_{j-1})$ (i.e., $\Delta_j^i(\vec{\alpha}) = 0$), (b) would like to increase $\alpha_i(b_{j-1})$ (i.e., $\Delta_j^i(\vec{\alpha}) > 0$), or (c) would like to decrease $\alpha_i(b_{j-1})$ (i.e., $\Delta_j^i(\vec{\alpha}) < 0$).

Now, for any $\vec{\alpha} \in \mathcal{D}$, let $G(\vec{\alpha}) = \vec{\alpha}'$, where for all $i \in N$ and $j=1,2,\dots,m$ (consecutively and in that order)
\begin{equation}\label{eq:projectiond}\alpha_i'(b_{j-1}) = \trunc_{[\max\{b_j,\alpha_i'(b_{j-2})\},1]} (\alpha_i(b_{j-1}) + \Delta_j^i(\vec{\alpha})).\end{equation}
Note in particular that this is well-defined, since $\alpha_i'(b_{j-2})$ is defined before $\alpha_i'(b_{j-1})$.
The truncation operator immediately ensures that $\vec{\alpha}' \in \mathcal{D}$. Since $G$ is also clearly continuous, and $\mathcal{D}$ is compact and convex, it follows by Brouwer's fixed point theorem that there exists a $\vec{\alpha} \in \mathcal{D}$ with $G(\vec{\alpha})=\vec{\alpha}$. It remains to prove that $\vec{\alpha}$ corresponds to an equilibrium of the auction.

Consider some bidder $i \in N$. We will show that $\alpha_i$ is a best-response to $\vec{\alpha}_{-i}$ using the characterization of \cref{lem:charepsilonbne}. Consider any non-empty interval of non-empty interior $[\alpha_i(b_{j-1}),\alpha_i(b_j)]$, for some $j \in \{0,1,\dots,m\}$, where we use the convention that $\alpha_i(b_{-1})=0$, and recall that $\alpha_i(b_m)=1$.

\begin{itemize}
	\item First, we show that $u_i(b_{j},\vec{\alpha}_{-i};\alpha_i(b_{j})) \geq \max_{\ell > j} u_i(b_\ell,\vec{\alpha}_{-i};\alpha_i(b_{j}))$. Clearly, for $j=m$ this holds trivially. For $j<m$, this can immediately be rephrased as showing $\Delta_{j+1}^i(\vec{\alpha}) \geq 0$. Now, note that by assumption we have $\alpha_i(b_j) > \alpha_i(b_{j-1})$. Thus, since $\alpha_i(b_j)$ remains fixed under $G$, it must be that $\alpha_i(b_j) = b_{j+1}$ or $\Delta_{j+1}^i(\vec{\alpha}) \geq 0$. However, if $\alpha_i(b_j) = b_{j+1}$, then it also trivially holds that $\Delta_{j+1}^i(\vec{\alpha}) \geq 0$.
	\item Next, we show that $u_i(b_{j},\vec{\alpha}_{-i};\alpha_i(b_{j-1})) \geq \max_{\ell < j} u_i(b_\ell,\vec{\alpha}_{-i};\alpha_i(b_{j-1}))$. Again, this holds trivially for $j=0$, so we now consider $j > 0$.
	By the first bullet above, it holds that
	$$u_i(b_{j},\vec{\alpha}_{-i};\alpha_i(b_{j})) = \max_{\ell \geq j} u_i(b_\ell,\vec{\alpha}_{-i};\alpha_i(b_{j})).$$
	By the monotonicity of the $H$-functions, we can simply replace $\alpha_i(b_{j})$ by $\alpha_i(b_{j-1})$ in the equation above. Indeed, since by definition of the $H$-functions we have that $H_i(b,\vec{\alpha}_{-i}) \leq H_i(b',\vec{\alpha}_{-i})$ for any two bids $b \leq b'$, it follows in particular that the function $v \mapsto u_i(b,\vec{\alpha}_{-i};v) - u_i(b',\vec{\alpha}_{-i};v)$ is monotonically non-increasing. To see this, it suffices to write the utilities in terms of the $H$-functions, i.e., $u_i(b,\vec{\alpha}_{-i};v) = (v-b) \cdot H_i(b,\vec{\alpha}_{-i})$, and similarly for $b'$. As a result, we obtain that
	$$u_i(b_{j},\vec{\alpha}_{-i};\alpha_i(b_{j-1})) = \max_{\ell \geq j} u_i(b_\ell,\vec{\alpha}_{-i};\alpha_i(b_{j-1})).$$
	
	On the other hand, since $\alpha_i(b_{j-1}) < \alpha_i(b_j)$, it follows in particular that $\alpha_i(b_{k}) < 1$ for all $k < j$. As a result, since $\alpha_i(b_{k})$ remains fixed under $G$, it must be that $\Delta_{k+1}^i(\vec{\alpha}) \leq 0$ for all $k < j$, i.e.,
	$$u_i(b_{k},\vec{\alpha}_{-i};\alpha_i(b_{k})) \leq \max_{\ell \geq k+1} u_i(b_\ell,\vec{\alpha}_{-i};\alpha_i(b_{k}))$$
	which by monotonicity of the $H$-functions, as explained above, continues to hold if we replace $\alpha_i(b_{k})$ by $\alpha_i(b_{j-1})$, i.e., for all $k < j$ we have
	$$u_i(b_{k},\vec{\alpha}_{-i};\alpha_i(b_{j-1})) \leq \max_{\ell \geq k+1} u_i(b_\ell,\vec{\alpha}_{-i};\alpha_i(b_{j-1})).$$
	As a result it follows by induction that for all $k < j$
	$$u_i(b_{k},\vec{\alpha}_{-i};\alpha_i(b_{j-1})) \leq \max_{\ell \geq j} u_i(b_\ell,\vec{\alpha}_{-i};\alpha_i(b_{j-1})) = u_i(b_{j},\vec{\alpha}_{-i};\alpha_i(b_{j-1})).$$
\end{itemize}
By \cref{lem:charepsilonbne}, it immediately follows that $\alpha_i$ is a best-response to $\vec{\alpha}_{-i}$. Since this holds for all bidders $i \in N$, $\vec{\alpha}$ is an equilibrium.
\end{proof}

\subsection{FIXP Membership}\label{sec:inFIXP}

In order to study the exact equilibrium problem for the first-price auction in the context of \fixp, we consider the model where the distributions $F_{i,j}$ are given by algebraic circuits using the operations $\{+,-,\times,/,\max,\min,\sqrt[k]{\cdot}\}$ and rational constants, as is usual in this setting \citep{etessami2010complexity}. We show that the proof of existence in the previous section can be turned into a reduction.

\begin{theorem}\label{thm:in-fixp}
The problem \exactfpa lies in \fixp.
\end{theorem}

\begin{proof}
Clearly, the domain $\mathcal{D}$ of the function $G: \mathcal{D} \to \mathcal{D}$ from the proof of \cref{thm:existence} can be represented by a set of linear inequalities that can be constructed in polynomial time in $n$, $m$ and the representation length of $B$. Thus, it remains to show that we can construct in polynomial time an algebraic circuit that computes $G$.

We now describe how to construct a circuit for $G$ that only uses operations $\{+$, $-$, $\times$, $/$, $\max$, $\min$, $\sqrt[k]{\cdot}\}$ and rational constants.
First of all, note that probabilities of the form $\probability_{v_j \sim F_{i,j}} [\beta_j(v_j) \leq b] = F_{i,j}(\alpha_j(b))$ can easily be computed by the circuit, since the (cumulative) distribution functions $F_{i,j}$ are provided as algebraic circuits, and $\vec{\alpha}$ is the input to the circuit. It follows that the quantities $G_{jb^-}$ and $g_{jb}$ defined in the proof of \cref{lem:H-functions} can also be computed by the circuit. As a result, we can use the dynamic programming procedure described in the proof of \cref{lem:H-functions}, to compute the terms $T(b,n-1,j)$ by only using a polynomial number of operations. Note in particular, that the dynamic programming assignment rules can all be implemented using the available set of operations. With the terms $T(b,n-1,j)$ in hand, we can then easily compute the terms $H_i(b,\vec{\alpha}_{-i})$ for all $b \in B$, and thus evaluate the utility function $u_i(b,\vec{\alpha}_{-i};v_i) = (v_i-b) \cdot H_i(b,\vec{\alpha}_{-i})$ at any given $v_i \in [0,1]$. Finally, using the utility functions and the $\max$ operation we can now compute the terms $\Delta_j^i(\vec{\alpha})$ from the proof of \cref{thm:existence}, and then using $+$, $\max$, $\min$ and the constant $1$ we can output $\vec{\alpha}' = G(\vec{\alpha})$ by noting that
\[\alpha_i'(b_{j-1}) = \max\{\max\{b_j,\alpha_i'(b_{j-2})\}, \min\{1,\alpha_i(b_{j-1}) + \Delta_j^i(\vec{\alpha})\}\}. \qedhere\]
\end{proof}

\subsection{PPAD Membership}\label{sec:inPPAD}

In order to study the approximate equilibrium problem for the first-price auction in the context of \ppad, we consider a model where the distributions $F_{i,j}$ are \emph{polynomially-computable}, i.e., can be evaluated in polynomial time by a Turing machine.\footnote{Note that a function represented as an algebraic circuit (as in the previous section on \fixp) is not necessarily polynomially-computable, e.g., because the circuit can use ``repeated squaring'' to construct numbers with exponential bit complexity. Conversely, a function that is polynomially-computable cannot necessarily be represented as an algebraic circuit, because a Turing machine is not restricted to using arithmetic gates. We note that these two different models for representing functions are standard for \fixp and \ppad respectively.} In order to guarantee that an approximate equilibrium with polynomial bit complexity exists, we also assume that the distributions are \emph{polynomially continuous}. For a formal definition of these two standard properties in the context of \ppad, see \cref{app:inputs}. In this section we show that in this model, the problem of computing an $\varepsilon$-BNE lies in the class \ppad. We begin by observing that the polynomial-continuity of the distribution functions $F_{i,j}$ implies that the utility functions are also polynomially-continuous. This is proved in \cref{app:util-poly-cont}.

\begin{lemma}\label{lem:util-poly-cont}
If the distributions $F_{i,j}$ are polynomially-continuous, then so are the utility functions $\vec{\alpha} \mapsto u_i(b,\vec{\alpha}_{-i};v_i)$. In more detail, given $\varepsilon > 0$, we can in polynomial time compute $\delta > 0$ such that for all $i \in N$, $b \in B$ and $v_i \in [0,1]$
$$\|\vec{\alpha} - \vec{\alpha}'\|_\infty \leq \delta \implies |u_i(b,\vec{\alpha}_{-i};v_i) - u_i(b,\vec{\alpha}_{-i}';v_i)| \leq \varepsilon.$$
In particular, $\delta$ can be represented using a polynomial number of bits.
\end{lemma}

\noindent We are now ready to state the main result of this section.

\begin{theorem}\label{thm:in-ppad}
The problem \epsfpa lies in \ppad.
\end{theorem}

\begin{proof}
We show that the existence proof of \cref{thm:existence} can be turned into a polynomial-time many-one reduction to the problem of computing an approximate Brouwer fixed point of a polynomially-computable and polynomially-continuous function over a bounded polytope given by linear inequalities, known to lie in \ppad \citep[Proposition 2]{etessami2010complexity}.

Since the distributions $F_{i,j}$ are polynomially-computable, and by the arguments provided in the proof of \cref{thm:in-fixp} (including the dynamic programming procedure from \cref{lem:H-functions}), it immediately follows that $G$ is polynomially-computable. The polynomial-continuity of $G$ also immediately follows from the polynomial-continuity of the utility functions (\cref{lem:util-poly-cont}). Thus, the problem of computing an approximate fixed point of $G$ lies in \ppad.

Given $\varepsilon > 0$, by \cref{lem:util-poly-cont} we can compute $\delta > 0$ so that for all $i \in N$, $b \in B$ and $v_i \in [0,1]$
$$\|\vec{\alpha} - \vec{\alpha}'\|_\infty \leq \delta \implies |u_i(b,\vec{\alpha}_{-i};v_i) - u_i(b,\vec{\alpha}_{-i}';v_i)| \leq \frac{\varepsilon}{16m}.$$

Now consider any $\delta$-approximate fixed point of $G$, i.e., $\vec{\alpha} \in \mathcal{D}$ such that $\|G(\vec{\alpha}) - \vec{\alpha}\|_\infty \leq \delta$. Let $\vec{\alpha}' = G(\vec{\alpha})$. We prove that $\vec{\alpha}'$ is an $\varepsilon$-approximate equilibrium of the first-price auction. This shows that \epsfpa reduces to the Brouwer fixed point computation problem, and thus lies in \ppad.

Since $\vec{\alpha}' = G(\vec{\alpha})$ and $\|G(\vec{\alpha}) - \vec{\alpha}\|_\infty \leq \delta$, it holds that $\|\vec{\alpha} - \vec{\alpha}'\|_\infty \leq \delta$ and thus
\begin{equation}\label{eq:utils-close}
|u_i(b,\vec{\alpha}_{-i};v_i) - u_i(b,\vec{\alpha}_{-i}';v_i)| \leq \frac{\varepsilon}{16m}
\end{equation}
for all $i \in N$, $b \in B$ and $v_i \in [0,1]$. In particular, we also have that $|\Delta_j^i(\vec{\alpha}) - \Delta_j^i(\vec{\alpha}')| \leq 2(\varepsilon/16m+\delta) \leq \varepsilon/4m$ (since the utility functions are also $1$-Lipschitz with respect to $v_i$). Note that here we assumed without loss of generality that $\delta \leq \varepsilon/16m$.

Fix some bidder $i \in N$. Consider any non-empty interval $[\alpha_i'(b_{j-1}),\alpha_i'(b_j)]$ for some $j \in \{0,1,\dots,m\}$, where we use the convention that $\alpha_i'(b_{-1})=0$, and recall that $\alpha_i'(b_m)=1$.
\begin{itemize}
	\item First, we show that $u_i(b_{j},\vec{\alpha}_{-i}';\alpha_i'(b_{j})) \geq \max_{\ell > j} u_i(b_\ell,\vec{\alpha}_{-i}';\alpha_i'(b_{j})) - \varepsilon/2$. Clearly, for $j=m$ this holds trivially.
	For $j<m$, this can immediately be rephrased as showing $\Delta_{j+1}^i(\vec{\alpha}') \geq -\varepsilon/2$. By \eqref{eq:utils-close}, it suffices to show that $\Delta_{j+1}^i(\vec{\alpha}) \geq -\varepsilon/2 + \varepsilon/4m$. But if $\Delta_{j+1}^i(\vec{\alpha}) < -\varepsilon/2 + \varepsilon/4m \leq - \varepsilon/16m \leq -\delta$, then by construction of $G$, since $|\alpha_i(b_j) - \alpha_i'(b_j)| \leq \delta$, it must be that $\alpha_i'(b_j) = b_{j+1}$ or $\alpha_i'(b_j) = \alpha_i'(b_{j-1})$. In the former case, it trivially holds that $\Delta_{j+1}^i(\vec{\alpha}') \geq 0 \geq -\varepsilon/2$. The latter case is impossible, since we assumed that $\alpha_i'(b_{j-1}) < \alpha_i'(b_j)$.
	
	\item Next, we show that $u_i(b_{j},\vec{\alpha}_{-i}';\alpha_i'(b_{j-1})) \geq \max_{\ell < j} u_i(b_\ell,\vec{\alpha}_{-i}';\alpha_i'(b_{j-1})) - \varepsilon$. Again, this holds trivially for $j=0$, so we now consider $j > 0$. By the first bullet above, it holds that
	$$u_i(b_{j},\vec{\alpha}_{-i}';\alpha_i'(b_{j})) \geq \max_{\ell \geq j} u_i(b_\ell,\vec{\alpha}_{-i}';\alpha_i'(b_{j})) - \varepsilon/2.$$
	As a result, by the monotonicity of the $H$-functions, as explained in the proof of \cref{thm:existence}, this continues to hold if we replace $\alpha_i'(b_{j})$ by $\alpha_i'(b_{j-1})$, i.e.,
	\begin{equation}\label{eq:in-ppad-bound}
	u_i(b_{j},\vec{\alpha}_{-i}';\alpha_i'(b_{j-1})) \geq \max_{\ell \geq j} u_i(b_\ell,\vec{\alpha}_{-i}';\alpha_i'(b_{j-1})) - \varepsilon/2.
	\end{equation}
	
	On the other hand, since $\alpha_i'(b_{j-1}) < \alpha_i'(b_j)$, it follows in particular that $\alpha_i'(b_{k}) < 1$ for all $k < j$. As a result, by construction of $G$, and since $|\alpha_i(b_k) - \alpha_i'(b_k)| \leq \delta$, it must be that $\Delta_{k+1}^i(\vec{\alpha}) \leq \delta$ for all $k < j$. By \eqref{eq:utils-close} it follows that $\Delta_{k+1}^i(\vec{\alpha}') \leq \delta + \varepsilon/4m \leq \varepsilon/2m$ for all $k < j$, which yields
	$$u_i(b_{k},\vec{\alpha}_{-i}';\alpha_i'(b_{k})) \leq \max_{\ell \geq k+1} u_i(b_\ell,\vec{\alpha}_{-i}';\alpha_i'(b_{k})) + \frac{\varepsilon}{2m}$$
	which by monotonicity of the $H$-functions, as explained in the proof of \cref{thm:existence}, continues to hold if we replace $\alpha_i'(b_{k})$ by $\alpha_i'(b_{j-1})$, i.e., for all $k < j$ we have
	$$u_i(b_{k},\vec{\alpha}_{-i}';\alpha_i'(b_{j-1})) \leq \max_{\ell \geq k+1} u_i(b_\ell,\vec{\alpha}_{-i}';\alpha_i'(b_{j-1})) + \frac{\varepsilon}{2m}.$$
	As a result it follows by induction that for all $k < j$
	$$u_i(b_{k},\vec{\alpha}_{-i}';\alpha_i'(b_{j-1})) \leq \max_{\ell \geq j} u_i(b_\ell,\vec{\alpha}_{-i}';\alpha_i'(b_{j-1})) + (j-k)\frac{\varepsilon}{2m}$$
	which together with \eqref{eq:in-ppad-bound} yields that for all $k < j$
	$$u_i(b_{k},\vec{\alpha}_{-i}';\alpha_i'(b_{j-1})) \leq u_i(b_{j},\vec{\alpha}_{-i}';\alpha_i'(b_{j-1})) + m \frac{\varepsilon}{2m} + \frac{\varepsilon}{2}.$$
\end{itemize}
By \cref{lem:charepsilonbne}, it immediately follows that $\alpha_i'$ is an $\varepsilon$-best-response to $\vec{\alpha}_{-i}'$. Since this holds for all bidders $i \in N$, $\vec{\alpha}'$ is an $\varepsilon$-equilibrium.
\end{proof}

\section{Computational Hardness}\label{sec:hardness}

In this section we prove computational hardness results for the problem of computing an equilibrium of a first-price auction with subjective priors. Namely, we show that computing an $\varepsilon$-BNE is \ppad-hard, while computing an exact BNE is \fixp-hard. Our computational hardness results are particularly robust, because they hold even if we apply all of the following restrictions:
\begin{itemize}
	\item the bidding space is $B = \{0,1/5,2/5,3/5,4/5\}$,
	\item the value distributions $F_{i,j}$ are given by very simple piecewise constant density functions,
	\item $\varepsilon$ is some sufficiently small \emph{constant}. \hfill \textit{(only relevant for $\varepsilon$-BNE)}
\end{itemize}
In particular, by a simple rescaling argument, the hardness results also hold when the bidding space consists of all monetary amounts that are increments of some fixed denomination (e.g., one cent) up to some number $m$.\footnote{Note that $m$ should be provided in the input in \emph{unary} representation. This is necessary to ensure that the bidding space has polynomial size, thus allowing efficient computation of best-responses. See the discussion in \cref{sec:prelims} regarding our assumption of an explicit bidding space.} For example, there exists a sufficiently small constant $\varepsilon$ such that it is \ppad-hard to compute an $\varepsilon$-BNE when the bidding space is $B=\{0,1/100,2/100,\dots,99/100,1,101/100,\dots,m-1/100,m\}$.

Together with the corresponding membership results proved in the previous section (\cref{thm:in-fixp,thm:in-ppad}), we thus obtain the following two theorems, which are the main results of this paper.

\begin{theorem}
There exists a constant $\varepsilon > 0$ such that the problem \epsfpa is \ppad-complete.
\end{theorem}

\begin{theorem}
The problem \exactfpa is \fixp-complete.
\end{theorem}

In the rest of this section, we present the proof of our hardness results. A nice feature of our proof is that we provide a \emph{single} reduction to prove both \ppad- and \fixp-hardness. In more detail, we reduce from the so-called \emph{Generalized Circuit problem}, which has been instrumental for proving \ppad-hardness results for Nash equilibrium computation problems \citep{daskalakis2009complexity,chen2009settling,rubinstein2018inapproximability}. In fact, we show that it suffices to consider significantly restricted versions of the Generalized Circuit problem when proving hardness results, and that an exact version of the problem can also be used to prove \fixp-hardness. Since we believe that these points may be of independent interest for future works, they are presented separately in \cref{sec:gcircuit}. Our reduction from this problem to equilibrium computation in first-price auctions is then presented in \cref{sec:reduction}.

\subsection{The Generalized Circuit Problem}\label{sec:gcircuit}

Generalized circuits, defined by \citet{chen2009settling}, can be viewed as a generalization of arithmetic circuits where we also allow \emph{cycles}. This means that instead of representing a function, a generalized circuit represents a certain kind of constraint satisfaction problem. Indeed, the goal in the Generalized Circuit problem is to assign a value to each gate of the circuit such that all the gates are (approximately) satisfied. Importantly, gates are only allowed to take values in $[0,1]$ and arithmetic operations are truncated accordingly. As a result, it can be shown that by Brouwer's fixed point theorem, there always exists an assignment of values that satisfies all the gates. However, computing even an approximate assignment is already \ppad-hard, i.e., essentially as hard as any Brouwer fixed point computation. We now provide some formal definitions.

\begin{definition}
A \emph{generalized circuit}\footnote{Note that in the usual definition of generalized circuits, every gate also contains a rational parameter $\zeta \in [0,1]$, which is used by some gate-types, e.g., a gate performing multiplication by the constant $\zeta$. In our definition, gates do not contain this rational parameter, because, as we show in \cref{prop:gcircuit-ppad,prop:gcircuit-fixp}, these gate-types are actually not needed for the problems to be hard.} with gate-types $\mathcal{G}$ is a list of gates $g_1, g_2, \dots, g_\nn$. Every gate $g_i$ is a 3-tuple $g_i=(G,j,k)$, where $G \in \mathcal{G}$ is the type of the gate, and $j,k \in [\nn] = \ssets{1,\dots,\nn}$ are the indices of the input gates $g_j,g_k$ ($i,j,k$ distinct).
\end{definition}

\noindent Before describing possible types of gates, we introduce some notation. Let $\trunc = \trunc_{[0,1]}$. Furthermore, we use the notation $x = y \pm \varepsilon$ to denote that $|x-y| \leq \varepsilon$.

Consider a generalized circuit $g_1, g_2, \dots, g_\nn$ and an assignment $\val: [\nn] \to [0,1]$ of values to its gates. We say that a gate is $\varepsilon$-satisfied by the assignment, if the constraint imposed by this gate is satisfied with error at most $\varepsilon$. The constraint that a gate $g_i=(G,j,k)$ must satisfy depends on its gate-type $G \in \mathcal{G}$, e.g.,
\begin{itemize}
	\item if $G=G_1$, then $\val [g_i] = 1 \pm \varepsilon$ \hfill \textit{(constant 1)}
	\item if $G=G_+$, then $\val[g_i] = \trunc(\val[g_j]+\val[g_k]) \pm \varepsilon$ \hfill \textit{(addition)}
	\item if $G=G_-$, then $\val[g_i] = \trunc(\val[g_j]-\val[g_k]) \pm \varepsilon$ \hfill \textit{(subtraction)}
	\item if $G=G_{1-}$, then $\val[g_i] = 1 - \val[g_j] \pm \varepsilon$ \hfill \textit{(complement)}
	\item if $G=G_{\times 2}$, then $\val[g_i] = \trunc(2 \cdot \val[g_j]) \pm \varepsilon$ \hfill \textit{(multiplication by 2)}
	\item if $G=G_\times$, then $\val[g_i] = \val[g_j]\cdot\val[g_k] \pm \varepsilon$ \hfill \textit{(multiplication)}
	\item if $G=G_{(\cdot)^2}$, then $\val[g_i] = (\val[g_j])^2 \pm \varepsilon$ \hfill \textit{(square)}
\end{itemize}
We are now ready to define the associated computational problem.

\begin{definition}
Let $\varepsilon > 0$. The problem $\varepsilon$-\gcircuit with gate-types $\mathcal{G}$ is defined as follows: given a generalized circuit $g_1, g_2, \dots, g_\nn$ with gate-types $\mathcal{G}$, find an assignment $\val: [\nn] \to [0,1]$ to the gates such that they are all $\varepsilon$-satisfied.
\end{definition}

\noindent \citet{rubinstein2018inapproximability} proved that this problem is \ppad-complete for some sufficiently small constant $\varepsilon > 0$ and a relatively large set of gate-types $\mathcal{G}$. In \cref{app:gcircuit-ppad}, we prove that the problem remains hard, even with a very restricted set of gate-types.

\begin{proposition}\label{prop:gcircuit-ppad}
There exists a constant $\varepsilon > 0$ such that the problem $\varepsilon$-\gcircuit with gate-types $\mathcal{G} = \{G_+, G_{1-}\}$ is \ppad-complete. This continues to hold if we instead take $\mathcal{G} = \{G_1, G_{-}\}$.
\end{proposition}

\noindent We can also define a problem \exactgcircuit, where the goal is to find an assignment that \emph{exactly} satisfies all constraints (i.e., with $\varepsilon=0$). In \cref{app:gcircuit-fixp}, we prove the following result.

\begin{proposition}\label{prop:gcircuit-fixp}
The problem \exactgcircuit with gate-types $\mathcal{G} = \{G_{1-},G_+,G_{(\cdot)^2}\}$ is \fixp-complete. This continues to hold if we instead take $\mathcal{G} = \{G_{1-},G_{\times 2},G_\times\}$.
\end{proposition}

\subsection{Reduction to BNE-FPA}\label{sec:reduction}
In this section, we present a reduction that achieves the following: given a generalized circuit, it constructs (in polynomial time) an instance of the first-price auction problem, such that for all $\varepsilon \in [0,1/10^5]$, from any $\varepsilon$-BNE we can extract an $500\varepsilon$-satisfying assignment for the generalized circuit. Furthermore, this ``extraction'' of the assignment from an $\varepsilon$-BNE can be done efficiently and, in fact, using a simple so-called separable linear transformation. This ensures that in the case $\varepsilon=0$, we obtain a so-called SL-reduction from \exactgcircuit, which yields the \fixp-hardness result \citep{etessami2010complexity}. If we let $\gceps>0$ be a constant such that $\gceps$-\gcircuit is \ppad-hard, then for $\varepsilon=\min\{1/10^5,\gceps/500\}$ the reduction is a valid polynomial-time many-one reduction, which yields the \ppad-hardness result.

An obstacle to obtaining the desired reduction is that it is unclear how to simulate a $G_+$-gate or a $G_\times$-gate. As a result, we reduce from the \gcircuit problem with gate-types $\mathcal{G} = \{G_{\times 2},G_{1-},G_\phi\}$, where $\phi: [0,1]^2 \to [0,1]$, $(x,y) \mapsto \frac{1}{4} (x+1)(y+1)$. This means that a gate $g_i=(G_\phi,j,k)$ enforces the constraint $\val[g_i]=\phi(\val[g_j],\val[g_k]) \pm \varepsilon$. In \cref{app:gcircuit-special} we prove that this set of gate-types is sufficient for our desired hardness results.

\begin{lemma}\label{lem:gcircuit-special}
Let $\mathcal{G} = \{G_{\times 2},G_{1-},G_\phi\}$.
There exists a constant $\gceps > 0$ such that the problem $\gceps$-\gcircuit with gate-types $\mathcal{G}$ is \ppad-complete. Furthermore, \exactgcircuit with gate-types $\mathcal{G}$ is \fixp-complete.
\end{lemma}

\noindent\textbf{The reduction.} We begin with a high-level description of the reduction. Consider a generalized circuit $g_1, g_2, \dots, g_\nn$ with gate-types $\mathcal{G} = \{G_{\times 2},G_{1-},G_\phi\}$. We construct a first-price auction with bidding space $B = \{0,1/5,2/5,3/5,4/5\}$ and a set of bidders $N=\{1,2,\dots,\mm\}$ where $\mm=10\nn$. For every $i \in [\nn]$, bidder $i$ will ``correspond'' to gate $g_i$, in the sense that, in any $\varepsilon$-BNE $\vec{\beta}$, the position of the second jump point of $\beta_i$, i.e., $\alpha_i(1/5)$ will encode the value $\val[g_i]$ that we will assign to gate $g_i$. Thus, we will refer to the bidders $1,2,\dots,\nn$ as \emph{gate-bidders}. The rest of the bidders will be used as intermediate steps to enforce the desired constraints on the strategies of the gate-bidders. Accordingly, we will refer to them as \emph{auxiliary-bidders}. Note that for every gate-bidder, there are 9 auxiliary-bidders available (if needed). For convenience, we describe the construction with the value space $[0,5]$ instead of $[0,1]$. This is without consequence, since this re-scaling of the instance simply means that we have to replace $\varepsilon$ by $5\varepsilon$ at the end. Note that as a result of the re-scaling, the bidding space is now simply $B = \{0,1,2,3,4\}$.

\medskip

\noindent\textbf{Valid strategies and encoded value.}
Let $\vec{\beta}$ be any $\varepsilon$-BNE of the auction. A bidder $i \in N$ is said to be \emph{valid}, if $\alpha_i(0) \in [1,1+1/2]$, $\alpha_i(1) \in [2+1/3-2\varepsilon,2+2/3+2\varepsilon]$, $\alpha_i(2) \in [3+1/2,5]$ and $\alpha_i(3)=5$. The bidder $i$ is \emph{almost-valid}, if the condition on $\alpha_i(1)$ is relaxed to $\alpha_i(1) \in [2,3]$.
For every bidder $i \in N$, we define the value encoded by bidder $i$ according to $\vec{\beta}$, as
\begin{equation*}
\val_{\vec{\beta}}[i] = \left\{\begin{tabular}{ll}
$\trunc_{[0,1]}(3 (\alpha_i(1)-2-1/3))$ & if $i$ is valid,\\
\nul & otherwise.
\end{tabular}\right.
\end{equation*}
Note that we always have $\val_{\vec{\beta}}[i] \in [0,1] \cup \{\nul\}$.
In the rest of the proof, we drop the subscript $\vec{\beta}$, since it is understood from the context. Our construction will ensure that for all $i \in [\nn]$, bidder $i$ is valid and as a result $\val[i] \in [0,1]$. Furthermore, letting $\val[g_i] := \val[i]$ will yield an $100\varepsilon$-satisfying assignment of the generalized circuit.

\medskip

\noindent\textbf{Gadgets.} The rest of the proof describes the construction of the distribution functions $F_{i,j}$. We begin by constructing some \emph{unary} gadgets. A unary gadget has a single ``input'' bidder $j \in N$  and an output bidder $i \in N \setminus \{j\}$. The goal of such a gadget is to establish a constraint on $\beta_i$ that depends on $\beta_j$, but not on the strategy of any other bidder. This is achieved by setting $F_{i,k}$ for all $k \in N \setminus \{i,j\}$, such that its (piecewise constant) density function has a single piece of volume $1$ lying in $[0,1]$. As a result, because of the no-overbidding assumption, bidder $i$ will believe that all bidders $k \in N \setminus \{i,j\}$ bid $0$ with probability $1$. The behavior of the gadget is then determined by the precise construction of $F_{i,j}$. 

\medskip

\noindent\textbf{Base Gadget.}
The base gadget is a unary gadget with input bidder $j$ and output bidder $i$ that has four parameters $\gamma_\ell, \gamma_r, \ell, r \in [0,1]$ with $\gamma_\ell + \gamma_r < 1$ and $r-\ell > 0$. The piecewise constant density function of $F_{i,j}$ is defined as follows. There is a piece of volume $\gamma_\ell$ in the interval $[1+1/2,1+3/4]$, a piece of volume $1-\gamma_\ell-\gamma_r$ in $[2+\ell,2+r]$, and finally a piece of volume $\gamma_r$ in $[3+1/4,3+1/2]$. See \cref{fig:base_gadget} for an illustration.

When the parameters are $(\gamma_\ell,\gamma_r,\ell,r) = (1/3,1/3,1/3,2/3)$, we call this the \emph{standard} base gadget. It will immediately follow from \cref{clm:base-gadget} below that if the input bidder $j$ of the standard base gadget is valid, then so is the output bidder $i$, and furthermore $\val[i] = \val[j] \pm 6 \varepsilon$. In other words, this gadget can be used to copy the value encoded by one bidder onto some other bidder.

\begin{figure}
\fbox{%
    \colorbox{blue!5!white}{
    \parbox{0.965\textwidth}{%
    \centering
\tikzset{every picture/.style={line width=0.75pt}} 
\begin{tikzpicture}[x=0.74pt,y=0.75pt,yscale=-0.75,xscale=0.75,every node/.style={scale=1}]

\draw    (54,310) -- (800,310) ;
\draw    (53,301) -- (53,319.5) ;
\draw    (233,301) -- (233,319.5) ;
\draw    (413,301) -- (413,319.5) ;
\draw    (593,301) -- (593,319.5) ;
\draw    (773,301) -- (773,319.5) ;
\draw    (323,301) -- (323,319.5) ;
\draw    (368,301) -- (368,319.5) ;
\draw  [color={rgb, 255:red, 0; green, 0; blue, 0 }  ,draw opacity=1 ][fill={rgb, 255:red, 155; green, 155; blue, 155 }  ,fill opacity=0.51 ] (323,260.5) -- (368,260.5) -- (368,309) -- (323,309) -- cycle ;
\draw    (462,301) -- (462,319.5) ;
\draw    (544,301) -- (544,319.5) ;
\draw  [color={rgb, 255:red, 0; green, 0; blue, 0 }  ,draw opacity=1 ][fill={rgb, 255:red, 155; green, 155; blue, 155 }  ,fill opacity=0.51 ] (462,276.5) -- (544,276.5) -- (544,310) -- (462,310) -- cycle ;
\draw    (638,301) -- (638,319.5) ;
\draw    (683,301) -- (683,319.5) ;
\draw  [color={rgb, 255:red, 0; green, 0; blue, 0 }  ,draw opacity=1 ][fill={rgb, 255:red, 155; green, 155; blue, 155 }  ,fill opacity=0.51 ] (638,260.5) -- (683,260.5) -- (683,309) -- (638,309) -- cycle ;

\draw (48,327) node [anchor=north west][inner sep=0.75pt]   [align=left] {0};
\draw (228,327) node [anchor=north west][inner sep=0.75pt]   [align=left] {1};
\draw (409,327) node [anchor=north west][inner sep=0.75pt]   [align=left] {2};
\draw (588,327) node [anchor=north west][inner sep=0.75pt]   [align=left] {3};
\draw (768,327) node [anchor=north west][inner sep=0.75pt]   [align=left] {4};
\draw (304,327) node [anchor=north west][inner sep=0.75pt]   [align=left] {};
\draw (309,327.4) node [anchor=north west][inner sep=0.75pt]  [font=\scriptsize]  {${\textstyle 1+\frac{1}{2}}$};
\draw (462,327.4) node [anchor=north west][inner sep=0.75pt]  [font=\scriptsize]  {${\textstyle 2+\ell }$};
\draw (354,327.4) node [anchor=north west][inner sep=0.75pt]  [font=\scriptsize]  {${\textstyle 1+\frac{3}{4}}$};
\draw (520,327.4) node [anchor=north west][inner sep=0.75pt]  [font=\scriptsize]  {${\textstyle 2+r}$};
\draw (335,273.4) node [anchor=north west][inner sep=0.75pt]    {$\gamma _{\ell }$};
\draw (652,273.4) node [anchor=north west][inner sep=0.75pt]    {$\gamma _{r}$};
\draw (465,286.4) node [anchor=north west][inner sep=0.75pt]  [font=\scriptsize]  {$1-\gamma _{\ell } \ -\ \gamma _{r}$};
\draw (807,295.4) node [anchor=north west][inner sep=0.75pt]    {$F_{i,j}$};
\draw (70,244) node [anchor=north west][inner sep=0.75pt]   [align=left] {{\LARGE Base Gadget}};
\draw (625,327.4) node [anchor=north west][inner sep=0.75pt]  [font=\scriptsize]  {${\textstyle 3+\frac{1}{4}}$};
\draw (671,327.4) node [anchor=north west][inner sep=0.75pt]  [font=\scriptsize]  {${\textstyle 3+\frac{1}{2}}$};

\end{tikzpicture}
}}}
    \caption{An illustration of the base gadget. The density of $F_{i,j}$ is depicted. When $\gamma_\ell=\gamma_r = 1/3$, $\ell=1/3$ and $r=2/3$ we obtain a \emph{standard} base gadget, which essentially (approximately) ``copies'' the value $\val[i]$ of the input bidder $i$ to the value $\val[j]$ of the output bidder $j$.}
    \label{fig:base_gadget}
\end{figure}
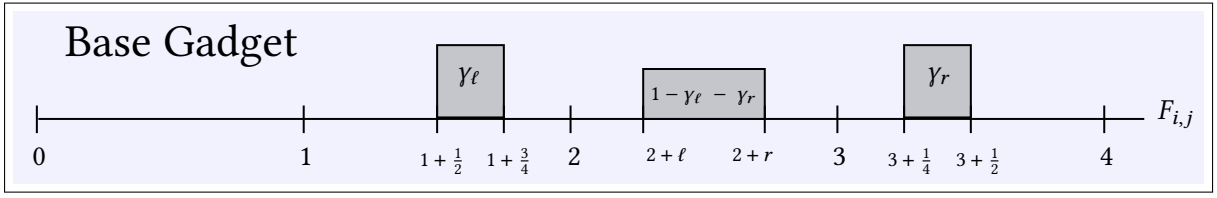

\begin{claim}\label{clm:base-gadget}
Let $\gamma_\ell, \gamma_r, \ell, r \in [0,1]$ with $\gamma_\ell, \gamma_r \geq 1/20$, $\gamma_\ell + \gamma_r < 1$ and $\ell < r$.
Consider a base gadget with input bidder $j$ and output bidder $i$, and parameters $(\gamma_\ell, \gamma_r, \ell, r)$. It holds that:
\begin{itemize}
	\item If the input bidder $j$ is almost-valid, then the output bidder $i$ is also almost-valid.
	\item If $\gamma_\ell,\gamma_r \geq 1/3$ and $j$ is almost-valid, then $i$ is valid and
	$$\val[i] = (3\gamma_\ell-1) + 3(1-\gamma_\ell-\gamma_r) \frac{T_{[2+\ell,2+r]}(\alpha_j(1)) - (2+\ell)}{r-\ell} \pm 6\varepsilon.$$
\end{itemize}
\end{claim}

\begin{proof}
We begin by obtaining some equations that will be useful for various proofs in this section. Consider any unary gadget with input bidder $j$ and output bidder $i$. To simplify notation, for $b \in \{0,1,2,3,4\}$, let $p_b$ be the probability that bidder $j$ bids $b$, as perceived by bidder $i$. Formally,
\begin{equation*}
	p_b := \probability_{v_j \sim F_{i,j}}[\beta_j(v_j) = b] = \left\{\begin{tabular}{ll}
		$F_{i,j}(\alpha_j(b)) - F_{i,j}(\alpha_j(b-1))$ & if $b \in \{1,2,3,4\}$\\
		$F_{i,j}(\alpha_j(0))$ & if $b=0$.
	\end{tabular}\right.
\end{equation*}
Recall the quantity $H_i(b,\vec{\beta}_{-i})$ defined in \cref{sec:bestresponses}, which represents the probability that bidder $i$ wins the auction if she bids $b$, and the other bidders act according to $\vec{\beta}_{-i}$. We drop $\vec{\beta}_{-i}$ from the notation, since it is clear from the context. Going back to our unary gadget, it is easy to see that $H_i(0)=p_0/\mm$, $H_i(1)=p_0+p_1/2$, $H_i(2)=p_0+p_1+p_2/2$, $H_i(3)=p_0+p_1+p_2+p_3/2$ and $H_i(4)=p_0+p_1+p_2+p_3+p_4/2$. Here we used the fact that, by construction of $F_{i,k}$, bidder $i$ perceives that all other bidders $k \in N \setminus \{i,j\}$ bid $0$ with probability $1$.

Now, by \cref{lem:charepsilonbne} the first jump point $\alpha_i(0)$ of $\beta_i$ must necessarily satisfy $u_i(0,\vec\beta_{-i};\alpha_i(0)) \geq u_i(1,\vec\beta_{-i};\alpha_i(0)) - \varepsilon$ (because the interval $(0,\alpha_i(0))$ is non-empty by the non-overbidding assumption). We can rewrite this as $H_i(0) \cdot (\alpha_i(0)-0) \geq H_i(1) \cdot (\alpha_i(0)-1) - \varepsilon$, which yields
\begin{equation}\label{eq:jump-0}
\alpha_i(0) \leq \frac{H_i(1)+\varepsilon}{H_i(1)-H_i(0)} = 1 + \frac{H_i(0) + \varepsilon}{H_i(1)-H_i(0)} = 1 + \frac{p_0/\mm + \varepsilon}{p_0(\mm-1)/\mm + p_1/2}
\end{equation}
where the fraction is interpreted as $+\infty$ when $p_0+p_1=0$. Similarly, by \cref{lem:charepsilonbne}, the fourth jump point must satisfy $u_i(4,\vec\beta_{-i};\alpha_i(3)) \geq u_i(3,\vec\beta_{-i};\alpha_i(3)) - \varepsilon$, unless $\alpha_i(3) = 5$. Rewriting this as $H_i(4) \cdot (\alpha_i(3)-4) \geq H_i(3) \cdot (\alpha_i(3)-3) - \varepsilon$, we obtain that $\alpha_i(3)=5$ or
\begin{equation}\label{eq:jump-3}
\alpha_i(3) \geq \frac{4H_i(4)-3H_i(3) - \varepsilon}{H_i(4)-H_i(3)} = 4 + \frac{H_i(3) - \varepsilon}{H_i(4)-H_i(3)} = 4 + \frac{p_0+p_1+p_2+p_3/2 - \varepsilon}{p_3/2+p_4/2}.
\end{equation}
Again, by \cref{lem:charepsilonbne}, the third jump point must satisfy $u_i(3,\vec\beta_{-i};\alpha_i(2)) \geq u_i(2,\vec\beta_{-i};\alpha_i(2)) - \varepsilon$, unless $\alpha_i(2) = \alpha_i(3)$, and it must satisfy $u_i(2,\vec\beta_{-i};\alpha_i(2)) \geq u_i(3,\vec\beta_{-i};\alpha_i(2)) - \varepsilon$, unless $\alpha_i(2) = \alpha_i(1)$. Thus it follows that
\begin{equation}\label{eq:jump-2}
\alpha_i(2) = \trunc_{[\alpha_i(1),\alpha_i(3)]}\left( \frac{3H_i(3) - 2H_i(2) \pm \varepsilon}{H_i(3)-H_i(2)} \right) = \trunc_{[\alpha_i(1),\alpha_i(3)]}\left( 3 + \frac{2p_0+2p_1+p_2 \pm 2\varepsilon}{p_2+p_3} \right).
\end{equation}
Finally, by \cref{lem:charepsilonbne}, the second jump point must satisfy $u_i(2,\vec\beta_{-i};\alpha_i(1)) \geq u_i(1,\vec\beta_{-i};\alpha_i(1)) - \varepsilon$, unless $\alpha_i(1) = \alpha_i(2)$, and it must satisfy $u_i(1,\vec\beta_{-i};\alpha_i(1)) \geq u_i(2,\vec\beta_{-i};\alpha_i(1)) - \varepsilon$, unless $\alpha_i(1) = \alpha_i(0)$. As a result, it must be that
\begin{equation}\label{eq:jump-1}
\alpha_i(1) = \trunc_{[\alpha_i(0),\alpha_i(2)]}\left( \frac{2H_i(2) - H_i(1) \pm \varepsilon}{H_i(2)-H_i(1)} \right) = \trunc_{[\alpha_i(0),\alpha_i(2)]}\left( 2 + \frac{2p_0+p_1 \pm 2\varepsilon}{p_1+p_2} \right).
\end{equation}

\medskip

We are now ready to prove \cref{clm:base-gadget}. Consider a base gadget with input bidder $j$, output bidder $i$ and parameters $(\gamma_\ell, \gamma_r, \ell, r)$, such that $\gamma_\ell,\gamma_r \geq 1/20$, $\gamma_\ell + \gamma_r < 1$ and $\ell < r$. Let $p_b$ denote the probability that bidder $j$ bids $b$, as perceived by bidder $i$.

Assume first that bidder $j$ is almost-valid. Then, by the construction of $F_{i,j}$, we obtain that $p_0=p_3=p_4=0$, $p_1 \in [\gamma_\ell,1-\gamma_r]$ and $p_2 = 1 - p_1$. Using \eqref{eq:jump-0} we have that $\alpha_i(0) \leq 1 + \frac{\varepsilon}{p_1/2} \leq 1 + \frac{2\varepsilon}{\gamma_\ell} \leq 1 + 1/2$ since $\gamma_\ell \geq 4\varepsilon$. Using \eqref{eq:jump-3} we obtain that $\alpha_i(3) = 5$, since $p_3=p_4=0$ and $1 - \varepsilon > 0$. \eqref{eq:jump-2} yields that $\alpha_i(2) \geq 3 + \frac{1 + p_1 - 2\varepsilon}{1-p_1} \geq 3+1/2$, since $\varepsilon \leq 1/4$. Thus, in order to show that bidder $i$ is almost-valid, it remains to prove that $\alpha_i(1) \in [2,3]$. Using \eqref{eq:jump-1} we can write
$$\alpha_i(1) = \trunc_{[\alpha_i(0),\alpha_i(2)]}\left( 2 + \frac{2p_0 + p_1 \pm 2\varepsilon}{p_1+p_2} \right) = \trunc_{[\alpha_i(0),\alpha_i(2)]}\left( 2 + p_1 \pm 2\varepsilon \right) = 2 + p_1 \pm 2\varepsilon$$
where we used the fact that $p_1 + 2\varepsilon \leq 1$, since $p_1 \leq 1 - \gamma_r$ and $\gamma_r \geq 2\varepsilon$. Note that this also yields that $\alpha_i(1) \leq 3$, while the bound $\alpha_i(1) \geq 2$ holds because $p_1 \geq \gamma_\ell$ and $\gamma_\ell \geq 2\varepsilon$ (or simply because of the no-overbidding assumption). As a result, bidder $i$ is almost-valid.

Now consider the case where, in addition, $\gamma_\ell,\gamma_r \geq 1/3$. We can write
$$p_1 = \gamma_\ell + (1-\gamma_\ell-\gamma_r)\frac{\trunc_{[2+\ell,2+r]}(\alpha_j(1)) - (2+\ell)}{r-\ell}.$$
In particular, it holds that $p_1 \in [1/3,2/3]$. Since, as shown above, $\alpha_i(1) = 2 + p_1 \pm 2\varepsilon$, we immediately obtain that $\alpha_i(1) \in [2+1/3-2\varepsilon,2+2/3+2\varepsilon]$, i.e., bidder $i$ is valid. Furthermore, we can write
$$\val[i] = \trunc_{[0,1]}(3(\alpha_i(1)-2-1/3)) = 3p_1-1 \pm 6\varepsilon = (3\gamma_\ell-1) + 3(1-\gamma_\ell-\gamma_r)\frac{\trunc_{[2+\ell,2+r]}(\alpha_j(1)) - (2+\ell)}{r-\ell} \pm 6\varepsilon$$
which proves the claim.
\end{proof}

\medskip

\noindent\textbf{Projection Gadget.}
The projection gadget with input bidder $j$ and output bidder $i$, uses two additional auxiliary-bidders $k$ and $k'$, and consists of three uses of the standard base gadget. Concretely, the first standard base gadget has input $j$ and output $k$, the second such gadget has input $k$ and output $k'$, and the third has input $k'$ and output $i$. See \cref{fig:projection_gadget} for an illustration. As stated in the claim below, the projection gadget has the notable property that the output bidder $i$ is \emph{always} valid. This gadget will be used to ultimately ensure that all the gate-bidders are valid.

\begin{claim}\label{clm:projection-gadget}
The projection gadget with input bidder $j$ and output bidder $i$ ensures that:
\begin{itemize}
	\item the output bidder $i$ is valid, and
	\item if the input bidder $j$ is valid, then $\val[i] = \val[j] \pm 18\varepsilon$.
\end{itemize}
\end{claim}

\begin{proof}
The second point follows immediately from \cref{clm:base-gadget} applied to the standard base gate. Thus, it remains to show that the output bidder $i$ is always valid.
Consider the first standard base gadget, which has input bidder $j$ and output bidder $k$. Let $p_b$ denote the probability that bidder $j$ bids $b$, as perceived by bidder $k$. Since the density function of $F_{k,j}$ has a block of volume $1/3$ lying in $[1+1/2,1+3/4]$, and since we do not allow overbidding, it follows that $p_0+p_1 \geq 1/3$. Using \eqref{eq:jump-0} this implies that
$$\alpha_k(0) \leq 1 + \frac{p_0/\mm + \varepsilon}{p_0(\mm-1)/\mm + p_1/2} \leq 1 +  6/\mm  + 6\varepsilon \leq 1 + 1/2$$
since wlog $\mm \geq 24$ and $\varepsilon \leq 1/24$. Next, using \eqref{eq:jump-3} we immediately get that $\alpha_k(3) \geq 4$ since $\varepsilon < 1/3$ (or just by using the no-overbidding assumption). Then, \eqref{eq:jump-2} implies that
$$\alpha_k(2) = \trunc_{[\alpha_k(1),\alpha_k(3)]}\left( 3 + \frac{2p_0+2p_1+p_2 \pm 2\varepsilon}{p_2+p_3} \right) \geq 4 - 2\varepsilon \geq 3 + 1/2$$
where we used $p_0+p_1 \geq 1/3$, $p_2+p_3 \leq 2/3$, and $\varepsilon \leq 1/4$. Finally, note that $\alpha_k(1) \geq 2$ by the no-overbidding assumption.

Next, consider the second standard base gadget, which has input bidder $k$ and output bidder $k'$. Let $p_b$ denote the probability that bidder $k$ bids $b$, as perceived by bidder $k'$. From the construction of the density function of $F_{k',k}$ and the bounds obtained on the jump points of $k$ in the first step, it follows that $p_0=p_3=p_4=0$ and $p_1 \geq 1/3$. Using \cref{eq:jump-0,eq:jump-2,eq:jump-3} similarly to above, we obtain that $\alpha_{k'}(0) \leq 1+1/2$, $\alpha_{k'}(2) \geq 3+1/2$ and $\alpha_{k'}(3)=5$. As before, we have that $\alpha_{k'}(1) \geq 2$ by the no-overbidding assumption, and using \eqref{eq:jump-1} we also obtain that
$$\alpha_{k'}(1) = \trunc_{[\alpha_{k'}(0),\alpha_{k'}(2)]}\left( 2 + \frac{2p_0+p_1 \pm 2\varepsilon}{p_1+p_2} \right) \leq 2 + 1 + 2\varepsilon \leq 3 + 1/4$$
since $\varepsilon \leq 1/8$.

Finally, consider the third and last standard base gadget, which has input bidder $k'$ and output bidder $i$. Let $p_b$ denote the probability that bidder $k'$ bids $b$, as perceived by bidder $i$. From the construction of the density function of $F_{i,k'}$ and the bounds obtained on the jump points of $k'$ in the previous step, it follows that $p_0=p_3=p_4=0$, $p_1 \geq 1/3$ and $p_2 \geq 1/3$. Again using \cref{eq:jump-0,eq:jump-2,eq:jump-3} as in the previous step, we obtain that $\alpha_{i}(0) \leq 1+1/2$, $\alpha_{i}(2) \geq 3+1/2$ and $\alpha_{i}(3)=5$. Using \eqref{eq:jump-1} we have that
$$\alpha_i(1) = \trunc_{[\alpha_i(0),\alpha_i(2)]}\left( 2 + \frac{2p_0+p_1 \pm 2\varepsilon}{p_1+p_2} \right) = 2 + p_1 \pm 2\varepsilon \in [2+1/3-2\varepsilon,2+2/3+2\varepsilon]$$
and thus bidder $i$ is indeed valid.
\end{proof}

\medskip

\noindent\textbf{$\bm{G_{\times 2}}$ Gadget.} The $G_{\times 2}$ gadget with input bidder $j$ and output bidder $i$, uses an additional auxiliary-bidder $k$, and consists of one use of the base gadget and one use of the projection gadget. In more detail, the base gadget has input $j$, output $k$ and parameters $(\gamma_\ell,\gamma_r,\ell,r) = (1/3,1/3,1/3,1/2)$, while the projection gate has input $k$ and output $i$. See \cref{fig:x2-gadget} for an illustration.

\begin{claim}\label{clm:mult-2-gadget}
The $G_{\times 2}$ gadget with input bidder $j$ and output bidder $i$ ensures that:
\begin{itemize}
	\item the output bidder $i$ is valid, and
	\item if the input bidder $j$ is valid, then $\val[i] = \trunc(2\cdot \val[j]) \pm 24\varepsilon$.
\end{itemize}
\end{claim}

\begin{proof}
The fact that bidder $i$ is valid follows from our use of the projection gadget and the first bullet point in \cref{clm:projection-gadget}. Now consider the case where bidder $j$ is valid. Since $\gamma_\ell=\gamma_r=1/3$, by \cref{clm:base-gadget} we know that bidder $k$ is also valid and it holds that
$$\val[k] =  \frac{\trunc_{[2+\ell,2+r]}(\alpha_j(1)) - (2+\ell)}{r-\ell} \pm 6\varepsilon = \trunc_{[0,1]} \left( 6\alpha_j(1) - 14 \right) \pm 6\varepsilon = \trunc_{[0,1]}(2\cdot \val[j]) \pm 6\varepsilon.$$
Since $k$ is valid, we can use the second bullet point in \cref{clm:projection-gadget}, which yields $\val[i] = \val[k] \pm 18\varepsilon = \trunc_{[0,1]}(2\cdot \val[j]) \pm 24\varepsilon$.
\end{proof}

\medskip

\noindent\textbf{$\bm{G_{1-}}$ Gadget.} The $G_{1-}$ gadget with input bidder $j$ and output bidder $i$ uses three additional auxiliary-bidders $k_1,k_2,k_3$. First, a base gadget is used with input $j$, output $k_1$ and parameters $(\gamma_\ell,\gamma_r,\ell,r) = (1/6,2/3,1/3,2/3)$. Next, the density function of $F_{k_2,k_1}$ has a block of volume $2/3$ in $[1+1/2,1+3/4]$, and a block of volume $1/3$ in $[4,5]$. Then, we use a base gadget with input $k_2$, output $k_3$ and parameters $(\gamma_\ell,\gamma_r,\ell,r) = (1/3,1/3,2/3,5/6)$. Finally, we use a projection gadget with input $k_3$ and output $i$. See \cref{fig:oneminus-gadget} for an illustration.

The crucial idea behind the construction of this gadget is that the third jump point (instead of the second one) is used to encode information in some intermediate step. This allows us to simulate the non-monotone operation $x \mapsto 1-x$.

\begin{claim}\label{clm:complement-gadget}
The $G_{1-}$ gadget with input bidder $j$ and output bidder $i$ ensures that:
\begin{itemize}
	\item the output bidder $i$ is valid, and
	\item if the input bidder $j$ is valid, then $\val[i] = 1 - \val[j] \pm 60\varepsilon$.
\end{itemize}
\end{claim}

\begin{proof}
First of all, note that $i$ must be valid, because of the corresponding property of the projection gadget (\cref{clm:projection-gadget}). Now consider the case where $j$ is valid. By \cref{clm:base-gadget} it follows that bidder $k_1$ is almost-valid, in particular $\alpha_{k_1}(3)=5$ and $\alpha_{k_1}(1) \leq 3$. Let $p_b$ denote the probability that bidder $j$ bids $b$, as perceived by bidder $k_1$. Since $j$ is valid, we immediately obtain that $p_0=p_3=p_4=0$. Furthermore, by the construction of $F_{k_1,j}$, it is easy to see that $p_1 = 1/6 + (1-1/6-2/3) \val[j] = 1/6 + \val[j]/6$. Next, using \eqref{eq:jump-2} we can write
\begin{equation*}
\begin{split}
\alpha_{k_1}(2) = \trunc_{[\alpha_{k_1}(1),\alpha_{k_1}(3)]}\left( 3 + \frac{2p_0+2p_1+p_2 \pm 2\varepsilon}{p_2+p_3} \right) &= \trunc_{[\alpha_{k_1}(1),\alpha_{k_1}(3)]} \left(3 + \frac{1+p_1 \pm 2\varepsilon}{1-p_1}\right)\\
&= 3 + \frac{7/6 + \val[j]/6}{5/6- \val[j]/6} \pm 3\varepsilon\\
&= 4 + \frac{2 + 2\val[j]}{5- \val[j]} \pm 3\varepsilon.
\end{split}
\end{equation*}
Now consider bidder $k_2$. Let $p_b$ denote the probability that bidder $k_1$ bids $b$, as perceived by bidder $k_2$. By construction of $F_{k_2,k_1}$ and since $k_1$ is almost-valid, it is easy to see that $p_0=p_4=0$, $p_1 = 2/3$ and $p_2+p_3 = 1/3$. By the same arguments used in the proof of \cref{clm:base-gadget} it follows that $\alpha_{k_2}(0) \leq 1+1/2$. By using \eqref{eq:jump-3} we obtain $\alpha_{k_2}(3) \geq 4 + \frac{2/3+1/6-\varepsilon}{1/6} \geq 5$. Next, using \eqref{eq:jump-2} we obtain
$$\alpha_{k_2}(2) \geq \trunc_{[\alpha_{k_2}(1),\alpha_{k_2}(3)]}\left( 3 + \frac{2p_0+2p_1+p_2 \pm 2\varepsilon}{p_2+p_3} \right) \geq \trunc_{[\alpha_{k_2}(1),5]}\left(3 + \frac{4/3 - 2\varepsilon}{1/3}\right) = 5.$$
Now observe that by construction of $F_{k_2,k_1}$ and the expression obtained earlier for $\alpha_{k_1}(2)$
$$p_2 = \frac{\trunc_{[4,5]}(\alpha_{k_1}(2))-4}{3} = \frac{2 + 2\val[j]}{15- 3\val[j]} \pm \varepsilon.$$
As a result, it follows that
$$\frac{2p_0+p_1}{p_1+p_2} = \frac{2/3}{2/3+\frac{2 + 2\val[j]}{15- 3\val[j]} \pm \varepsilon} = \frac{2/3}{2/3+\frac{2 + 2\val[j]}{15- 3\val[j]}} \pm 3\varepsilon = 5/6 - \val[j]/6 \pm 3\varepsilon$$
where we used $\varepsilon \leq 2/15$. Finally, using \eqref{eq:jump-1} we obtain
\begin{equation*}\begin{split}
\alpha_{k_2}(1) = \trunc_{[\alpha_{k_2}(0),\alpha_{k_2}(2)]}\left( 2 + \frac{2p_0+p_1 \pm 2\varepsilon}{p_1+p_2} \right) &= \trunc_{[\alpha_{k_2}(0),\alpha_{k_2}(2)]}\left(2 + 5/6 - \val[j]/6 \pm 3\varepsilon \pm \frac{2\varepsilon}{p_1+p_2} \right)\\
&= 2 + 5/6 - \val[j]/6 \pm 6\varepsilon.
\end{split}\end{equation*}
Note in particular that bidder $k_2$ is almost-valid, since $\varepsilon \leq 1/36$.

Since bidder $k_2$ is almost-valid, and we use a base gadget with $\gamma_\ell=\gamma_r=1/3$ with input $k_2$ and output $k_3$, it follows by \cref{clm:base-gadget} that bidder $k_3$ is valid and
$$\val[k_3] =  \frac{\trunc_{[2+\ell,2+r]}(\alpha_{k_2}(1)) - (2+4/6)}{1/6} \pm 6\varepsilon = 1 - \val[j] \pm 42\varepsilon.$$
Finally, the projection gadget with input $k_3$ and output $i$ ensures that $\val[i] = \val[k_3] \pm 18\varepsilon = 1 - \val[j] \pm 60\varepsilon$.
\end{proof}

\medskip

\noindent\textbf{$\bm{G_\phi}$ Gadget.} The $G_\phi$ gadget with input bidders $j_1$ and $j_2$ and output bidder $i$ is a binary gadget with additional auxiliary-bidders $k_1,k_2,k_3$. First of all, for all $t \in N \setminus \{j_1,j_2,k_1\}$, we set $F_{k_1,t}$ to have density function with a single block of volume $1$ in $[0,1]$. We set \emph{both} $F_{k_1,j_1}$ and $F_{k_1,j_2}$ to be distributions as in our construction of the base gadget with parameters $(\gamma_\ell,\gamma_r,\ell,r) = (1/20,8/20,1/3,2/3)$. The density function of $F_{k_2,k_1}$ has a block of volume $1/2$ in $[1+1/2,1+3/4]$, and a block of volume $1/2$ in $[3+1/2, 5]$. Next, we use a base gadget with input $k_2$, output $k_3$ and parameters $(\gamma_\ell,\gamma_r,\ell,r) = (1/3,1/3(1+1/4),104/200,779/800)$. Finally, we use a $G_{1-}$ gadget with input $k_3$ and output $i$. See \cref{fig:phi-gadget} for an illustration. We have the following claim.

\begin{figure}
     \centering
     \fbox{%
    \colorbox{blue!5!white}{
    \parbox{0.965\textwidth}{%
     \begin{subfigure}[t]{0.95\textwidth}
         \centering
         \input{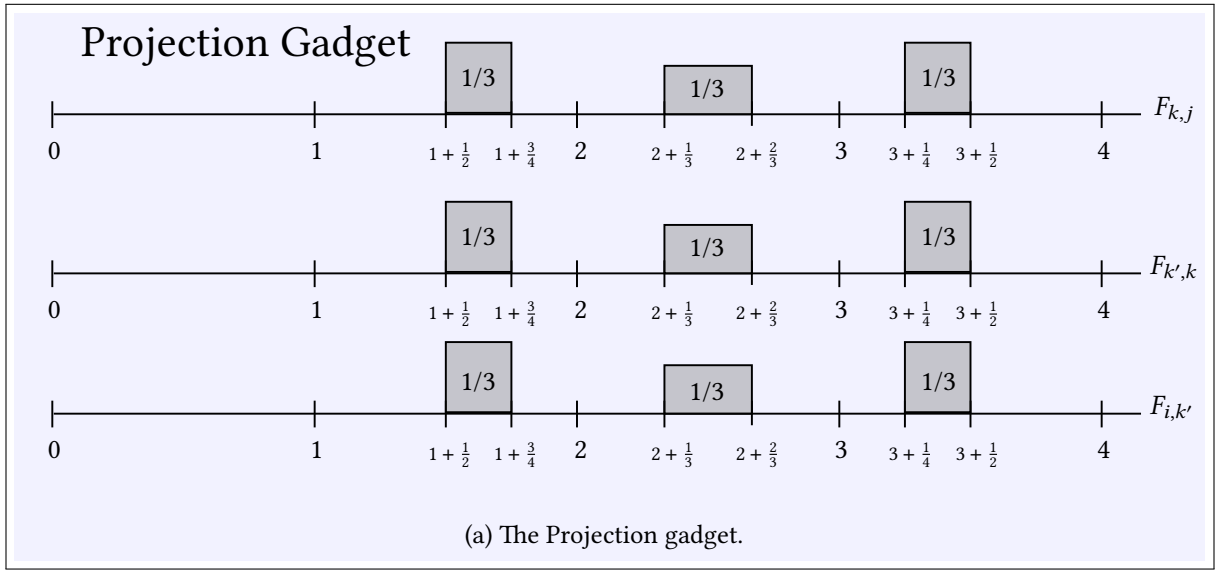}
         \caption{The Projection gadget.}
        \label{fig:projection_gadget}
     \end{subfigure}}}}\vspace{0.2cm}
     \hfill
     \centering
     \fbox{%
    \colorbox{blue!5!white}{
    \parbox{0.965\textwidth}{%
     \begin{subfigure}[t]{0.95\textwidth}
         \centering
        \scalebox{0.97}{
\tikzset{every picture/.style={line width=0.75pt}} 

\begin{tikzpicture}[x=0.75pt,y=0.75pt,yscale=-0.75,xscale=0.75,every node/.style={scale=1}]
\draw    (54,119) -- (800,119) ;
\draw    (53,110) -- (53,128.5) ;
\draw    (233,110) -- (233,128.5) ;
\draw    (413,110) -- (413,128.5) ;
\draw    (593,110) -- (593,128.5) ;
\draw    (773,110) -- (773,128.5) ;
\draw    (323,110) -- (323,128.5) ;
\draw    (368,110) -- (368,128.5) ;
\draw  [color={rgb, 255:red, 0; green, 0; blue, 0 }  ,draw opacity=1 ][fill={rgb, 255:red, 155; green, 155; blue, 155 }  ,fill opacity=0.51 ] (323,69.5) -- (368,69.5) -- (368,118) -- (323,118) -- cycle ;
\draw    (473,110) -- (473,128.5) ;
\draw    (503,110) -- (503,127) ;
\draw  [color={rgb, 255:red, 0; green, 0; blue, 0 }  ,draw opacity=1 ][fill={rgb, 255:red, 155; green, 155; blue, 155 }  ,fill opacity=0.51 ] (473,45.5) -- (503,45.5) -- (503,119.5) -- (473,119.5) -- cycle ;
\draw    (638,110) -- (638,128.5) ;
\draw    (683,110) -- (683,128.5) ;
\draw  [color={rgb, 255:red, 0; green, 0; blue, 0 }  ,draw opacity=1 ][fill={rgb, 255:red, 155; green, 155; blue, 155 }  ,fill opacity=0.51 ] (638,69.5) -- (683,69.5) -- (683,118) -- (638,118) -- cycle ;
\draw  [fill={rgb, 255:red, 155; green, 155; blue, 155 }  ,fill opacity=0.61 ] (48,176.5) -- (831,176.5) -- (831,219.5) -- (48,219.5) -- cycle ;

\draw (48,136) node [anchor=north west][inner sep=0.75pt]   [align=left] {0};
\draw (228,136) node [anchor=north west][inner sep=0.75pt]   [align=left] {1};
\draw (409,136) node [anchor=north west][inner sep=0.75pt]   [align=left] {2};
\draw (588,136) node [anchor=north west][inner sep=0.75pt]   [align=left] {3};
\draw (768,136) node [anchor=north west][inner sep=0.75pt]   [align=left] {4};
\draw (304,136) node [anchor=north west][inner sep=0.75pt]   [align=left] {};
\draw (309,136.4) node [anchor=north west][inner sep=0.75pt]  [font=\scriptsize]  {${\textstyle 1+\frac{1}{2}}$};
\draw (454,136.4) node [anchor=north west][inner sep=0.75pt]  [font=\scriptsize]  {${\textstyle 2+\frac{1}{3}}$};
\draw (354,136.4) node [anchor=north west][inner sep=0.75pt]  [font=\scriptsize]  {${\textstyle 1+\frac{3}{4}}$};
\draw (502,136.4) node [anchor=north west][inner sep=0.75pt]  [font=\scriptsize]  {${\textstyle 2+\frac{1}{2}}$};
\draw (807,104.4) node [anchor=north west][inner sep=0.75pt]    {${F_{k,j}}$};
\draw (61,185) node [anchor=north west][inner sep=0.75pt]   [align=left] {{\large Projection Gadget}};
\draw (625,136.4) node [anchor=north west][inner sep=0.75pt]  [font=\scriptsize]  {${\textstyle 3+\frac{1}{4}}$};
\draw (671,136.4) node [anchor=north west][inner sep=0.75pt]  [font=\scriptsize]  {${\textstyle 3+\frac{1}{2}}$};
\draw (300,187) node [anchor=north west][inner sep=0.75pt]   [align=left] {Input: $k$};
\draw (400,188) node [anchor=north west][inner sep=0.75pt]   [align=left] {Output: $i$};
\draw (71,55) node [anchor=north west][inner sep=0.75pt]   [align=left] {{\LARGE $G_{\times 2}$\ \  Gadget}};

\draw (331,85) node [anchor=north west][inner sep=0.75pt]   [align=left] {1/3};
\draw (474,76) node [anchor=north west][inner sep=0.75pt]   [align=left] {1/3};
\draw (647,85) node [anchor=north west][inner sep=0.75pt]   [align=left] {1/3};
\end{tikzpicture}}
         \caption{The $G_{\times 2}$ gadget.}
         \label{fig:x2-gadget}
     \end{subfigure}}}}\vspace{0.2cm}
     \hfill
     \centering
     \fbox{%
    \colorbox{blue!5!white}{
    \parbox{0.965\textwidth}{%
     \begin{subfigure}[t]{0.95\textwidth}
         \centering
         \input{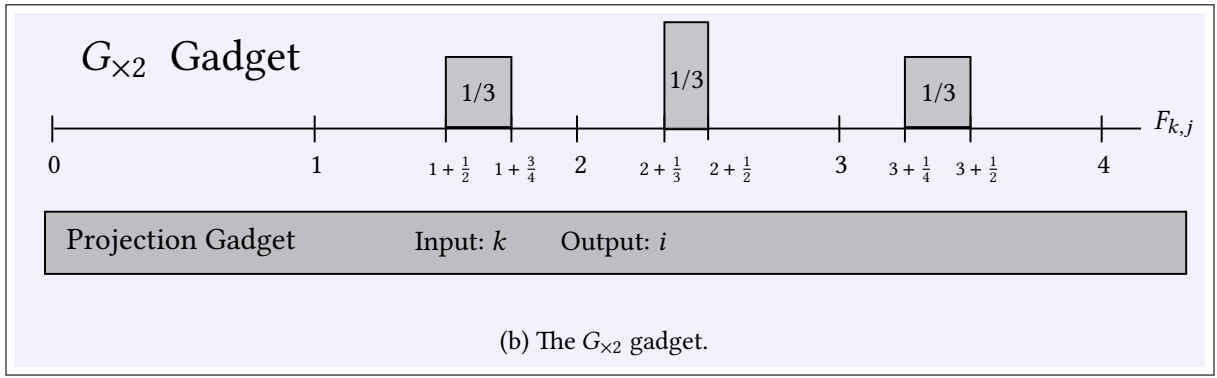}
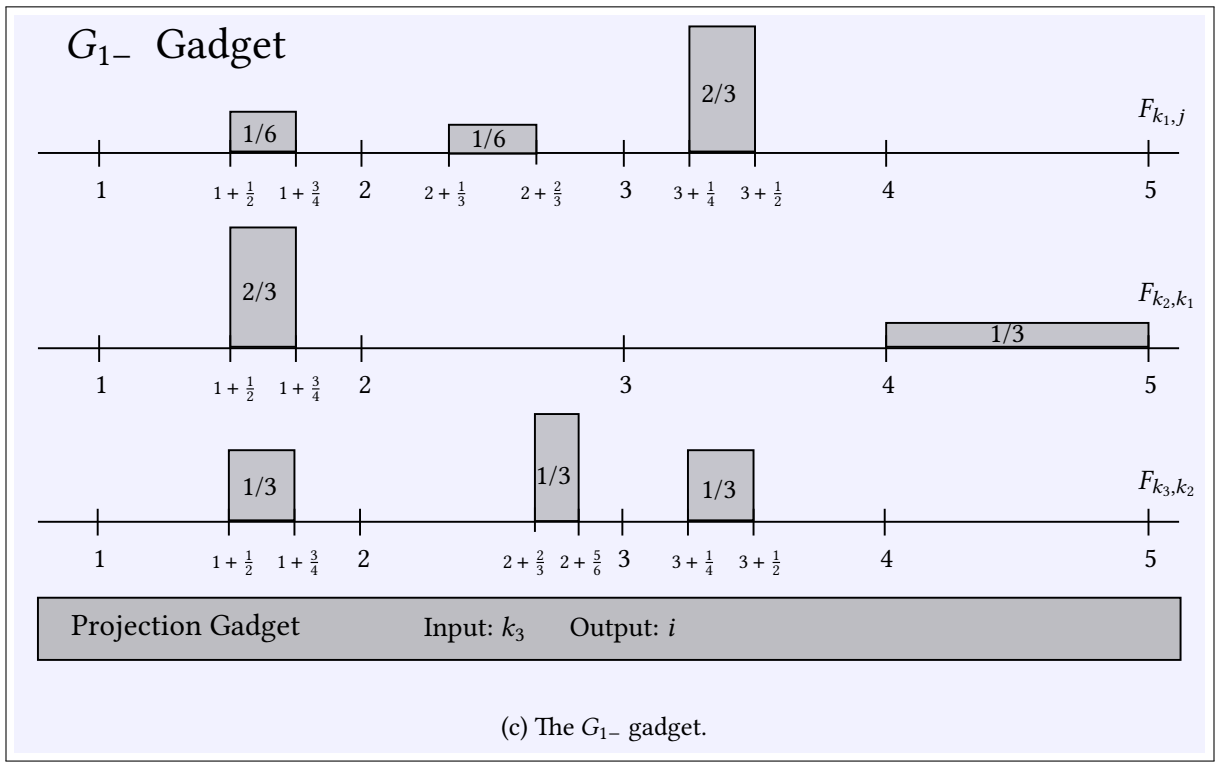
         \caption{The $G_{1-}$ gadget.}
        \label{fig:oneminus-gadget}
         \end{subfigure}}}}
         \caption{The Projection, $G_{\times 2}$ and $G_{1-}$ gadgets. The probability density functions of the corresponding subjective priors are shown.}
\end{figure}

\begin{figure}
\fbox{%
    \colorbox{blue!5!white}{
    \parbox{0.965\textwidth}{%
    \centering
    \input{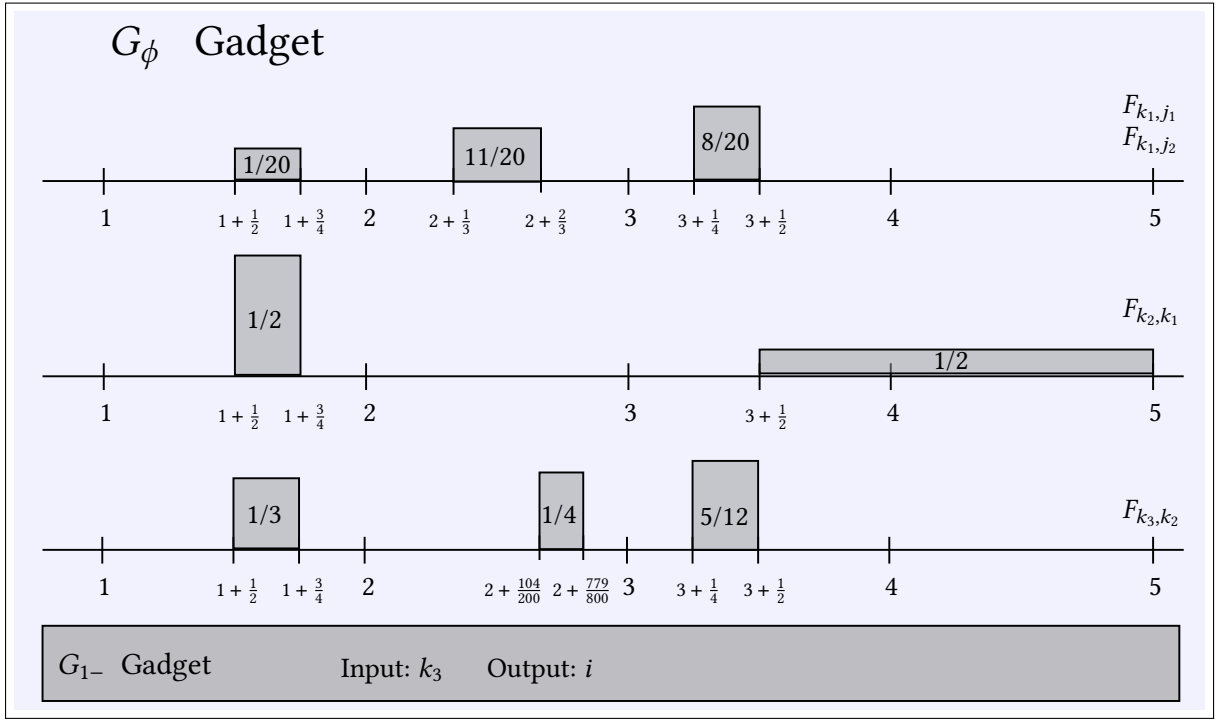}
    }}}
    \caption{The $G_\phi$ gadget. The probability density functions of the corresponding subjective priors are shown.}
    \label{fig:phi-gadget}
\end{figure}

\begin{claim}\label{clm:phi-gadget}
The $G_\phi$ gadget with input bidders $j_1,j_2$ and output bidder $i$ ensures that:
\begin{itemize}
	\item the output bidder $i$ is valid, and
	\item if the input bidders $j_1$ and $j_2$ are valid, then
	$$\val[i] = \phi(\val[j_1],\val[j_2]) \pm 86\varepsilon = \frac{1}{4}(\val[j_1]+1)(\val[j_2]+1) \pm 86\varepsilon.$$
\end{itemize}
\end{claim}

\begin{proof}
Bidder $i$ is guaranteed to be valid, because it is the output bidder of the $G_{1-}$ gadget (\cref{clm:complement-gadget}). Now assume that $j_1$ and $j_2$ are valid. Let $p_b$ denote the probability that bidder $j_1$ bids $b$, as perceived by bidder $k_1$. Similarly, let $q_b$ denote the probability that bidder $j_2$ bids $b$, as perceived by bidder $k_1$. By construction of $F_{k_1,j_1}$ and $F_{k_1,j_2}$, and because $j_1$ and $j_2$ are valid, we know that $p_0=p_3=p_4=q_0=q_3=q_4=0$, $p_1,q_1 \geq 1/20$ and $p_2,q_2 \geq 8/20$. Recall that $H_{k_1}(b)$ is used to denote the probability that bidder $k_1$ wins if she bids $b$ (from $k_1$'s perspective). Thus we immediately obtain that $H_{k_1}(0)=0$, $H_{k_1}(1)=p_1q_1/3$, $H_{k_1}(2)=p_1q_1 + p_2q_1/2 + p_1q_2/2 + p_2q_2/3 = 1/3 + (p_1+q_1)/6 + p_1q_1/3$ and $H_{k_1}(3) = H_{k_1}(4) = 1$. With this in hand, we now obtain (just as we did for \cref{eq:jump-0,eq:jump-3,eq:jump-2,eq:jump-1}):
$$\alpha_{k_1}(0) \leq 1 + \frac{H_{k_1}(0) + \varepsilon}{H_{k_1}(1)-H_{k_1}(0)} = 1 + \frac{\varepsilon}{p_1q_1/3} \leq 1 + 1200 \varepsilon \leq 1 + 1/2$$
since $\varepsilon \leq 1/2400$. Similarly, since $H_{k_1}(4)-H_{k_1}(3)=0$ and $H_{k_1}(3) = 1 > \varepsilon$, we have that $\alpha_{k_1}(3)=5$. We also have
$$\alpha_{k_1}(1) \leq 2 + \frac{H_{k_1}(1) + \varepsilon}{H_{k_1}(2)-H_{k_1}(1)} \leq 2 + \frac{p_1q_1/3 + \varepsilon}{1/3 + (p_1+q_1)/6} \leq 3$$
where we used the bounds we have on these probabilities and $\varepsilon \leq 1/4$. Finally, we have
\begin{equation*}\begin{split}
\alpha_{k_1}(2) = \trunc_{[\alpha_{k_1}(1),\alpha_{k_1}(3)]} \left( 3 + \frac{H_{k_1}(2) \pm \varepsilon}{H_{k_1}(3) - H_{k_1}(2)} \right) &= \trunc_{[3,5]} \left(3 + \frac{1/3 + (p_1+q_1)/6 + p_1q_1/3 \pm \varepsilon}{1-(1/3 + (p_1+q_1)/6 + p_1q_1/3)}\right)\\
&= 3 + \frac{1 + (p_1+q_1)/2 + p_1q_1}{2-(p_1+q_1)/2 - p_1q_1} \pm 3\varepsilon\\
&= 3 + \frac{1}{2} + \frac{3}{2}\frac{(p_1+q_1)/2 + p_1q_1}{2-(p_1+q_1)/2 - p_1q_1} \pm 3\varepsilon
\end{split}\end{equation*}
where we used the fact that $\frac{(p_1+q_1)/2 + p_1q_1}{2-(p_1+q_1)/2 - p_1q_1} \leq 1$, since $p_1,q_1 \leq 12/20$. As $p_1,q_1 \geq 1/20$ and $\varepsilon \leq 1/60$, we also have that $\alpha_{k_1}(2) \geq 3+1/2$. In particular, $k_1$ is almost-valid. Note that since $j_1$ and $j_2$ are valid, we have $p_1=1/20 + 11 \val[j_1]/20$ and $q_1=1/20 + 11 \val[j_2]/20$.

Next, we consider bidder $k_2$. Let $p_b'$ denote the probability that bidder $k_1$ bids $b$, as perceived by bidder $k_2$. By the previous paragraph, we have $p_0'=0$, $p_1'=1/2$, $p_4'=0$ and
$$p_2' = \frac{1}{3}\frac{3}{2}\frac{(p_1+q_1)/2 + p_1q_1}{2-(p_1+q_1)/2 - p_1q_1} \pm 3\varepsilon = \frac{1}{2}\frac{(p_1+q_1)/2 + p_1q_1}{2-(p_1+q_1)/2 - p_1q_1} \pm 3\varepsilon$$
where we used the fact that the height of the block of volume of $F_{k_2,k_1}$ in $[3+1/2,5]$ is $1/3$.
Since the density function of $F_{k_2,k_1}$ has a block of volume $1/2$ in $[1+1/2,1+3/4]$, as before we obtain that $\alpha_{k_2}(0) \leq 1+1/2$. Using \eqref{eq:jump-3} and \eqref{eq:jump-2}, we also have
$$\alpha_{k_2}(3) \geq 4 + \frac{p_0'+p_1'+p_2'+p_3'/2 - \varepsilon}{p_3'/2+p_4'/2} \geq 5$$
as well as
$$\alpha_{k_2}(2) \geq \trunc_{[\alpha_{k_2}(1),\alpha_{k_2}(3)]}\left( 3 + \frac{2p_0'+2p_1'+p_2' \pm 2\varepsilon}{p_2'+p_3'} \right) \geq 3+1/2.$$
Finally, \eqref{eq:jump-1} yields
\begin{equation*}\begin{split}
\alpha_{k_2}(1) &= \trunc_{[\alpha_{k_1}(0),\alpha_{k_1}(2)]}\left( 2 + \frac{2p_0'+p_1' \pm 2\varepsilon}{p_1'+p_2'} \right) \\ &= \trunc_{[\alpha_{k_1}(0),\alpha_{k_1}(2)]}\left( 2 + \frac{1/2}{1/2+\frac{1}{2}\frac{(p_1+q_1)/2 + p_1q_1}{2-(p_1+q_1)/2 - p_1q_1} \pm 3\varepsilon} \right) \pm 4\varepsilon\\
&=  2 + \frac{2 - (p_1+q_1)/2 - p_1q_1}{2} \pm 10\varepsilon.
\end{split}\end{equation*}
Substituting in $p_1=1/20 + 11 \val[j_1]/20$ and $q_1=1/20 + 11 \val[j_2]/20$, we compute
\begin{equation*}\begin{split}
\alpha_{k_2}(1) &= 2 + 1+1/8-\frac{1}{2}(11/20+11 \val[j_1]/20)(11/20+ 11 \val[j_2]/20) \pm 10\varepsilon\\
&= 2+9/8-\frac{121}{200} \phi(\val[j_1],\val[j_2]) \pm 10\varepsilon.
\end{split}\end{equation*}
Note that we have $\alpha_{k_2}(1) \in [2+104/200,2+779/800] \pm 10\varepsilon$. In particular, bidder $k_2$ is almost-valid.

Since bidder $k_3$ is the output of a base gadget with input $k_2$ and parameters $(\gamma_\ell,\gamma_r,\ell,r) = (1/3,1/3(1+1/4),104/200,779/800)$, it follows by \cref{clm:base-gadget} that $k_3$ is valid and
\begin{equation*}\begin{split}
\val[k_3] &= 3(1-\gamma_\ell-\gamma_r) \frac{\trunc_{[2+\ell,2+r]}(\alpha_{k_2}(1)) - (2+\ell)}{r-\ell} \pm 6\varepsilon\\
&= \frac{200}{121} (\alpha_{k_2}(1) - (2+\ell)) \pm 6\varepsilon\\
&= \frac{200}{121} \left(2+9/8-\frac{121}{200} \phi(\val[j_1],\val[j_2]) - (2+104/200)\right) \pm 26\varepsilon\\
&= 1 - \phi(\val[j_1],\val[j_2]) \pm 26\varepsilon.
\end{split}\end{equation*}
Finally, it is easy to see that the $G_{1-}$ gadget with input $k_3$ and output $i$ ensures the desired value for bidder $i$ (\cref{clm:complement-gadget}).
\end{proof}

\noindent\textbf{Finishing the proof.} Using the gadgets we have described above we can now enforce the constraints of the \gcircuit instance. Indeed, for each gate $g_i=(G,j,k)$ where $G \in \mathcal{G} = \{G_{\times 2}, G_{1-}, G_{\phi}\}$, it suffices to use the gadget corresponding to the gate-type $G$, with output bidder $i$ and input bidder $j$ (as well as $k$, in the case $G=G_\phi$). Since the distributions are subjective, we can re-use a bidder $j$ as an input to multiple different gadgets, without any interference. By \cref{clm:mult-2-gadget,clm:complement-gadget,clm:phi-gadget} it immediately follows that the gate-bidders $1,2,\dots, \nn$ must all be valid, since each of them is the output of some gadget. But this means that for any gate $g_i=(G,j,k)$, the input bidder $j$ (and $k$, if applicable) will be valid, because she is also a gate-bidder. As a result, again by \cref{clm:mult-2-gadget,clm:complement-gadget,clm:phi-gadget}, it follows that the gadgets will correctly enforce their constraints on all values $\val[i]$.

To obtain a solution, it suffices to set $\val[g_i] := \val[i]$ for all $i \in [\nn]$. For the case $\varepsilon=0$, note that since every gate-bidder $i$ is valid, we have that $\alpha_i(1) \in [2+1/3,2+2/3]$ and as a result $\val[i] = \trunc_{[0,1]} (3(\alpha_i(1)-2-1/3)) = 3(\alpha_i(1)-2-1/3)$, which indeed yields an SL-reduction \citep{etessami2010complexity}. By scaling back to the original value space $[0,1]$, the proof yields that for all $\varepsilon \in [0,1/10^5]$, from any $\varepsilon$-BNE of the auction we can extract an $500\varepsilon$-satisfying assignment for the generalized circuit. As discussed at the beginning of the section, this yields both \ppad- and \fixp-hardness.

\section{An Efficient Algorithm for a Constant Number of Bidders and Bids}\label{sec:positive}
In this section, we design an algorithm which computes an $\varepsilon$-Bayes-Nash equilibrium of the FPA when (a) the number of bidders $n$ is constant, (b) the size of the bidding space $|B|$ is constant, and (c) the value distributions $F_{i,j}$ of the bidders are \emph{piecewise polynomial}.\\ 

\noindent To be more precise, our input comprises of:
\begin{itemize}
    \item[-] a set of bids\footnote{Recall that here $\card{B}$ is \emph{fixed}, i.e., not part of the input.} $B=\ssets{b_0,b_1,\dots,b_{\card{B}-1}} \subset [0,1]$
    \item[-] a partition\footnote{Our assumption here of a \emph{common} interval partition for the piecewise polynomial representation of all subjective priors $F_{i,j}$ is for the sake of simplicity, and it is not critical for the positive results of this section. In particular, it is not difficult to see that our model can handle different partitions $[x_{\ell-1}^{i,j}, x_\ell^{i,j}]$ with just a polynomial blow-up in the size of the representation; essentially one needs to take the interval partition induced by all points $\ssets{x_\ell^{i,j}}$.} of $[0,1]$ into $K$ intervals $[x_{\ell-1}, x_\ell]$, $\ell = \{1, 2, \ldots, K\}$, with rational endpoints
    \item[-] for each distribution $F_{i,j}$ and each subinterval $[x_{\ell-1}, x_\ell]$, a vector of rationals $(a_0^{i,j,\ell},a_1^{i,j,\ell},\dots,a_d^{i,j,\ell})$. 
\end{itemize}
Then, (the cumulative distribution function of) $F_{i,j}$ is defined as
$$
F_{i,j}(z)= F_{i,j}^\ell(z),
\qquad\text{for}\;\; z\in[x_{\ell-1},x_{\ell}],
$$
where 
\begin{equation}
\label{eq:piecewise_poly_algorithm_representation}
F_{i,j}^\ell(z)=\sum_{\kappa=0}^d a_\kappa^{i,j,\ell} z^\kappa
\end{equation}
is the polynomial representation of $F_{i,j}$ in the $\ell$-th interval. Of course, the input should respect the conditions
$$
F^1_{i,j}(0) \geq 0,\quad F^{K}_{i,j}(1)= 1,\quad F^\ell_{i,j}(x_\ell)=F^{\ell+1}_{i,j}(x_\ell)\;\;\text{for}\;\; \ell=1,2,\dots,K-1,
$$
and that each $F^{\ell}_{i,j}$ is nondecreasing on $[x_{\ell-1},x_\ell]$.

Finally, when we say that $n$ and $|B|$ are fixed, we mean that they are constant functions of the other parameters of the input. \medskip

\noindent We have the following theorem.

\begin{theorem}\label{thm:positive}
For a fixed number of bidders, a fixed bidding space, and piecewise polynomial value distributions, an $\varepsilon$-BNE of the first-price auction can be computed in polynomial time, even for subjective priors and even when $\varepsilon$ is inversely-exponential in the input size.
\end{theorem}

\noindent The remainder of the section is devoted to developing the algorithm that will prove \cref{thm:positive}. 

At a high level, the algorithm will perform the following four steps:\medskip 
\begin{enumerate}
    \item It ``guesses'', for each bidder, an assignment of the jump points of her best-response strategy to the $K$ sub-intervals $[x_{\ell-1},x_\ell]$ above; intervals may be allocated zero or multiple jump points. Since the number of bidders and the size of the bidding space are constant, there is a total constant number of jump points for all bidders. Therefore, this ``guessing'' step is an enumeration of all such possible assignments; the subsequent steps of the algorithm are run for any such assignment.\medskip
    \item It ``guesses'' a set of \emph{effective} jump points and bids. This is a technical corner case, to eliminate degenerate cases in which multiple jump points coincide. Again, this can be done via enumeration given that the number of jump points is constant. \medskip
    \item It formulates the problem of finding the \emph{exact positions} of the effective jump points (within the intervals corresponding to the guessed allocation above) as a system of polynomial inequalities of polynomially-large degree. A $\delta$-approximate solution to this system can be found using standard methods, in time polynomial in $\log(1/\delta)$ and the input parameters.\medskip
    \item It ``projects'' the approximate solution to the ``equilibrium space'', as defined by the constraints of the aforementioned system, ensuring that the resulting object is indeed an $\varepsilon$-BNE, for some $\varepsilon$ that can be made as small as needed, by making $\delta$ as small as needed.
\end{enumerate}

\noindent Below we describe these steps in more detail.

\subsubsection*{Step 1: Guessing an allocation of jump points to intervals}

Recall the definition of the jump points $\alpha_i(b)$ from \cref{sec:prelims}, which represent the equilibrium strategy of bidder $i$. Intuitively, $\alpha_i(b)$ is the largest value for which bidder $i$ would bid $b$ or lower. Since $|B|$ is constant, there is a constant number of such jump points for each bidder, and since $n$ is also constant, there is a constant number of jump points overall. The algorithm enumerates over all the possible ways of assigning the $n \cdot (|B|-1)$ jump points to the intervals $[x_{\ell-1}, x_\ell]$, for $\ell=1, \ldots, K$; this can be done in time $O(K^{n|B|})$. Then, for any possible such allocation, it moves to the next step. We introduce variables $y_{i,j}$, $j=1,2,\dots,\card{B}-1$ for the positions of the jump points of the strategy of bidder $i$ in $[0,1]$, and we set $y_{i,0}=0$, $y_{i,\card{B}}=1$.

\subsubsection*{Step 2: Guessing a set of effective jump points and bids} 

We ``guess'' possible ``collisions'' of sequential jump points, where a collision happens when the positions of two or more jump points coincide. In that case, we would like to only keep a \emph{single} representative from each coinciding jump point; the positions of these representatives are denoted using the variables $z_{ij}$. We also use the variables $b'_{i,j}$ to denote the corresponding bids, as subscribed by the chosen jump points. We refer to the chosen jump points and bids as \emph{effective} jump points and bids respectively. See \cref{fig:effective-points} for an illustration. \medskip

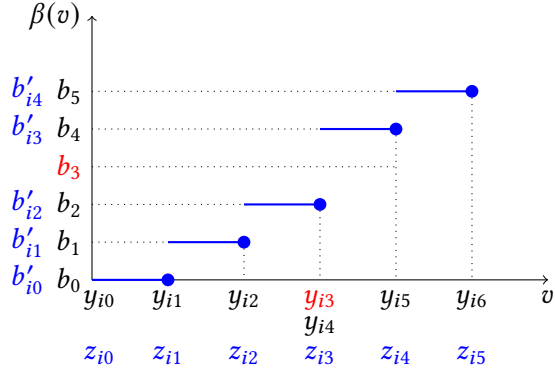
\begin{figure}
\centering
\begin{tikzpicture}
\draw[->] (0,0) -- (6,0);
\draw[->] (0,0) -- (0,3.5);
\node[below] at (6,0) {$v$};
\node[below] at (0.1,0) {$y_{i0}$};
\node[below] at (0.1,-0.75) {{\color{blue}$z_{i0}$}};
\node[below] at (1,0) {$y_{i1}$};
\node[below] at (1,-0.75) {{\color{blue}$z_{i1}$}};
\node[below] at (3,0) {{\color{red} $y_{i3}$}};
\node[below] at (3,-0.35) {$y_{i4}$};
\node[below] at (3,-0.75) {{\color{blue}$z_{i3}$}};
\node[below] at (2,0) {$y_{i2}$};
\node[below] at (2,-0.75) {{\color{blue}$z_{i2}$}};
\node[below] at (4,0) {$y_{i5}$};
\node[below] at (4,-0.75) {{\color{blue}$z_{i4}$}};
\node[below] at (5,0) {$y_{i6}$};
\node[below] at (5,-0.75) {{\color{blue}$z_{i5}$}};
\node[left] at (0,3.5) {$\beta(v)$};
\node[left] at (0,0) {$b_0$};
\node[left] at (-0.5,0) {{\color{blue}$b'_{i0}$}};
\node[left] at (0,0.5) {$b_1$};
\node[left] at (-0.5,0.5) {{\color{blue}$b'_{i1}$}};
\node[left] at (0,1.0) {$b_2$};
\node[left] at (-0.5,1.0) {{\color{blue}$b'_{i2}$}};
\node[left] at (0,1.5) {{\color{red} $b_3$}};
\node[left] at (0,2.0) {$b_4$};
\node[left] at (-0.5,2.0) {{\color{blue}$b'_{i3}$}};
\node[left] at (0,2.5) {$b_5$};
\node[left] at (-0.5,2.5) {{\color{blue}$b'_{i4}$}};
\coordinate (b1l) at (0,0.0) {};
\coordinate (b1r) at (1,0.0) {};
\coordinate (b2l) at (1,0.5) {};
\coordinate (b2r) at (2,0.5) {};
\coordinate (b3l) at (2,1.0) {};
\coordinate (b3r) at (3,1.0) {};
\coordinate (b4l) at (3,1.5) {};
\coordinate (b4r) at (4,1.5) {};
\coordinate (b5l) at (3,2.0) {};
\coordinate (b5r) at (4,2.0) {};
\coordinate (b6l) at (4,2.5) {};
\coordinate (b6r) at (5,2.5) {};
\draw[blue,thick] (b1l) -- (b1r);
\draw[blue,thick] (b2l) -- (b2r);
\draw[blue,thick] (b3l) -- (b3r);
\draw[blue,thick] (b5l) -- (b5r);
\draw[blue,thick] (b6l) -- (b6r);
\draw[dotted] (0,0.5) -- (b2l);
\draw[dotted] (0,1.0) -- (b3l);
\draw[dotted] (0,2.0) -- (b5l);
\draw[dotted] (0,2.5) -- (b6l);
\draw[dotted] (0,1.5) -- (4,1.5);
\draw[dotted] (b3r) -- (3,0.0);
\draw[dotted] (b2r) -- (2,0.0);
\draw[dotted] (b5r) -- (4,0.0);
\draw[dotted] (b6r) -- (5,0.0);
\draw[blue,thick,fill=blue] (b1r) circle (2pt);
\draw[blue,thick,fill=blue] (b2r) circle (2pt);
\draw[blue,thick,fill=blue] (b3r) circle (2pt);
\draw[blue,thick,fill=blue] (b5r) circle (2pt);
\draw[blue,thick,fill=blue] (b6r) circle (2pt);
\end{tikzpicture}
\caption{An illustration of the selection of effective jump points ($b'$) and effective bids ($z'$), for $\card{B}=6$. In the figure, jump points 3 and 4 coincide, and therefore among those, only jump point 4 will be in the sequence used in the next step. Also, bid $b_3$ is never used in the best-response function, as the strategy jumps directly from $b_2$ to $b_4$, and therefore $b_3$ will be excluded from the set of effective bids. In the end, the effective jump points would be $1,2,4$ and $5$ and the effective bids will be $b_0,b_1,b_2,b_4$ and $b_5$.}
\label{fig:effective-points}
\end{figure}

\noindent Formally, this corresponds to picking, for each bidder $i$, an (increasing) subsequence $\mu_i(j)\subseteq\ssets{1,\dots,\card{B}-1}$, such that 
$$ z_{i,j}=y_{i,\mu_i(j)}\quad\text{and}\quad
\ssets{0=z_{i,0}<z_{i,1}<\dots<z_{i,m_i}=1}=\ssets{0=y_{i,0}\leq y_{i,1} \leq y_{i,2}\leq \dots \leq y_{i,\card{B}}=1}.
$$
Notice that $m_i \leq \card{B}$. 
Given the ``guessing'' in the current step, we let $L_{i,j}, R_{i,j}$ denote the left and right, respectively, endpoints of the sub-interval in which the $j$-th effective break point of player $i$ lies; i.e., $z_{i,j}\in [L_{i,j},R_{i,j}]$. For ease of notation, we also use the shortcut $b_{i,j}'=b_{\mu_i(j)}$ for the $j$-th effective bid of player $i$.

Again, since $|B|$ is constant, we can enumerate over all possible effective jump point subsequences $\mu_i$ in constant time and for each such subsequence, we proceed to the next step.

\subsubsection*{Step 3: Solving a system of polynomial inequalities}

From the previous two steps we have, for each bidder $i$, an assignment of effective jump points $z_{i,0}, \ldots, z_{i,m_i}$ to intervals $[x_{\ell-1},x_\ell]$. 
In particular, $z_{i,j}$ is mapped to $[L_{i,j},R_{i,j}]$. 
Below, we express all the properties that must be satisfied by the effective jump points at an (exact) BNE of the FPA as a system of polynomial inequalities; the system includes inequalities to ensure
\begin{itemize}
    \item[-] that the positions of the jump points of each bidder $i$ respect the ordering implied by the set of indices, i.e., $z_{i,j-1} < z_{i,{j}}$ for all $j=1,\dots, m_i$, 
    \item[-] that the bidding strategies are non-overbidding,
    \item[-] that the variables $z_{i,j}$ indeed correspond to jump points of best-responses, in terms of the implications to the utility functions.
\end{itemize}

\fbox{%
\colorbox{gray!10!white}{
\hspace{-2cm}\parbox{\textwidth}{%
\begin{align}
& z_{i,j-1} < z_{i,j} &&\forall i,\;\forall j \label{eq:monotonicity_programme}\\ 
&L_{ij} \leq  z_{i,j} \leq R_{ij} &&\forall i,\; \forall j \label{eq:interval_points_programme}\\ 
& z_{i,j} \geq b'_{i,j} &&\forall i,\; \forall j \label{eq:no_overbidding_programme}\\ 
&u_i(b'_{i,j},\vec{z}_{-i};z_{i,j}) \geq
u_i(b,\vec{z}_{-i};z_{i,j})
&&\forall i,\;\forall j, \; \forall b<b'_{i,j} \label{eq:break_points_indiferance_programme}\\
&u_i(b'_{i,j-1},\vec{z}_{-i};z_{i,j}) \geq
u_i(b,\vec{z}_{-i};z_{i,j}) &&\forall i,\; \forall j,\;\forall b>b'_{i,j-1} \label{eq:best_response_programme}
\end{align}}}}
\providecommand{\refsystem}{System~\eqref{eq:monotonicity_programme}--\eqref{eq:best_response_programme}}
\medskip
\begin{lemma}
Fix a bidder $i$ and a bid $b\in B$. Then, for every $j=1,\dots,m_i$, her utility $u_i(b,\vec{z}_{-i};z_{i,j})$ can be expressed (in polynomial time) as a polynomial 
of degree at most $dn$ with respect to the effective jump point variables $\sset{z_{i',j'}}_{i'\in N,\; j'=0,\dots, m_i}$.
\end{lemma}

\begin{proof}
Without loss of generality, similar to what we did in the
proof of~\cref{lem:H-functions}, we will show the lemma from the perspective of bidder
$n$. Fix an index $j=0,\dots,m_n$ for an effective jump point $z_{n,j}\in
[L_{n,j},R_{n,j}]$, and consider a bid $b$. Then, importing some notation from our proof
of~\cref{lem:H-functions}, the utility of player $n$ when she has a true value of
$v_n=z_{n,j}$ is
$$
u_n(b,\vec{z}_{-n};z_{n,j}) = H_n(b,\vec{z}_{-n}) (z_{n,j}-b),
$$
where $H_n(b,\vec{z}_{-n})$ is the probability that bidder $n$ wins the item. Due
to~\eqref{eq:H_functions_sum} and~\eqref{eq:tblj} (and the fact that $n$ is now constant), it is enough to show that, for
any bidder $i\leq n-1$, the quantities $G_{i,b^-}$ and $g_{i,b}$, defined in the
proof of~\cref{lem:H-functions}, are polynomials of the jump point variables $z_{i',j'}$. Furthermore, to guarantee a maximum degree of $dn$, as in the statement of our lemma, it is enough to show that each of these polynomials are of degree at most $d$: the number of factors in the products appearing as summands in~\eqref{eq:tblj} are at most $n$.   

Recall that  $G_{i,b^-}$ and $g_{i,b}$ are the probabilities (from the perspective of bidder $n$) that bidder $i$ bids below $b$ and exactly $b$, respectively. So, if $b=b_{i,j'}'$ for some index $j'=0,1,\dots,m_i-1$, then
$G_{i,b^-}= F_{n,j}(z_{i,j'})$ 
and 
$g_{i,b}=F_{n,j}(z_{i,j'+1})-F_{n,j}(z_{i,j'})$.
If, on the other hand, $b_{i,j'}'<b<b_{i,j'+1}'$ for an index $j'$, then $G_{i,b^-}= F_{n,j}(z_{i,j'})$ and $g_{i,b}=0$.
In any case, deploying the representation
from~\eqref{eq:piecewise_poly_algorithm_representation}, quantities $G_{i,b^-}$ and
$g_{i,b}$ can indeed be written (in polynomial time with respect to the input of the problem) as polynomials, of degree at most $d$, of the jump point variables.
\end{proof}

As the following lemma suggests, a solution to \refsystem\ corresponds to a BNE of the first-price auction. Note that although the existence of a BNE is guaranteed by \cref{thm:existence}, it might be the case that the equilibrium strategies are \emph{not} consistent with the specific ``preliminary'' guesses of Steps~1 and~2 that gave rise to the particular instantiation of \refsystem\ above. However, there has to exist \emph{some guess} for which the system has a solution, and since we are enumerating over all possible choices, we are guaranteed to find it.
\begin{lemma}
Given the ``guessed'' allocations of jump points to intervals and the ``guessed'' effective jump points and bids, a compatible BNE of the FPA exists if and only if \refsystem\ has a solution. 
\end{lemma}

\begin{proof}
Immediate by the characterization of BNE in~\cref{lem:charepsilonbne} (using $\varepsilon=0$), by setting $\alpha_i(b^-)=z_{i,j}$ in condition~\eqref{eq:charepsilonbne1} and $\alpha_i(b)=z_{i,j}$  in~\eqref{eq:charepsilonbne2}. 
\end{proof}

\subsubsection*{Step 4: ``Projecting'' back to the equilibrium domain}

From Step~3 above, we know that by solving \refsystem\, we can compute an exact BNE of the auction. More precisely, we can compute a $\delta$-approximation to \refsystem\ in time polynomial in $\log(1/\delta)$, by making use of the following result by~\citet[Remark, p.~38]{grigor1988solving}:

\begin{theorem}
For any $\delta \in (0,1]$, it is possible to find a rational $\delta$-approximation to \refsystem\ in time polynomial in $\log(1/\delta)$ and the size of the input.  
\end{theorem}

By $\delta$-approximation here, we mean a point which is geometrically close, with respect to the max norm, to an exact solution of \refsystem. This is almost a strong approximation to an exact BNE; if we were to translate this point to a feasible strategy profile, it would yield jump points which are close to the jump points of an exact equilibrium strategy. However, these would only \emph{approximately} satisfy the conditions in \refsystem; in particular special care should be taken for the monotonicity and no-overbidding conditions, which we want to be satisfied \emph{exactly}, rather than approximately. 

To remedy this, we must first ``project'' the $\delta$-approximate solution of \refsystem\ back to the equilibrium domain $\dcal$ introduced in the proof of \cref{thm:existence}. Formally, let us denote by $\vec{z}$ the $\delta$-approximate solution of \refsystem, and by $\vec{z}^\ast$ the exact solution which it approximates, so that $\|\vec{z}-\vec{z}^\ast\|_\infty\leq\delta$. 
We compute the projection $\tilde{\vec{z}}$ from $\vec{z}$ as $$\tilde{z}_{i,0}=0\quad\text{and}\quad\tilde{z}_{i,j}=\trunc_{[\max\{b'_{i,j},\tilde{z}_{i,j-1}\},1]}(z_{i,j}).$$
Our next claim is that $\|\tilde{\vec{z}}-\vec{z}^\ast\|_\infty\leq\delta$ as well. This is equivalent to saying that $|\tilde{z}_{i,j}-z^\ast_{i,j}|\leq\delta$ for every $i,j$, which can be done by induction on $j$, the base case $j=0$ being trivial. For $j>0$, observe that $\tilde{z}_{i,j}$ must coincide with one of $b'_{i,j},\tilde{z}_{i,j-1},1,z_{i,j}$.
\begin{itemize}
    \item If $\tilde{z}_{i,j}=z_{i,j}$ then obviously $|\tilde{z}_{i,j}-z^\ast_{i,j}|\leq\delta$.
    \item If  $\tilde{z}_{i,j}=b'_{i,j}$ then we must have had $z_{i,j}\leq b'_{i,j}$. Since $z^\ast_{i,j}\geq b'_{i,j}$ and $|z_{i,j}-z^\ast_{i,j}|\leq\delta$, we must also have $|\tilde{z}_{i,j}-z^\ast_{i,j}|\leq\delta$.
    \item Similarly, if $\tilde{z}_{i,j}=1$ then we must have had $z_{i,j}\geq 1$. Since $z^\ast_{i,j}\leq 1$ and $|z_{i,j}-z^\ast_{i,j}|\leq\delta$, we must also have $|\tilde{z}_{i,j}-z^\ast_{i,j}|\leq\delta$.
    \item Finally, suppose $\tilde{z}_{i,j}=\tilde{z}_{i,j-1}$. Then we must have had $z_{i,j}\leq \tilde{z}_{i,j-1}$. Using the induction hypothesis, we have that $\tilde{z}_{i,j-1}\leq z^\ast_{i,j-1}+\delta\leq z^\ast_{i,j}+\delta$; thus we also have $|\tilde{z}_{i,j}-z^\ast_{i,j}|\leq\delta$.
\end{itemize}

Therefore, $\tilde{\vec{z}}$ constitutes a valid monotone non-decreasing, non-overbidding joint strategy profile, which is within distance $\delta$ of the exact BNE $\vec{z}^\ast$. In other words, $\tilde{\vec{z}}$ is a valid joint strategy profile that is a \emph{strong} $\delta$-approximation to a BNE.

Finally, we need to show that if $\delta$ is chosen to be sufficiently small, then any strong $\delta$-approximate BNE is also an $\varepsilon$-BNE of the auction. For this, we use the fact that the family of piecewise polynomial distributions is polynomially continuous (see \cref{app:inputs} for the formal definition). Indeed, given such a piecewise polynomial distribution, it is easy to see that it must be Lipschitz-continuous, and, crucially, we can in polynomial time compute a corresponding Lipschitz-constant. (Note that any polynomial function $F(z_j)=a_0+a_1z_j+\ldots+a_dz_j^d$ is $L$-Lipschitz-continuous over $[0,1]$, where $L = |a_1|+2|a_2|+\ldots d|a_d|$.)
With this observation in hand, we can now use \cref{lem:util-poly-cont} to efficiently construct $\delta > 0$ sufficiently small such that for all $i \in N$, $b \in B$ and $v_i \in [0,1]$
$$\|\vec{z} - \vec{z}'\|_\infty \leq \delta \implies |u_i(b,\vec{z}_{-i};v_i) - u_i(b,\vec{z}_{-i}';v_i)| \leq \varepsilon/2.$$
Since $\tilde{\vec{z}}$ is a strong $\delta$-approximation, i.e., $\|\tilde{\vec{z}} - \vec{z}^\ast\|_\infty \leq \delta$, it immediately follows that inequalities \eqref{eq:break_points_indiferance_programme} and \eqref{eq:best_response_programme} of the System are satisfied with additive error at most $\varepsilon$. Using \cref{lem:charepsilonbne}, it immediately follows that $\tilde{\vec{z}}$ is an $\varepsilon$-BNE.

As a result, to summarize, given $\varepsilon > 0$ and the problem instance, we can in polynomial time compute $\delta > 0$ such that running the algorithm described in this section is guaranteed to find an $\varepsilon$-BNE. Since the number of agents and bids is fixed, and the algorithm runs in polynomial time in $\log(1/\delta)$ and the instance size, \cref{thm:positive} follows.

\section{Conclusion and Future Directions}\label{sec:conclusion}

In this paper, we have classified the complexity of computing a Bayes-Nash equilibrium of the first-price auction with subjective priors, by proving that it is PPAD-complete. As we explained in the introduction, our result contributes fundamentally to our understanding of this celebrated auction format, as well as the literature on total search problems and TFNP. The challenging next step is to move towards the special case of the common priors assumption, where the value distribution of each bidder is common knowledge ($F_{i,j} = F_{i',j}$ for all $i,i'$). Our PPAD-membership result obviously already extends to this case, as it is a special case of the subjective priors setting. The really intriguing question is to extend our PPAD-hardness result to this case as well. To this end, we state the following open problem, which we consider to be one of the most important problems both in computational game theory and in the literature of total search problems.   

\begin{opproblem}
What is the complexity of computing an $\varepsilon$-Bayes-Nash equilibrium of the first-price auction with common priors? Is it \ppad-complete? Is it polynomial-time solvable? Or could it be complete for some other (smaller) sub-class of \ppad?
\end{opproblem}

\noindent A potential candidate for such a smaller class could be the class $\ppad \cap \pls$, which was recently shown by \citet{fearnley2021complexity} and \citet{babichenko2020settling} to capture the complexity of interesting problems related to optimization via gradient descent, and computing mixed Nash equilibria in congestion games \citep{rosenthal1973class} respectively. The class \pls\ was introduced by \citet{johnson1988easy} and captures the computation of local minima of some objective function, and notably characterizes the complexity of finding \emph{pure} Nash equilibria in congestion games \citep{fabrikant2004complexity}.

A possible ``intermediate'' step before settling the open problem above for common priors would be to consider priors that are still subjective, but \emph{consistent}, meaning that there exists some common prior distribution $P$ (a ``ground truth'') over the set of value profiles, such that each bidder's subjective prior distribution given her own value can be directly computed from $P$. As \citet[Sec.~2.8]{myerson2013game} argues, when the subjective priors are consistent, the differences in beliefs can be explained by differences in information, rather than differences in opinion (which are captured even by inconsistent beliefs). In settings like the First-Price Auction, it is meaningful to assume that beliefs are often formed based on observing public signals (e.g., the bidding history of the competitors), possibly with varying degrees of information, and hence subjective priors are quite meaningful.  

Another very interesting question is to study the case where both the value distributions and the bidding space are discrete. A special case of this setting was studied by \citet{escamocher2009existence}, but they only obtained conclusive results for the case of two bidders with bi-valued distributions. We believe that some of our technical contributions (e.g., the computation of the best response functions or the gadgets used in the PPAD-hardness proof) can be adapted to show similar results for that case as well; we leave the details for future work. Finally, it would be very interesting to identify further (in)tractable special cases for our problem; for example, can we obtain a positive result similar to \cref{thm:positive} for more general value distributions? Do the hardness results also hold in the setting where the number of bidders is constant, but the bidding space is allowed to be large?

\clearpage

\appendix

\section*{APPENDIX}

\section{The Input Model for the Value Distributions}\label{app:inputs}

Let $\mathcal{F}$ be a class of cumulative distribution functions on the interval $[0,1]$. In other words, for any $F \in \mathcal{F}$ and any $x \in [0,1]$, $F(x)$ is the probability of the interval $[0,x]$ according to $F$. For every $F \in \mathcal{F}$, let $\sz(F)$ denote the representation size of $F$, i.e., the number of bits needed to represent $F$. (Here we implicitly assume that some representation scheme is given in the definition of $\mathcal{F}$.)

For any rational number $x$, let $\sz(x)$ denote the representation size of $x$, namely the length of the binary representation of the denominator and numerator of $x$. The definitions in this section are based on the corresponding notions introduced by \citet{etessami2010complexity}.

\begin{definition}\label{def:polycompfun}
A class of cumulative distribution functions $\mathcal{F}$ is \emph{polynomially computable}, if there exists some polynomial $p$ such that for all $F \in \mathcal{F}$ and all rational $x \in [0,1]$, $F(x)$ can be computed in time $p(\sz(F)+\sz(x))$.
\end{definition}

\noindent In order to guarantee the existence of approximate equilibria with polynomial representation size we add an extra requirement on $\mathcal{F}$.

\begin{definition}
A class of cumulative distribution functions $\mathcal{F}$ is \emph{polynomially continuous}, if there exists some polynomial $q$ such that for all $F \in \mathcal{F}$ and all rational $\varepsilon > 0$, there exists rational $\delta > 0$ with $\sz(\delta) \leq q(\sz(F)+\sz(\varepsilon))$ such that
$$|F(x)-F(y)| \leq \varepsilon$$
for all $x,y \in [0,1]$ with $|x-y| \leq \delta$.
\end{definition}

\noindent Note that distribution functions given by piecewise-constant density functions on the interval $[0,1]$ are an example of such a class of polynomially-computable and polynomially-continuous $\mathcal{F}$. The density functions are represented explicitly, i.e., as a list of ``blocks'', where for every block we give the sub-interval of $[0,1]$ that it occupies and the height of the block.

\section{Impossibilities for Implicit Bidding Spaces}\label{app:bidding-space}

In \cref{sec:prelims}, we emphasized that it is necessary for our computational problem to have the bidding space explicitly as part of the input, as otherwise it is hard to even compute the best responses of the auction. We provide more details on this topic in this section.

If the bidding space $B \subseteq [0,1]$ is discrete but represented in some implicit way, this immediately gives rise to some computational obstacles. When we proved in \Cref{sec:bestresponses} that best-responses could be computed efficiently, our procedure essentially goes over all possible bids, and checks which bid achieves the highest utility. If the bidding space is large (say, exponential in the input size), this approach is no longer efficient. In fact, in this subsection we will prove that, essentially, one cannot hope to find a better approach; in particular, we provide lower bounds from an information-theoretical as well as a computational perspective.

For simplicity, in this subsection we will assume that the bidding space is the set of all rational numbers in $[0,1]$ that have denominator $2^m$,
\[B=\left\{\frac{p}{2^m}\fwh{0\leq p\leq 2^m}\right\},\]
where $m$ is part of the input and given in unary representation. Notice that each bid can then be encoded by a binary string of size $m$ (with the exception of the bid $1$, which can be encoded with $m+1$ bits). We will also assume that there are only two bidders, each having a valuation over the unit interval, $V=[0,1]$. This is arguably the simplest natural example one could consider.

As we explained in \cref{sec:prelims} we can identify a strategy by its set of jump points
$$\alpha_i(b)=\sup\{v\fwh{\beta_i(v)\leq b}\}.$$
Intuitively, $\alpha_i(b)$ is the largest value for which player $i$ would bid $b$ or lower. At this point we have two options on how to represent the functions $\alpha_i$:
\begin{itemize}
    \item[-] \textbf{Black-box model}: in the black-box model we have access to an oracle that, given a bid $b\in B$, returns the corresponding jump point $\alpha_i(b)$.
    \item[-] \textbf{White-box model}: in the white-box model we have an algorithm that, given a bid $b\in B$, computes the jump point $\alpha_i(b)$. For example, this could be given by a circuit. Alternatively, we can assume that $\alpha_i$ is a function computable in polynomial time.
\end{itemize}
In both cases we need to describe how the jump points themselves are represented. For simplicity, we just assume that all jump points are rational quantities (as we are going for a hardness result).

Besides the inverse bidding strategies, we also need to represent the cumulative density functions $F_i:[0,1]\rightarrow[0,1]$. Here similar considerations apply, or we can use the notions in \cref{app:inputs}.

Now, given $F_i$ and $\alpha_i$, an important quantity of interest is
$$\Pi_i(b)=F_i(\alpha_i(b));$$
since $\alpha_i(b)$ is the largest value for which bidder $i$ will bid $b$ or lower, and $F_i(\alpha_i(b))$ is the probability that bidder $i$'s valuation is at most this value, it turns out that $\Pi_i(b)$ can be very naturally interpreted as the probability that player $i$ bids on or below $b$. Notice that we can then get the probability that player $i$ bids exactly $b$ as $\Pi_i(b)-\Pi_i(b^-)$, where $b^-$ is the bid immediately below $b$, for $b>0$. Regarding the computation of $\Pi_i$, it will be either a black-box or white-box computation, depending on whether we have assumed $\alpha_i$ and $F_i$ to be given in a black-box or white-box fashion.

As we already mentioned, in our reduction we will consider only two bidders. We shall fix the second bidder's bidding strategy and cumulative distribution function throughout the reduction, and look at the best-response of bidder $1$. For ease of notation, we will drop the subscript $2$ and write $\alpha,F,\Pi$ instead of $\alpha_2,F_2,\Pi_2$; there will be no confusion since we will never look at bidder $1$'s valuation distribution or bidding strategy. Given a bid $b$, we can express the probability that bidder $1$ wins the auction when bidding $b$, denoted by $H(b)$, via
\begin{align*}
H(0)&=\frac{1}{2}\Pi(0);\\
H(b)&=\Pi(b^-)+\frac{1}{2}\left(\Pi(b)-\Pi(b^-)\right)=\frac{1}{2}\left(\Pi(b^-)+\Pi(b)\right), \qquad\text{for}\;\; b>0.
\end{align*}
Finally, we wish to maximize the utility of bidder $1$; when she has a valuation of $v$ and bids $b$, this is given by $u(v,b)=H(b)(v-b)$.

Now that we have given the preliminaries of our reduction, let us go into the construction. Let us fix some $m\geq 3$ and define a baseline instance. We will want to choose a bidding strategy $\alpha$ and distribution $F$ for bidder $2$, so that the resulting function $\Pi(\cdot)$ is given as follows.

\begin{align*}
\Pi(0)=\Pi(2^{-m})&=\Pi(2\cdot 2^{-m})=0;\\
\Pi(b-2^{-m})=\Pi(b)&=\Pi(b+2^{-m})=\Pi(b+2\cdot2^{-m})\\
&=\frac{1}{2(1-b)},\quad\text{for}\quad b=p\cdot2^{-m},\,p\,\text{a multiple of }4,\,\text{and}\,b\leq\frac{1}{2};\\
\Pi(b)&=1\quad\text{for}\quad b\geq\frac{1}{2}.
\end{align*}

Our function $\Pi$ essentially corresponds to a discrete probability distribution on the bids with the following properties. First, it only has mass at points of the form $(4k-1)\cdot 2^{-m}$, for positive integer $k$, where $4k-1<2^{m-1}$. Second, the mass at $3\cdot 2^{-m}$ equals $\frac{1}{2(1-4\cdot 2^{-m})}$, whereas for $k\geq 2$ the mass at $(4k-1)\cdot 2^{-m}$ equals $\frac{1}{2(1-4k\cdot 2^{-m})}-\frac{1}{2(1-(4k-1)\cdot 2^{-m})}$. To yield the desired $\Pi$, we can for example take $F(x)=x$, corresponding to the uniform distribution on $[0,1]$, and $\alpha(b)=\Pi(b)$ defined as above.

Given the probability distribution $\Pi$ on the bids of player $2$, we are interested in computing the best-response strategy for player $1$. In fact, we will do so for the case that player $1$'s valuation equals $1$. If we can show it is hard to compute the best-response for this value, then it follows that it is hard to compute the best-response strategy function in general. Using the definition of $H(b)$, we can write
\[H(0)=0;\quad H(2^{-m})=0;\quad H(2\cdot 2^{-m})=0;\quad H(3\cdot 2^{-m})\frac{1}{4(1-4\cdot 2^{-m})};\]
for $b=p\cdot2^{-m}$, $p$ a multiple of $4$, and $b\leq\frac{1}{2}-4\cdot 2^{-m}$,
\[H(b)=\frac{1}{2(1-b)};\quad H(b+2^{-m})=\frac{1}{2(1-b)};\quad H(b+2\cdot2^{-m})=\frac{1}{2(1-b)};\]
\[H(b+3\cdot 2^{-m})=\frac{1}{4(1-b)}+\frac{1}{4(1-b-4\cdot 2^{-m})};\]
finally, for $b\geq 1/2$,
\[H(b)=1.\]

\begin{figure}[ht]\begin{center}
\begin{tikzpicture}
\tikzmath{\eps=1/32;\xsc=10;\ysc=4;}
\draw[->] (0,0) -- (\xsc*0.75,\ysc*0);
\draw[->] (0,0) -- (\xsc*0,\ysc*1.25);
\node[below] at (\xsc*0.75,\ysc*0) {$b$};
\node[below] at (\xsc*0.5,\ysc*0) {$1/2$};
\node[below] at (\xsc*22*\eps,\ysc*0) {$1$};
\node[left] at (\xsc*0,\ysc*1.25) {$H(b)$};
\node[left] at (\xsc*0,\ysc*1) {$1$};
\coordinate (p00) at (\xsc*0*\eps,\ysc*0);
\coordinate (p01) at (\xsc*1*\eps,\ysc*0);
\coordinate (p02) at (\xsc*2*\eps,\ysc*0);
\coordinate (p03) at (\xsc*3*\eps,{\ysc*1/(4*(1-4*\eps))});
\coordinate (p04) at (\xsc*4*\eps,{\ysc*1/(2*(1-4*\eps))});
\coordinate (p05) at (\xsc*5*\eps,{\ysc*1/(2*(1-4*\eps))});
\coordinate (p06) at (\xsc*6*\eps,{\ysc*1/(2*(1-4*\eps))});
\coordinate (p07) at (\xsc*7*\eps,{\ysc*1/(4*(1-4*\eps))+\ysc*1/(4*(1-8*\eps))});
\coordinate (p08) at (\xsc*8*\eps,{\ysc*1/(2*(1-8*\eps))});
\coordinate (p09) at (\xsc*9*\eps,{\ysc*1/(2*(1-8*\eps))});
\coordinate (p10) at (\xsc*10*\eps,{\ysc*1/(2*(1-8*\eps))});
\coordinate (p11) at (\xsc*11*\eps,{\ysc*1/(4*(1-8*\eps))+\ysc*1/(4*(1-12*\eps))});
\coordinate (p12) at (\xsc*12*\eps,{\ysc*1/(2*(1-12*\eps))});
\coordinate (p13) at (\xsc*13*\eps,{\ysc*1/(2*(1-12*\eps))});
\coordinate (p14) at (\xsc*14*\eps,{\ysc*1/(2*(1-12*\eps))});
\coordinate (p15) at (\xsc*15*\eps,{\ysc*1/(4*(1-12*\eps))+\ysc*1/(4*(1-16*\eps))});
\coordinate (p16) at (\xsc*16*\eps,{\ysc*1/(2*(1-16*\eps))});
\coordinate (p17) at (\xsc*17*\eps,{\ysc*1/(2*(1-16*\eps))});
\coordinate (p18) at (\xsc*18*\eps,{\ysc*1/(2*(1-16*\eps))});
\coordinate (p19) at (\xsc*19*\eps,{\ysc*1/(2*(1-16*\eps))});
\coordinate (p22) at (\xsc*22*\eps,{\ysc*1/(2*(1-16*\eps))});
\draw[blue,thick] (p00) -- (p01) -- (p02) -- (p03) -- (p04) -- (p05) -- (p06) -- (p07) -- (p08);
\draw[blue,dotted,thick] (p08) -- (p09);
\draw[blue,dotted,thick] (p11) -- (p12);
\draw[blue,thick] (p12) -- (p13) -- (p14) -- (p15) -- (p16) -- (p17) -- (p18) -- (p19);
\draw[blue,dotted,thick] (p19) -- (p22);
\draw[thick,fill=blue] (p00) circle (1.5pt);
\draw[thick,fill=blue] (p01) circle (1.5pt);
\draw[thick,fill=blue] (p02) circle (1.5pt);
\draw[thick,fill=blue] (p03) circle (1.5pt);
\draw[thick,fill=blue] (p04) circle (1.5pt);
\draw[thick,fill=blue] (p05) circle (1.5pt);
\draw[thick,fill=blue] (p06) circle (1.5pt);
\draw[thick,fill=blue] (p07) circle (1.5pt);
\draw[thick,fill=blue] (p08) circle (1.5pt);
\draw[thick,fill=blue] (p12) circle (1.5pt);
\draw[thick,fill=blue] (p13) circle (1.5pt);
\draw[thick,fill=blue] (p14) circle (1.5pt);
\draw[thick,fill=blue] (p15) circle (1.5pt);
\draw[thick,fill=blue] (p16) circle (1.5pt);
\draw[thick,fill=blue] (p17) circle (1.5pt);
\draw[thick,fill=blue] (p18) circle (1.5pt);
\draw[thick,fill=blue] (p22) circle (1.5pt);
\draw[dotted] (\xsc*16*\eps,0) -- (\xsc*16*\eps,\ysc*1);
\draw[dotted] (\xsc*22*\eps,0) -- (\xsc*22*\eps,\ysc*1);
\draw[dotted] (0,\ysc*1) -- (\xsc*16*\eps,\ysc*1);
\draw[black,dashed,domain=0:{9*\eps}] plot (\xsc*\x,{\ysc*1/(2*(1-\x))});
\draw[black,dotted,domain={9*\eps}:{11*\eps}] plot (\xsc*\x,{\ysc*1/(2*(1-\x))});
\draw[black,dashed,domain={11*\eps}:{16*\eps}] plot (\xsc*\x,{\ysc*1/(2*(1-\x))});
\end{tikzpicture}
\end{center}\caption{Depiction of the baseline construction. $H(b)$ denotes the probability of player $1$ winning the auction when bidding $b$, and is represented by the blue circles. We also plot in dashed line the auxiliary function $x\mapsto\frac{1}{2(1-x)}$.\label{fig:hfunctionbaseline}}\end{figure}
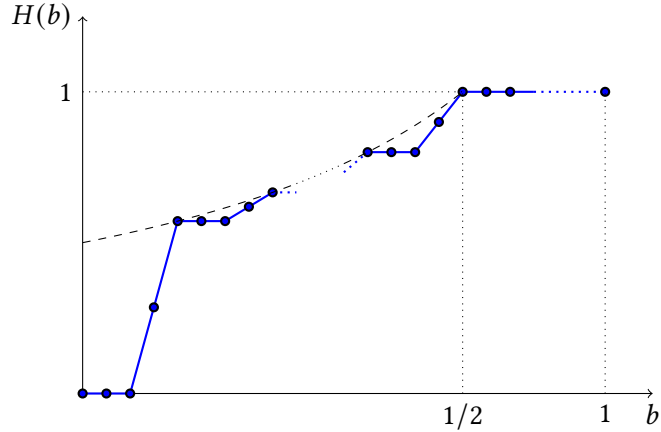

A graphical depiction of $H(b)$ can be found in \cref{fig:hfunctionbaseline}. It is not hard to check that, for every bid $b$, we have that
\[H(b)\leq \frac{1}{2(1-b)};\]
moreover, this is achieved with equality for every bid of the form $b=p\cdot 2^{-m}$, for $p$ a multiple of $4$, as long as $b\leq 1/2$. Therefore, the maximum utility that player $1$ can achieve is $1/2$, and all such multiple-of-four bids are equally best-responses.

Now that we understand the baseline instance, we can construct a family of ``perturbed'' instances that will be used in our reduction. For a fixed subset $S\subseteq\{0,1\}^{m-3}$ of binary strings of size $m-3$, we will define a corresponding $\Pi_S$, $H_S$ as follows. $\Pi_S$ and $H_S$ coincide with $\Pi$ and $H$ on every bid $b\geq 1/2$. For the bids smaller than $1/2$, we can write their binary expansion as a sequence of $m$ bits, the first of which is $0$. For example, if $m=4$, then the bid $3/2^4$ can be written as $0011$. For every $x\in\{0,1\}^{m-3}$, if $x\not\in S$, then $\Pi_S$ and $H_S$ coincide with $\Pi$ and $H$ for bids of the form $0xb_1b_2$; in particular, if $b=x\cdot 2^{m-2}$,
\begin{itemize}
    \item if $x=0\cdots0$, then we have $\Pi_S(0)=0$, $\Pi_S(2^{-m})=0$, $\Pi_S(2\cdot 2^{-m})=0$, $\Pi_S(3\cdot 2^{-m})=\frac{1}{2(1-4\cdot 2^{-m})}$;
    \item otherwise, we have $\Pi_S(b)=\frac{1}{2(1-b)}$, $\Pi_S(b+2^{-m})=\frac{1}{2(1-b)}$, $\Pi_S(b+2\cdot 2^{-m})=\frac{1}{2(1-b)}$, $\Pi_S(b+3\cdot 2^{-m})=\frac{1}{2(1-b-4\cdot 2^{-m})}$.
\end{itemize}

On the other hand, for $x\in S$ and $b=x\cdot 2^{-m+2}$, $\Pi_S$ is obtained from $\Pi$ by shifting the mass at $b+3\cdot 2^{-m}$ to $b+2^{-m}$; in other words,
\begin{itemize}
    \item if $x=0\cdots0$, then we have $\Pi_S(0)=0$, $\Pi_S(2^{-m})=\frac{1}{2(1-4\cdot 2^{-m})}$, $\Pi_S(2\cdot 2^{-m})=\frac{1}{2(1-4\cdot 2^{-m})}$, $\Pi_S(3\cdot 2^{-m})=\frac{1}{2(1-4\cdot 2^{-m})}$;
    \item otherwise, we have $\Pi_S(b)=\frac{1}{2(1-b)}$, $\Pi_S(b+2^{-m})=\frac{1}{2(1-b-4\cdot 2^{-m})}$, $\Pi_S(b+2\cdot 2^{-m})=\frac{1}{2(1-b-4\cdot 2^{-m})}$, $\Pi_S(b+3\cdot 2^{-m})=\frac{1}{2(1-b-4\cdot 2^{-m})}$.
\end{itemize}
This gives rise to a change in $H_S$ as well:
\begin{itemize}
    \item if $x=0\cdots0$, then we have $H_S(0)=0$, $H_S(2^{-m})=\frac{1}{4(1-4\cdot 2^{-m})}$, $H_S(2\cdot 2^{-m})=\frac{1}{2(1-4\cdot 2^{-m})}$, $H_S(3\cdot 2^{-m})=\frac{1}{2(1-4\cdot 2^{-m})}$;
    \item otherwise, we have $H_S(b)=\frac{1}{2(1-b)}$, $H_S(b+2^{-m})=\frac{1}{4(1-b)}+\frac{1}{4(1-b-4\cdot 2^{-m})}$, $H_S(b+2\cdot 2^{-m})=\frac{1}{2(1-b-4\cdot 2^{-m})}$, $H_S(b+3\cdot 2^{-m})=\frac{1}{2(1-b-4\cdot 2^{-m})}$. 
\end{itemize}

Similarly as above, we can define an $\alpha_S$ for player $2$ that give rise to this choice of $\Pi_S$ and $H_S$. The net effect of our construction is that, for $x\not\in S$, the bids of the form $0x00$, $0x01$, $0x10$ and $0x11$ achieve the same utility in both the baseline and the perturbed instances (and thus, at most $1/2$); but if $x\in S$, the bids of the form $0x01$, $0x10$, $0x11$ now achieve higher utility; in fact, if $x\in S$, then bidding $0x10$ achieves a utility strictly higher than $1/2$. Writing $b=x\cdot 2^{-m+2}$, we can see that
\[u(1,b+2\cdot2^{-m})=\frac{1-b-2\cdot 2^{-m}}{2(1-b-4\cdot 2^{-m})}>\frac{1}{2}.\]
We want to find an $\varepsilon$ that bounds the utility gap, in order to show that computing $\varepsilon$-best-responses is hard. Using the trivial bound that $1-b-4\cdot 2^{-m}<1$, it turns out that $\varepsilon\leq 2^{-m}$ is small enough:
\[u(1,b+2\cdot2^{-m})-\frac{1}{2}=\frac{1-b-2\cdot 2^{-m}}{2(1-b-4\cdot 2^{-m})}-\frac{1-b-4\cdot 2^{-m}}{2(1-b-4\cdot 2^{-m})}=\frac{2\cdot 2^{-m}}{2(1-b-4\cdot 2^{-m})}>2^{-m}.\]

\begin{figure}[ht]\begin{center}
\begin{tikzpicture}
\tikzmath{\eps=1/32;\xsc=40;\ysc=40;}
\coordinate (p04) at (\xsc*4*\eps,{\ysc*1/(2*(1-4*\eps))});
\coordinate (p05) at (\xsc*5*\eps,{\ysc*1/(2*(1-4*\eps))});
\coordinate (p06) at (\xsc*6*\eps,{\ysc*1/(2*(1-4*\eps))});
\coordinate (p07) at (\xsc*7*\eps,{\ysc*1/(4*(1-4*\eps))+\ysc*1/(4*(1-8*\eps))});
\coordinate (p08) at (\xsc*8*\eps,{\ysc*1/(2*(1-8*\eps))});
\coordinate (q05) at (\xsc*5*\eps,{\ysc*1/(4*(1-4*\eps))+\ysc*1/(4*(1-8*\eps))});
\coordinate (q06) at (\xsc*6*\eps,{\ysc*1/(2*(1-8*\eps))});
\coordinate (q07) at (\xsc*7*\eps,{\ysc*1/(2*(1-8*\eps))});
\node[below] at (\xsc*4*\eps,{\ysc*1/(2*(1-4*\eps))}) {$0x00$};
\node[below] at (\xsc*5*\eps,{\ysc*1/(2*(1-4*\eps))}) {$0x01$};
\node[below] at (\xsc*6*\eps,{\ysc*1/(2*(1-4*\eps))}) {$0x10$};
\node[below] at (\xsc*7*\eps,{\ysc*1/(2*(1-4*\eps))}) {$0x11$};
\node[below] at (\xsc*8*\eps,{\ysc*1/(2*(1-4*\eps))}) {$0x^+00$};
\node[right] at (p07) {$H(b)$};
\node[left] at (q05) {$H_S(b)$};
\draw[blue,thick] (p04) -- (p05) -- (p06) -- (p07) -- (p08);
\draw[blue,thick] (p04) -- (q05) -- (q06) -- (q07) -- (p08);
\draw[thick,fill=blue] (p04) circle (1.5pt);
\draw[thick,fill=blue] (p05) circle (1.5pt);
\draw[thick,fill=blue] (p06) circle (1.5pt);
\draw[thick,fill=blue] (p07) circle (1.5pt);
\draw[thick,fill=blue] (p08) circle (1.5pt);
\draw[thick,fill=blue] (q05) circle (1.5pt);
\draw[thick,fill=blue] (q06) circle (1.5pt);
\draw[thick,fill=blue] (q07) circle (1.5pt);
\draw[black,dashed,domain={4*\eps}:{8*\eps}] plot (\xsc*\x,{\ysc*1/(2*(1-\x))});
\end{tikzpicture}
\end{center}\caption{Depiction of the baseline construction. $H_S(b)$ denotes the probability of player $1$ winning the auction when bidding $b$, and is represented by the upper blue circles. These are higher than the probabilities in $H(b)$ (lower blue circles), and go above the function $x\mapsto\frac{1}{2(1-x)}$ (dashed line). Here $0x^+00$ represents the binary string immediately after $0x11$.\label{fig:hsfunction}}\end{figure}
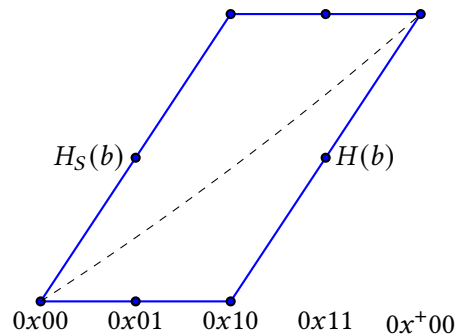

We can depict the change from function $H$ to function $H_S$ as in \cref{fig:hsfunction}. To conclude this section, we just need to prove that one cannot distinguish between $H$ and $H_S$ unless we explicitly compute utilities for a large number possible bids.

\begin{theorem}
Consider a FPA where the bidding space corresponds to all dyadic rationals of order $m$, and bidding strategies are represented implicitly according to the black-box model. Then, any algorithm that computes $\varepsilon$-best-responses, for $\varepsilon\leq 2^{-m}$, makes an exponential number of queries in the worst-case.
\end{theorem}

\begin{proof}
Let $A$ be an algorithm that computes exact best-responses. Fix an integer $m\geq 3$, the number of players to be $2$, and run the algorithm $A$ for the baseline instance, where player $2$ bids according to function $H$, and player $1$'s value is fixed to $1$. Suppose that $A$ makes less than $2^{m-3}-1$ queries; let $Q$ be the set of queries made by $A$, and $b$ be the bid returned by $A$. 
Next, notice that there are $2^{m-3}$ disjoint sets of bids of the form $\{0x01,0x10,0x11\}$, one for each $x\in\{0,1\}^{m-3}$. Since $A$ makes less than $2^{m-3}-1$ queries, it follows that there must exist some $x$ for which none of $0x01$, $0x10$, $0x11$ belongs to $Q\cup\{b\}$. Now consider the perturbed instance $H_x$, that is, we take $S=\{x\}$. Notice that $H_S$ and $\Pi_S$ coincide with $H$ and $\Pi$ everywhere except at $\{0x01,0x10,0x11\}$; therefore, running $A$ on the instance $H_x$ would produce the same answers on all queries, and so would produce the same best-response bid of $b$. However, by our construction we know that bidding $b$ gives an utility of at most $1/2$, whereas bidding according to the string $0x1$ gives an utility strictly higher than $1/2+\varepsilon$. Hence, the algorithm would not give a correct answer. We conclude that any algorithm for computing best-responses would have to make at least $2^{m-3}-1$ queries.
\end{proof}

\begin{theorem}
Consider a FPA where the bidding space corresponds to all dyadic rationals of order $m$, and bidding strategies are represented implicitly according to the white-box model. Then, computing $\varepsilon$-best-responses, for exponentially small $\varepsilon$, is an NP-hard optimization problem.
\end{theorem}

\begin{proof}
Let $\pcal$ be any problem in NP. Without loss of generality assume that certificates for instances of size $n$ must all have size $p(n)$, for some polynomial $p$. Given an input $y$ for $\pcal$, let $S(y)=\{x\in\{0,1\}^{p(n)}\,\,:\,\,x\text{ is a valid certificate for }y\}\subseteq\{0,1\}^{p(n)}$ be the set of valid certificates for $y$. In other words, $y$ is a yes-instance if and only if $S(y)\neq\emptyset$; and there is a polynomial-time algorithm that, given $x,y$, decides whether $x\in S(y)$.

We can define our reduction, from $\pcal$ to the problem of computing best-responses, as follows. Given an input $y$ of size $n$, consider a first-price auction where:
\begin{itemize}
    \item $m=p(n)+3$ and $\varepsilon=2^{-m}$;
    \item the bidding space corresponds to all dyadic rationals of order $m$;
    \item there are two players; the second player has a bidding distribution according to the perturbed instance $H_{S(y)}$;
    \item the first player has a valuation of $1$.
\end{itemize}

Notice that we can indeed construct this auction in polynomial time. In particular, there is an algorithm that computes $\Pi_{S(y)}(b)$ as follows. If $b\geq 1/2$, then $\Pi_{S(y)}(b)=1$. Otherwise, write $b$ in the form $0xb_1b_2$; decide whether $x\in S(y)$ (in polynomial time); depending on the answer, compute $\Pi_{S(y)}(b)$ according to the formulas above.

To complete the proof, suppose $y$ is a no-instance. Then $\Pi_{S(y)}=\Pi_\emptyset=\Pi$, and the best-response for player $1$ in this auction must achieve utility of exactly $1/2$, so that any $\varepsilon$-best-response achieves utility at most $1/2$. On the other hand, suppose $y$ is a yes-instance. Then $S(y)\neq\emptyset$, and the best-response for player $1$ in this auction must achieve utility strictly higher than $1/2+2^{-m}$, so that any $\varepsilon$-best-response achieves utility strictly higher than $1/2$.
\end{proof}

\section{Exact Equilibria Can Be Irrational}\label{app:irrational}

In this section we provide the technical details on \cref{ex:irrational}, which shows that a FPA can have \emph{only} irrational equilibria. Recall that in \cref{sec:prelims} we imposed two standard assumptions in the literature, namely that equilibrium strategies are \emph{monotone non-decreasing} and exhibit \emph{no overbidding}. Here, since we would like to argue that \emph{all equilibria} are irrational, to make our statement even stronger, we will show that in the example that we construct, all equilibria are necessarily monotone non-decreasing and non-overbidding, as well as irrational. In fact, we will show that the example admits a unique equilibrium, and that this equilibrium has all three properties.

To this end, we start with the following proposition that states that essentially, violations of overbidding and monotonicity only occur in trivial corner cases. In our subsequent construction, such cases will not occur. 

\begin{proposition}\label{prop:overbidnonmonotone}
Let $\vec{\beta}$ be an exact equilibrium of a FPA. For a bid $b_i$ by player $i$, let $H_i(b_i,\vec{\beta}_{-i})$ denote the (perceived) probability that player $i$ gets the item, given this bid and the bidding strategies by the other players.
Then, strategies will always be no over-bidding and monotone non-decreasing except only possibly when the probability of winning is zero. Formally,  

\begin{enumerate}
    \item let $v_i$ be a valuation by player $i$ and $b_i=\beta_i(v_i)$. If $b_i>v_i$, then $H_i(b_i,\vec{\beta}_{-i})=0$;
    \item let $v_i,v'_i$ be valuations by player $i$ and $b_i=\beta_i(v_i)$, $b'_i=\beta_i(v'_i)$. If $v_i<v'_i$ and $b_i>b'_i$, then $H_i(b_i,\vec{\beta}_{-i})=H_i(b'_i,\vec{\beta}_{-i})=0$.
\end{enumerate}
\end{proposition}

\begin{proof}~
\begin{enumerate}
    \item If $b_i>v_i$ and $H_i(b_i,\vec{\beta}_{-i})>0$, then player $i$ achieves a strictly negative utility by bidding $b_i$ when her valuation is $v_i$. However, player $i$ could achieve non-negative utility by bidding below $v_i$ (e.g.\ by bidding $0$). Hence $\vec{\beta}$ would not be an equilibrium.
    \item Suppose that $v_i<v'_i$ and $b_i>b'_i$. As $\vec{\beta}$ is an exact equilibrium, we know that $b_i,b'_i$ are the best bidding responses by player $i$. In other words, $u_i(b_i,\vec{\beta}_{-i};v_i)\geq u_i(b'_i,\vec{\beta}_{-i};v_i)$ and $u_i(b'_i,\vec{\beta}_{-i};v'_i)\geq u_i(b_i,\vec{\beta}_{-i};v'_i)$. Moreover, as $b_i>b'_i$ we also know that $H_i(b_i,\vec{\beta}_{-i})\geq H_i(b'_i,\vec{\beta}_{-i})$. Putting these together, we find that
    \begin{align*}
        u_i(b_i,\vec{\beta}_{-i};v'_i)+u_i(b'_i,\vec{\beta}_{-i};v_i)&=(v'_i-b_i)H_i(b_i,\vec{\beta}_{-i})+(v_i-b'_i)H_i(b'_i,\vec{\beta}_{-i})\\
        &=(v'_i-v_i)H_i(b_i,\vec{\beta}_{-i})+(v_i-b_i)H_i(b_i,\vec{\beta}_{-i})+(v_i-b'_i)H_i(b'_i,\vec{\beta}_{-i})\\
        &\geq (v'_i-v_i)H_i(b'_i,\vec{\beta}_{-i})+(v_i-b_i)H_i(b_i,\vec{\beta}_{-i})+(v_i-b'_i)H_i(b'_i,\vec{\beta}_{-i})\\
        &=(v'_i-b'_i)H_i(b'_i,\vec{\beta}_{-i})+(v_i-b_i)H_i(b_i,\vec{\beta}_{-i})\\
        &=u_i(b'_i,\vec{\beta}_{-i};v'_i)+u_i(b_i,\vec{\beta}_{-i};v_i).
    \end{align*}
    From this, we conclude that all steps in the above derivation must hold with equality, implying that $u_i(b_i,\vec{\beta}_{-i};v_i) = u_i(b'_i,\vec{\beta}_{-i};v_i)$, $u_i(b'_i,\vec{\beta}_{-i};v'_i) = u_i(b_i,\vec{\beta}_{-i};v'_i)$ and $H_i(b_i,\vec{\beta}_{-i}) = H_i(b'_i,\vec{\beta}_{-i})$. But then $0=u_i(b'_i,\vec{\beta}_{-i};v_i)-u_i(b_i,\vec{\beta}_{-i};v_i)=(b_i-b'_i)H_i(b_i,\vec{\beta}_{-i})$. As $b_i>b'_i$ we conclude that $H_i(b_i,\vec{\beta}_{-i})=H_i(b'_i,\vec{\beta}_{-i})=0$.
\end{enumerate}
\end{proof}

We are now ready to proceed with the example showing that all equilibria of the FPA can be irrational. Consider a first-price auction with $n=3$ bidders and common priors, whose valuations are independently and identically distributed according to the uniform distribution on $[0,1]$; that is, $F_i(x)=x$ for $i=1,2,3$. Let the bidding space be $B=\{0,1/2\}$. Clearly, this auction can be represented with piecewise-constant density functions (with a single piece) and with a finite number of rational quantities. We shall show that the auction has a unique equilibrium, and that this equilibrium is described by an irrational jump point.

First observe that, at an exact equilibrium, the probability of a player winning when bidding $0$ is positive. Otherwise, one of the other players would be bidding $1/2$ with probability $1$, and would achieve expected negative utility when having a valuation in $[0,1/2)$, which contradicts the best-response conditions. Since the probability of winning is never zero, \cref{prop:overbidnonmonotone} implies that any equilibrium must consist of non-overbidding, monotone non-decreasing strategies. In particular, the best response strategy of a player $i$ can be described by a single jump point $a_i$, that is,
$$\beta_i(x)=\left\{\begin{array}{cc}0 & \text{ if }0\leq x\leq a_i;\\ 1/2 & \text{ if }a_i<x\leq 1.\end{array}\right.$$

Since strategies are non-overbidding, we must have that $1/2\leq a_i\leq 1$. Moreover, a joint strategy profile can be described by the jump points of each player, which form a triple $(a_1,a_2,a_3)$.

Next we show that, at an exact equilibrium, each of the $a_i$ must be strictly less than 1. Suppose that bidder $1$ has a valuation of $v_1$ and that bidders $2$ and $3$ have played according to $(a_2,a_3)$. This means that bidder $2$ bids $0$ with probability $a_2$ and bids $1/2$ with probability $(1-a_2)$, and similarly for bidder $3$. Thus, the probability of player $1$ winning when bidding $0$ or when bidding $1/2$ is, respectively,
\begin{align*}
H(0;a_2,a_3) &=\frac{1}{3}a_2a_3,\\ H(1/2;a_2,a_3)&=a_2a_3+\frac{1}{2}a_2(1-a_3)+\frac{1}{2}(1-a_2)a_3+\frac{1}{3}(1-a_2)(1-a_3)\\
&=\frac{1}{3}+\frac{1}{3}a_2a_3+\frac{1}{6}a_2+\frac{1}{6}a_3.\end{align*}
From this we can compute the utility of player $1$ when bidding $0$ or when bidding $1/2$,
\begin{align*}u_1(v_1,0;a_2,a_3)&=\frac{1}{3}a_2a_3v_1\\
u_1(v_1,1/2;a_2,a_3)&=\left(\frac{1}{3}+\frac{1}{3}a_2a_3+\frac{1}{6}a_2+\frac{1}{6}a_3\right)\left(v_1-\frac{1}{2}\right).\end{align*}

We can compute the jump point $v_1$ for which player $1$ is indifferent between bidding $0$ or bidding $1/2$, by solving the equation
\begin{align}&u_1(v_1,0;a_2,a_3)=u_1(v_1,1/2;a_2,a_3)\nonumber\\
\Rightarrow\quad&\frac{1}{3}a_2a_3v_1=\left(\frac{1}{3}+\frac{1}{3}a_2a_3+\frac{1}{6}a_2+\frac{1}{6}a_3\right)\left(v_1-\frac{1}{2}\right)\nonumber\\
\Rightarrow\quad& 2v_1+v_1a_2+v_1a_3=1+a_2a_3+\frac{1}{2}a_2+\frac{1}{2}a_3\label{eq:pne3playersv1}\\
\Rightarrow\quad& v_1=\frac{1}{2}+\frac{a_2a_3}{2+a_2+a_3}\label{eq:pne3playersv1b}.\end{align}

Next observe that the expression $\frac{a_2a_3}{2+a_2+a_3}$ is increasing in both $a_2$ and $a_3$; hence, by setting $a_2=1$, $a_3=1$ we get that the right hand side of \eqref{eq:pne3playersv1b} is at most $\frac{1}{2}+\frac{1\times 1}{2+1+1}=\frac{3}{4}$. In other words, the break-even point must occur in the interval $[1/2,3/4]$, and thus in particular setting $v_1=a_1$ must give a solution to \eqref{eq:pne3playersv1}.

Repeating this argument for players $2$ and $3$ we obtain similarly that $a_2$ and $a_3$ must lie in $[1/2,3/4]$, and that these jump points must be the solutions of equations similar to \eqref{eq:pne3playersv1}. In order for $(a_1,a_2,a_3)$ to define an equilibrium, each player's jump point must be optimal in response to the other players' strategies. Thus, $(a_1,a_2,a_3)$ must be a solution of the system of equations

\begin{align}
   2a_1+a_1a_2+a_1a_3&=1+a_2a_3+\frac{1}{2}a_2+\frac{1}{2}a_3\label{eq:pne3playersa1}\\
   2a_2+a_1a_2+a_2a_3&=1+a_1a_3+\frac{1}{2}a_1+\frac{1}{2}a_3\label{eq:pne3playersa2}\\
   2a_3+a_2a_3+a_1a_3&=1+a_1a_2+\frac{1}{2}a_1+\frac{1}{2}a_2 \notag
\end{align}

Finally, we show that the above system has a unique solution. By subtracting \eqref{eq:pne3playersa2} from \eqref{eq:pne3playersa1}, we get
\begin{align*}&2(a_1-a_2)+a_3(a_1-a_2)=(a_2-a_1)a_3+\frac{1}{2}(a_2-a_1)\\
\Rightarrow\quad&\left(\frac{5}{2}+2a_3\right)(a_1-a_2)=0\\
\Rightarrow\quad& a_3=-\frac{5}{4}\quad\text{or}\quad a_1=a_2.\end{align*}
Since we know that $a_3\in[1/2,3/4]$, we conclude that $a_1=a_2$. By the same argument, we must have $a_2=a_3$ and $a_1=a_3$, that is, any equilibrium must be symmetric. Now letting $a:=a_1=a_2=a_3$, we get that $a$ must be a solution to the equation
\begin{align*}&2a+a^2+a^2=1+a^2+\frac{1}{2}a+\frac{1}{2}a\\
\Rightarrow\quad&a^2+a-1=0\\
\Rightarrow\quad& a=\frac{-1\pm \sqrt{5}}{2}.\end{align*}

Since $a$ must be positive, we conclude that the unique equilibrium of this auction is given by the jump point $a=\frac{-1+\sqrt{5}}{2}\approx 0.618$ (the inverse of the golden ratio), which is irrational.

\section{Proof of \texorpdfstring{\cref{lem:util-poly-cont}}{Lemma~\ref*{lem:util-poly-cont}}}\label{app:util-poly-cont}

Since the distributions are polynomially-continuous, it follows that given any $\varepsilon > 0$, we can compute $\delta > 0$ in polynomial time such that $|F_{i,j}(x)-F_{i,j}(y)| \leq \varepsilon/2^{n+1}$ for all $x,y$ with $|x-y| \leq \delta$ and all $i,j \in N$ ($i \neq j$).

Consider any $\vec{\alpha},\vec{\alpha}' \in \mathcal{D}$ (see the proof of \Cref{thm:existence}) with $\|\vec{\alpha} - \vec{\alpha}'\|_\infty \leq \delta$. Then, we have
\begin{equation*}\begin{split}
\left|\probability_{v_i \sim F_{j,i}}\left[\beta_i(v_i) \leq b\right] - \probability_{v_i \sim F_{j,i}}\left[\beta_i'(v_i) \leq b\right]\right| &\leq \left|\probability_{v_i \sim F_{j,i}}\left[v_i \leq \alpha_{i}(b)\right] - \probability_{v_i \sim F_{j,i}}\left[v_i \leq \alpha_{i}'(b)\right]\right|\\
&\leq \left|F_{j,i}(\alpha_{i}(b)) - F_{j,i}(\alpha_{i}'(b))\right|\\
&\leq \varepsilon/2^{n+1}
\end{split}\end{equation*}
for all $i,j \in N$ ($i \neq j$) and $b \in B$. It follows that $\probability_{v_i \sim F_{j,i}}\left[\beta_i(v_i) < b\right]$ differs from $\probability_{v_i \sim F_{j,i}}\left[\beta_i'(v_i) < b\right]$ by at most $\varepsilon/2^{n+1}$. Similarly, $\probability_{v_i \sim F_{j,i}}\left[\beta_i(v_i) = b\right]$ differs from $\probability_{v_i \sim F_{j,i}}\left[\beta_i'(v_i) = b\right]$ by at most $\varepsilon/2^n$.

Let $T_i(b,\ell;\vec{\alpha}_{-i})$ denote the probability that, from the perspective of bidder $i$, exactly $\ell$ out of the bidders $N \setminus \{i\}$ bid exactly $b$, and the remaining $n-1-\ell$ bidders bid below $b$. We can write
$$T_i(b,\ell;\vec{\alpha}_{-i}) = \sum_{\substack{S\subseteq N \setminus \{i\}\\ |S|=\ell}} \prod_{k \in S} \probability_{v_k \sim F_{i,k}}\left[\beta_k(v_k) = b\right] \prod_{k \in N \setminus (\{i\} \cup S)} \probability_{v_k \sim F_{i,k}}\left[\beta_k(v_k) < b\right].$$
From this it follows that $T_i(b,\ell;\vec{\alpha}_{-i})$ and $T_i(b,\ell;\vec{\alpha}_{-i}')$ differ by at most $\binom{n-1}{\ell} n \varepsilon/2^n$, for all $i \in N$, $b \in B$ and $\ell \in \{0,1,\dots, n-1\}$. As defined in \cref{sec:bestresponses}, recall that $H_i(b,\vec{\alpha}_{-i})$ denotes the probability that bidder $i$ wins if she bids $b$ and the other bidders bid according to $\vec{\alpha}_{-i}$. Then, we can write
$$H_i(b,\vec{\alpha}_{-i}) = \sum_{\ell=0}^{n-1} \frac{1}{\ell+1} T_i(b,\ell;\vec{\alpha}_{-i}).$$
It follows that $H_i(b,\vec{\alpha}_{-i})$ differs from $H_i(b,\vec{\alpha}_{-i}')$ by at most
$$\sum_{\ell=0}^{n-1} \frac{1}{\ell+1} \binom{n-1}{\ell} n \varepsilon/2^n = \sum_{\ell=0}^{n-1} \binom{n}{\ell+1} \varepsilon/2^n \leq \varepsilon$$
for all $i \in N$ and $b \in B$. Finally, note that $u_i(b,\vec{\alpha}_{-i};v_i) = H_i(b,\vec{\alpha}_{-i}) \cdot (v_i-b)$. Thus, we obtain
$$\left|u_i(b,\vec{\alpha}_{-i};v_i) - u_i(b,\vec{\alpha}_{-i}';v_i)\right| \leq \left|H_i(b,\vec{\alpha}_{-i}) - H_i(b,\vec{\alpha}_{-i}')\right| \left|v_i-b\right| \leq \varepsilon$$
for all $i \in N$, $b \in B$ and $v_i \in [0,1]$, since $|v_i-b| \leq 1$.

\section{PPAD and FIXP-completeness of Generalized Circuit Variants}

\subsection{PPAD-completeness (Proof of \texorpdfstring{\cref{prop:gcircuit-ppad}}{Proposition~\ref*{prop:gcircuit-ppad}})}\label{app:gcircuit-ppad}

Membership in \ppad follows from the fact that a generalized circuit with gates $g_1, \dots, g_\nn$ can be interpreted as defining an algebraic circuit $F: [0,1]^\nn \to [0,1]^\nn$, where for $x \in [0,1]^\nn$ and $i \in [\nn]$ we let $F_i(x) = G(x_j,x_k)$, where $g_i=(G,j,k)$. Then, it is known that the problem of computing an $\varepsilon$-approximate fixed point of such a function $F$ lies in \ppad \citep{etessami2010complexity} (and in fact, even when $\varepsilon$ is provided in the input in binary representation). Finally, note that an $\varepsilon$-approximate fixed point of $F$ exactly corresponds to an $\varepsilon$-satisfying assignment for the generalized circuit.

In order to prove \ppad-hardness, consider the $\varepsilon$-\gcircuit problem with gate-types $\mathcal{G} = \{G_{1-},G_+\}$, for some sufficiently small constant $\varepsilon > 0$ (which will be set later). We begin by showing that additional gate-types can be simulated if we allow a larger (but still constant) error.

\medskip

\noindent\textbf{$\bm{G_{=}}$: Copy.} The goal of such a gate is to copy the value of some gate $g_1$. For this, we use the fact that $1-(1-x) = x$. Thus, we introduce a gate $g_2$ of type $G_{1-}$ with input $g_1$ and a gate $g_3$ of type $G_{1-}$ with input $g_2$. It holds that $\val[g_3] = 1 - \val[g_2] \pm \varepsilon = \val[g_1] \pm 2\varepsilon$. In other words, we can simulate a copy gate with error at most $2\varepsilon$.

\medskip

\noindent\textbf{$\bm{G_1}$: Constant 1.} In order to obtain a gate that has value $1$, we use the fact that $x + (1-x) = 1$. First, we introduce an arbitrary gate $g_1$. Then, we introduce a gate $g_2$ of type $G_{1-}$ with input $g_1$, and a gate $g_3$ of type $G_+$ with inputs $g_1$ and $g_2$. It holds that $\val[g_3] = \trunc(\val[g_1]+\val[g_2]) \pm \varepsilon = 1 \pm 2\varepsilon$. Thus, we can simulate a constant $1$ with error at most $2\varepsilon$.

\medskip

\noindent\textbf{$\bm{G_-}$: Subtraction.} The goal of this gate is to compute $\trunc(\val[g_1]-\val[g_2])$. For this, we use the identity
$$\trunc(x-y) = 1 - \trunc\big((1-x) + y\big)$$
which allows us to express subtraction using only addition and the complement operation. With this in hand, we can implement subtraction as follows. We introduce a gate $g_3$ of type $G_{1-}$ with input $g_1$, a gate $g_4$ of type $G_+$ with inputs $g_3$ and $g_2$, and finally a gate $g_5$ of type $G_{1-}$ with input $g_4$. Then, it holds that $\val[g_5] = 1 - \val[g_4] \pm \varepsilon = 1 - \trunc(\val[g_3]+\val[g_2]) \pm 2\varepsilon = 1 - \trunc(1-\val[g_1]+\val[g_2]) \pm 3\varepsilon = \trunc(\val[g_1]-\val[g_2]) \pm 3\varepsilon$. Thus, we can simulate a subtraction gate with error at most $3\varepsilon$.

\medskip

\noindent\textbf{$\bm{G_{/2}}$: Division by 2.} The goal of this gate is to compute $\val[g_1]/2$. This is achieved by constructing a cycle. Namely, we introduce two gates $g_2$ and $g_3$. The gate $g_2$ is of type $G_{-}$ with inputs $g_1$ and $g_3$, and the gate $g_3$ is of type $G_{=}$ with input $g_2$. As a result, it holds that
$$\val[g_3] = \val[g_2] \pm 2\varepsilon = \trunc(\val[g_1] - \val[g_3]) \pm 5\varepsilon.$$
From this, it follows that $\val[g_3] = \val[g_1]/2 \pm 5\varepsilon$. To see this, note that if $\val[g_3] \geq \val[g_1]$, then $\val[g_3] = 0 \pm 5\varepsilon = \val[g_1]/2 \pm 5\varepsilon$, since $[0,5\varepsilon] \subseteq [\val[g_1]/2-5\varepsilon,\val[g_1]/2+5\varepsilon]$ (because $\val[g_1]/2 \leq \val[g_3]/2 \leq 5\varepsilon$). On the other hand, if $\val[g_3] < \val[g_1]$, then we obtain that $2\val[g_3] = \val[g_1] \pm 5\varepsilon$, which again yields the same conclusion, namely $\val[g_3] = \val[g_1]/2 \pm 5\varepsilon$. Thus, we can simulate division by $2$ with error at most $5\varepsilon$.

\medskip

\noindent\textbf{$\bm{G_{\times \zeta}}$: Multiplication by $\bm{\zeta \in \left[0,1\right]}$.} If $\zeta = 0$, then we can simply output $G_{1-}(G_1) = 0 \pm 3\varepsilon$. If $\zeta=1$, we can simply use a $G_{=}$ gate that has error at most $2\varepsilon$. Consider now the case where $\zeta \in (0,1)$. Let $k=\lceil \log_2(1/\varepsilon) \rceil$. Recall that $\varepsilon$ will be a fixed constant, so $k$ will also be a fixed constant. It is easy to see that in polynomial time (in the representation size of $\zeta$) we can find $a \in \{1,2,\dots, 2^k-1\}$ such that $|\zeta - a/2^k| \leq \varepsilon$.

Let $g_1$ denote the input. Our goal now is to compute $(a/2^k) \cdot \val[g_1]$, since this will be $\varepsilon$-close to $\zeta \cdot \val[g_1]$. We compute $(a/2^k) \cdot \val[g_1]$ in a careful manner to ensure that the error remains small. This is achieved as follows. Using the binary representation of $a = \sum_{i=0}^{k-1} a_i 2^i$, $a_i \in \{0,1\}$, we can express the product $(a/2^k) \cdot x$ as
\begin{equation*}
\begin{tabular}{cc}
   $\frac{\frac{0+a_0\frac{x}{2}}{2} + a_1\frac{x}{2}}{2} + a_2\frac{x}{2}$  &  \\
     & $\ddots$
\end{tabular}
\end{equation*}
We implement this as follows. First, introduce $g_2$ such that $\val[g_2] = \val[g_1]/2 \pm 5\varepsilon$.
Next, introduce $g_3$ such that (i) if $a_0=0$, then $\val[g_3] = 0 \pm 3\varepsilon$, (ii) if $a_0=1$, then $g_3=g_2$. In both cases we have
$$\val[g_3] = a_0 \val[g_2] \pm 3\varepsilon.$$
Next, introduce $g_4$ such that (i) if $a_1=0$, then $\val[g_4] = \val[g_3]/2 \pm 5\varepsilon = a_0 \val[g_2]/2 \pm 5(1+1/2)\varepsilon$, (ii) if $a_1=1$, then $\val[g_4] = \val[g_3]/2 + \val[g_2] \pm 6\varepsilon = a_0 \val[g_2]/2 + \val[g_2] \pm 6(1+1/2)\varepsilon$. In both cases we have
$$\val[g_4] = a_0 \val[g_2]/2 + a_1\val[g_2] \pm 6(1+1/2)\varepsilon = (a_0+2a_1)\val[g_2]/2 \pm 6(1+1/2)\varepsilon.$$
Next, introduce $g_5$ such that (i) if $a_2=0$, then $\val[g_5] = \val[g_4]/2 \pm 5\varepsilon = (a_0+2a_1)\val[g_2]/4 \pm 6(1+1/2+1/4)\varepsilon$, (ii) if $a_2=1$, then $\val[g_5] = \val[g_4]/2 + \val[g_2] \pm 6\varepsilon = (a_0+2a_1)\val[g_2]/4 + \val[g_2] \pm 6(1+1/2+1/4)\varepsilon$. In both cases we have
\begin{equation*}\begin{split}
\val[g_5] &= (a_0+2a_1)\val[g_2]/4 + a_2\val[g_2] \pm 6(1+1/2+1/4)\varepsilon\\
&= (a_0+2a_1+4a_2)\val[g_2]/4 \pm 6(1+1/2+1/4)\varepsilon.
\end{split}\end{equation*}
Continuing in the same manner, it follows by induction that after $k-1$ such steps we obtain 
$$\val[g_{k+2}] = \left(\sum_{i=0}^{k-1} a_i 2^i\right)\val[g_2]/2^{k-1} \pm 12\varepsilon = \frac{a}{2^k} \left(2\val[g_2]\right) \pm 12\varepsilon = \frac{a}{2^k} \val[g_1] \pm 22\varepsilon = \zeta \cdot \val[g_1] \pm 23\varepsilon.$$
Thus, we can compute multiplication by $\zeta \in [0,1]$ with error at most $23\varepsilon$. Note that this gadget can be constructed in polynomial time in the representation size of $\zeta$. Furthermore, the number of gates needed to construct the gadget is $O(k)$, which is constant, since $k=\lceil \log_2(1/\varepsilon) \rceil$ and $\varepsilon$ will be a fixed constant.

\medskip

\noindent We are now ready to show \ppad-hardness. To do this, we reduce from a slightly modified version of \gcircuit studied by \citet{goldberg2020consensus}, that we call $\gcircuit^{[-1,1]}$. This modified version operates on $[-1,1]$ instead of $[0,1]$, and it uses the gates $G_+^{[-1,1]}$, $G_1^{[-1,1]}$ and $G_{\times -\zeta}^{[-1,1]}$ (where the gates truncate to $[-1,1]$, and $\zeta \in [0,1]$). \citet{goldberg2020consensus} proved that $\varepsilon'$-$\gcircuit^{[-1,1]}$ is \ppad-hard for some sufficiently small constant $\varepsilon' > 0$. We now set $\varepsilon := \varepsilon'/50$. Below, we show that $\varepsilon'$-$\gcircuit^{[-1,1]}$ reduces to $\varepsilon$-\gcircuit (with gate-types $\mathcal{G} = \{G_{1-},G_+\}$).

Given a generalized circuit with gates $G_+^{[-1,1]}$, $G_1^{[-1,1]}$ and $G_{\times -\zeta}^{[-1,1]}$, we construct a corresponding circuit with gates $G_{1-}$ and $G_+$ as follows. Every gate $g$ of the original circuit is replaced by two gates $g^+$ and $g^-$. The idea is that the value of $g$, which lies in $[-1,1]$, will be encoded by the values of $g^+$ and $g^-$, which lie in $[0,1]$. Formally, we interpret $\val[g] := \val[g^+] - \val[g^-]$. Next, we show that the constraints of the original circuit can be enforced by corresponding constraints on the new circuit.

\medskip

\noindent\textbf{Simulating $\bm{G_1^{[-1,1]}}$.} In order to enforce that $\val[g] = 1 \pm \varepsilon'$, we proceed as follows. We simply let $\val[g^+] = 1 \pm 2\varepsilon$ and $\val[g^-] = 0 \pm 3\varepsilon$ (using the constructions described above). Thus, it holds that $\val[g] = \val[g^+] - \val[g^-] = 1 \pm 5\varepsilon = 1 \pm \varepsilon'$.

\medskip

\noindent\textbf{Simulating $\bm{G_{\times -\zeta}^{[-1,1]}}$.} In order to enforce that $\val[g_2] = -\zeta \cdot \val[g_1] \pm \varepsilon'$, for some $\zeta \in [0,1]$, we proceed as follows. Using the constructions described above, we can enforce that $\val[g_2^+] = \zeta \cdot \val[g_1^-] \pm 23\varepsilon$ and $\val[g_2^-] = \zeta \cdot \val[g_1^+] \pm 23\varepsilon$. Thus, $\val[g_2] = -\zeta \cdot \val[g_1] \pm 46\varepsilon = -\zeta \cdot \val[g_1] \pm \varepsilon'$.

\medskip

\noindent\textbf{Simulating $\bm{G_{+}^{[-1,1]}}$.} In order to enforce that $\val[g_3] = \trunc_{[-1,1]}(\val[g_1] + \val[g_2]) \pm \varepsilon'$, we proceed in two steps. First, using our construction for performing subtraction, we ``normalize'' the gates by letting $\val[h_1^+] = \trunc(\val[g_1^+] - \val[g_1^-]) \pm 3\varepsilon$ and $\val[h_1^-] = \trunc(\val[g_1^-] - \val[g_1^+]) \pm 3\varepsilon$, which yields $\val[h_1] = \val[g_1] \pm 6\varepsilon$. We similarly obtain $h_2$ from $g_2$. This ``normalization'' will ensure that addition is then performed correctly.

In the second step, using the addition gate $G_+$, we let $\val[g_3^+] = \trunc(\val[h_1^+] + \val[h_2^+]) \pm \varepsilon$ and $\val[g_3^-] = \trunc(\val[h_1^-] + \val[h_2^-]) \pm \varepsilon$. Thus, it holds that
\begin{equation*}\begin{split}
\val[g_3] = \val[g_3^+] - \val[g_3^-] &= \trunc(\val[h_1^+] + \val[h_2^+]) - \trunc(\val[h_1^-] + \val[h_2^-]) \pm 2\varepsilon\\
&= \trunc_{[-1,1]}(\val[h_1^+] + \val[h_2^+]) - \trunc_{[-1,1]}(\val[h_1^-] + \val[h_2^-]) \pm 2\varepsilon.
\end{split}\end{equation*}
Because of the ``normalization'' step, we know that
$$\min\{\val[h_1^+],\val[h_1^-]\} \leq 3\varepsilon \quad \text{ and } \quad \min\{\val[h_2^+],\val[h_2^-]\} \leq 3\varepsilon.$$
In the case where $\val[h_1^-] \leq 3\varepsilon$ and $\val[h_2^-] \leq 3\varepsilon$, it holds that $\val[h_1] = \val[h_1^+] \pm 3\varepsilon$ and $\val[h_2] = \val[h_2^+] \pm 3\varepsilon$, which implies that
$$\val[g_3] = \trunc_{[-1,1]}(\val[h_1] + \val[h_2]) - \trunc_{[-1,1]}(\val[h_1^-] + \val[h_2^-]) \pm 8\varepsilon = \trunc_{[-1,1]}(\val[h_1] + \val[h_2]) \pm 14\varepsilon.$$
In the case where $\val[h_1^+] \leq 3\varepsilon$ and $\val[h_2^-] \leq 3\varepsilon$, it holds that $\val[h_1] = - \val[h_1^-] \pm 3\varepsilon$ and $\val[h_2] = \val[h_2^+] \pm 3\varepsilon$, which implies that
\begin{equation*}\begin{split}
\val[g_3] = \trunc_{[-1,1]}(\val[h_1^+] + \val[h_2]) - \trunc_{[-1,1]}(-\val[h_1] + \val[h_2^-]) \pm 8\varepsilon &= \val[h_1] + \val[h_2] \pm 14\varepsilon\\
&= \trunc_{[-1,1]}(\val[h_1] + \val[h_2]) \pm 14\varepsilon.
\end{split}\end{equation*}
The remaining two cases are handled in the same way, and thus we always obtain that
$$\val[g_3] = \trunc_{[-1,1]}(\val[h_1] + \val[h_2]) \pm 14\varepsilon = \trunc_{[-1,1]}(\val[g_1] + \val[g_2]) \pm 26\varepsilon = \trunc_{[-1,1]}(\val[g_1] + \val[g_2]) \pm \varepsilon'.$$

\bigskip

\noindent Clearly, this construction can be performed in polynomial time in the size of the original generalized circuit. Furthermore, given any $\varepsilon$-satisfying assignment of the new generalized circuit, we can easily obtain an $\varepsilon'$-satisfying assignment of the original generalized circuit by setting $\val[g] : = \val[g^+] - \val[g^-] \in [-1,1]$ for all gates $g$. It follows that the $\varepsilon$-\gcircuit problem with gate-types $\mathcal{G} = \{G_{1-},G_+\}$ is \ppad-hard.

Finally, note that if we let $\mathcal{G} = \{G_1,G_-\}$ instead, we again obtain the same result, because $G_{1-}$ and $G_+$ can easily be simulated. Indeed, it is clear that $G_{1-}$ can immediately be simulated. Furthermore, $G_+$ can be simulated by using the equation $\trunc(x+y) = 1-\trunc((1-x)-y)$.

\subsection{FIXP-completeness (Proof of  \texorpdfstring{\cref{prop:gcircuit-fixp}}{Proposition~\ref*{prop:gcircuit-fixp}})}\label{app:gcircuit-fixp}

Membership in \fixp follows immediately by noting that a generalized circuit with gates $g_1, \dots, g_\nn$ defines an algebraic circuit $F: [0,1]^\nn \to [0,1]^\nn$, where for $x \in [0,1]^\nn$ and $i \in [\nn]$ we let $F_i(x) = G(x_j,x_k)$, where $g_i=(G,j,k)$. Indeed, any fixed point of $F$ corresponds to an assignment that exactly satisfies the gate constraints. In particular, note that all the gate-types we consider can be exactly computed using the usual operations allowed in \fixp, namely $+, \times, \max$ and rational constants. Furthermore, it is easy to see that this trivially yields an SL-reduction \citep{etessami2010complexity}.

In order to prove \fixp-hardness we will show that our very restricted set of gates is actually enough to simulate various more complex gates. \citet[Section 7.2]{deligkas2019computing}, using a special Brouwer function for the \fixp-complete problem 3-Nash given by \citet{etessami2010complexity}, proved that computing fixed points of very restricted algebraic circuits is already \fixp-hard. In more detail, they consider functions $F: [0,1]^n \to [0,1]^n$ computed by circuits with a restricted set of gates and such that every gate always has value in $[0,1]$, for any input $x \in [0,1]^n$ to the circuit. Because of this property we can use our gates that truncate to $[0,1]$ without changing any of the computations.

In more detail, they allow the following gates: $G_\zeta$ (constant $\zeta \in \mathbb{Q} \cap [0,1]$), $G_+$, $G_-$ (subtraction truncated to $[0,1]$), $G_\times$, $G_{\times 2}^{[0,1]}$, $G_{\max}$ and $G_{\min}$. We show below that we can simulate all of these gates, using only the gates $G_{1-}$, $G_{\times 2}$ and $G_{\times}$ (or alternatively, $G_{1-}$, $G_+$ and $G_{(\cdot)^2}$). In particular, $G_{\times 2}^{[0,1]}$ is a restricted gate $G_{\times 2}$ that only works on inputs in $[0,1/2]$. Since our $G_{\times 2}$ gate has the same behavior as that gate for such inputs, it is correctly simulated.

Finally, we simply use copy gates $G_{=}$ to enforce the fixed point constraint, namely that the $i$th input to $F$ be equal to its $i$th output. It is easy to see that this construction yields a polynomial-time reduction, and that it is in fact an SL-reduction \citep{etessami2010complexity}, since we only need to extract the values assigned to the input gates in order to obtain a fixed point of $F$. In the remainder of this proof, we show how all the required gates can be simulated using our restricted set of gates $G_{1-}$, $G_{\times 2}$ and $G_{\times}$.

\medskip

\noindent\textbf{$\bm{G_{=}}$: Copy.} In order to copy the value of some gate $g_1$, we use the complement gate $G_{1-}$ twice. Namely, we first introduce a gate $g_2$ of type $G_{1-}$ with input $g_1$, and then another gate $g_3$ of type $G_{1-}$ with input $g_2$. Clearly it holds that $\val[g_3]=1-\val[g_2] = 1-(1-\val[g_1]) = \val[g_1]$.

\medskip

\noindent\textbf{$\bm{G_{1/2}}$: Constant $\bm{1{/}2}$.} In order to obtain a gate that has value $1/2$, we create a small cycle. We introduce two gates $g_1$ and $g_2$. The gate $g_1$ is of type $G_{=}$ with input $g_2$, and the gate $g_2$ is of type $G_{1-}$ with input $g_1$. It follows that $\val[g_1]$ satisfies $\val[g_1] = 1 - \val[g_1]$, which implies $\val[g_1]=1/2$. Note that together with the $G_\times$ gate we can now also perform multiplication by $1/2$, denoted by $G_{\times 1/2}$.

\medskip

\noindent\textbf{$\bm{G_-}$: Subtraction.} In the proof of \cref{lem:gcircuit-special}, we show how to construct a subtraction gate given access only to $G_{1-}$, $G_{\times 2}$ and a special gate $G_\phi$, where $\phi: [0,1]^2 \to [0,1]$, $(x,y) \mapsto (x+1)(y+1)/4$. Thus, to obtain the subtraction gate, it is enough for us here to construct a gate $G_\phi$. Since we have access to $G_\times$, it suffices to construct a gate that implements the function $x \mapsto (x+1)/2$. Let $g_1$ be the input gate. We introduce a gate $g_2$ of type $G_{1-}$ with input $g_1$, a gate $g_3$ of type $G_{\times 1/2}$ with input $g_2$, and finally a gate $g_4$ of type $G_{1-}$ with input $g_3$. It follows that $\val[g_4]=1-\val[g_3]=1-\val[g_2]/2 = 1-(1-\val[g_1])/2 = (\val[g_1]+1)/2$, as desired.

\medskip

\noindent\textbf{$\bm{G_{+}}$: Addition.} Addition can easily be obtained from subtraction by using the following equality for all $x,y \in [0,1]$
$$\trunc(x+y) = 1-\trunc((1-x)-y) = G_{1-}(G_-(G_{1-}(x),y)).$$

\medskip

\noindent\textbf{$\bm{G_{\max}}, \bm{G_{\min}}$: Maximum and Minimum.} The function $(x,y) \mapsto \max\{x,y\}$ can easily be simulated with existing gates by noting that
$$\max\{x,y\} = \trunc(x + \trunc(y-x)) = G_+(x,G_-(y,x)).$$
Then, $(x,y) \mapsto \min\{x,y\}$ can simply be obtained by $\min\{x,y\} = 1 - \max\{1-x,1-y\}$.

\medskip

\noindent\textbf{$\bm{G_{\times k}}$: Multiplication by integer $\bm{k}$.} Let $k$ be an integer that is given in binary representation, i.e., $k = \sum_{i=0}^\ell a_i 2^i$, where $a_i \in \{0,1\}$. Our goal is to construct a gate that computes $x \mapsto \trunc(k \cdot x)$. Using the $G_{\times 2}$ gate we can compute $\trunc(2^i \cdot x)$ for $i=0,1,\dots,\ell$. This requires $\ell$ separate $G_{\times 2}$ gates. Then, we use the addition gate to compute
$$\trunc\left(\sum_{i:\, a_i=1} \trunc(2^i \cdot x)\right) = \trunc\left(\sum_{i=0}^\ell a_i 2^i x\right) = \trunc(k \cdot x).$$
This uses at most $\ell$ separate $G_+$ gates. Thus, overall we use a number of gates that is polynomial in the representation length of $k$.

\medskip

\noindent\textbf{$\bm{G_{\zeta}}$: Constant $\bm{\zeta \in \left[0,1\right] \cap} \Q$.} If $\zeta = 1$, we can simply do $G_{\times 2}(G_{1/2}) = 1$. If $\zeta = 0$, we can do $G_{1-}(G_{\times 2}(G_{1/2})) = 0$. Now assume that $\zeta \in (0,1)$. Write $\zeta = c/d$ where $c$ and $d$ are positive integers, $c \geq 1$, $c < d$, $d \geq 2$. Clearly, if we can construct the constant $1/d$, then we can use a $G_{\times k}$ gate with $k=c$ to obtain $\zeta$. In order to construct $1/d$, we use a small cycle. We introduce two gates $g_1$ and $g_2$. The gate $g_1$ is of type $G_{\times k}$ with $k=d-1$ and with input $g_2$. The gate $g_2$ is of type $G_{1-}$ with input $g_1$. Thus, it holds that $\val[g_2] = 1 - \val[g_1] = 1 - \trunc((d-1) \cdot \val[g_2])$. It is easy to check that the only solution of this equation is $\val[g_2] = 1/d$.

\medskip

\noindent Finally, let us show that the set of gate-types $G_{1-}$, $G_+$ and $G_{(\cdot)^2}$ also suffices to simulate all the gates above, by showing that they can simulate $G_{1-}$, $G_{\times 2}$ and $G_\times$. As before, $G_{1-}$ can be used to create $G_{=}$. Then, $G_+$ and $G_{=}$ can be used to obtain $G_{\times 2}$. Thus, it remains to simulate $G_{\times}$.

Note that $G_-$ can be obtained by $\trunc(x-y) = 1 - (\trunc((1-x) + y))$. Furthermore, we can construct $G_{\times 1/2}$ on input gate $g_1$ as follows. We introduce two gates $g_2$ and $g_3$. The gate $g_2$ is of type $G_-$ and has inputs $g_1$ and $g_3$. The gate $g_3$ is of type $G_{=}$ with input $g_2$. It follows that $\val[g_3] = \trunc(\val[g_1] - \val[g_3])$, which has the only solution $\val[g_3] = \val[g_1]/2$.

In order to simulate $G_\times$, note that
$$\left(\frac{x}{2} + \frac{y}{2}\right)^2 = (x/2)^2+(y/2)^2 + xy/2.$$
We can easily compute $x/2+y/2$ and then square using $G_{(\cdot)^2}$. Similarly, we can also compute $(x/2)^2+(y/2)^2$. By using $G_-$, we then obtain $xy/2$, and thus $xy$ after using a $G_{\times 2}$ gate.

\subsection{Proof of \texorpdfstring{\cref{lem:gcircuit-special}}{Lemma~\ref*{lem:gcircuit-special}}}\label{app:gcircuit-special}

In order to prove that the problem remains hard with $\mathcal{G} = \{G_{\times 2},G_{1-},G_\phi\}$, we will show that other gate-types can be simulated using only these three gate-types. Let $\varepsilon \in [0,1/14]$ and assume that we have access to gates of type $G_{\times 2}$, $G_{1-}$ and $G_\phi$.

\medskip

\noindent\textbf{$\bm{G_1}$: Constant 1.} In order to create a constant $1$ we use the fact that for any $x,y \in [0,1]$
$$\trunc\left(2^3 \cdot \phi(x,y)\right) = \trunc \left( 2^3 (x+1)(y+1)/4 \right) \geq \trunc(2) = 1.$$
In more detail, we use a gate $g_1$ of type $G_\phi$ (with arbitrary inputs), then a gate $g_2$ of type $G_{\times 2}$ with input $g_1$, another gate $g_3$ of type $G_{\times 2}$ with input $g_2$, and finally another gate $g_4$ of type $G_{\times 2}$ with input $g_3$. We have that $\val[g_1] \geq 1/4 - \varepsilon$, $\val[g_2] \geq \trunc(2 \cdot \val[g_1]) - \varepsilon \geq 1/2 - 3\varepsilon$, $\val[g_3] \geq \trunc(2 \cdot \val[g_2]) - \varepsilon \geq 1 - 7\varepsilon$, and $\val[g_4] \geq \trunc(2 \cdot \val[g_3]) - \varepsilon \geq 1 - \varepsilon$, since $\varepsilon \leq 1/14$. Thus, we can construct a gate that has the value $1 \pm \varepsilon$.

\medskip

\noindent\textbf{$\bm{G_{/2}}$: Division by 2.} In order to divide the value of some gate $g_1$ by $2$, we use the fact that
$$1-\phi(1-\val[g_1],1) = 1 - (2-\val[g_1])(1+1)/4 = \val[g_1]/2.$$
In more detail, we use a gate $g_2$ of type $G_{1-}$ with input $g_1$, then we use a gate $g_3$ of type $G_{\phi}$ with inputs $g_2$ and a constant $1 \pm \varepsilon$, and finally we use a gate $g_4$ of type $G_{1-}$ with input $g_3$. It holds that $\val[g_2] = 1 - \val[g_1] \pm \varepsilon$, $\val[g_3] = \phi(\val[g_2],1 \pm \varepsilon) \pm \varepsilon = 1 - \val[g_1]/2 \pm 2\varepsilon$, and $\val[g_4] = 1 - \val[g_3] \pm \varepsilon = \val[g_1]/2 \pm 3\varepsilon$. Thus, we can construct a gate that performs division by $2$ with error at most $3\varepsilon$.

\medskip

\noindent\textbf{$\bm{G_{=}}$: Copy.} It is easy to see that using two gates of type $G_{1-}$, one after the other, copies the original value with error at most $2\varepsilon$.

\medskip

\noindent\textbf{$\bm{G_{inv}}$: Inverse.} We now show how to construct the gate $G_{inv}$, which computes the function $x \mapsto -1 + 4/(2+x)$, and will be very useful to construct the subtraction gate below. The construction of $G_{inv}$ uses a cycle. Let $g_1$ be the input gate. We first use a gate $g_2$ of type $G_{1-}$ with input $g_1$, then we use a gate $g_3$ of type $G_\phi$ with input $g_2$ and $g_4$, and finally we let gate $g_4$ be of type $G_{=}$ with input $g_3$. We have that $\val[g_2] = 1 - \val[g_1] \pm \varepsilon$, $\val[g_3] = \phi(\val[g_2],\val[g_4]) \pm \varepsilon$, and $\val[g_4] = \val[g_3] \pm 2\varepsilon$. It follows that $\val[g_4]$ must satisfy the equation
$$\val[g_4] = \phi(\val[g_2],\val[g_4]) \pm 3\varepsilon = (\val[g_2]+1)(\val[g_4]+1)/4 \pm 3\varepsilon$$
which implies that
$$\val[g_4] = \frac{1+\val[g_2]}{3-\val[g_2]} \pm 6\varepsilon.$$
As a result, we obtain that
$$\val[g_4] = \frac{2-\val[g_1]}{2+\val[g_1]} \pm 8\varepsilon = -1 + \frac{4}{2+\val[g_1]} \pm 8\varepsilon$$
i.e., we can compute the function with error at most $8\varepsilon$.

\medskip

\noindent\textbf{$\bm{G_-}$: Subtraction.} Given gates $g_1$ and $g_2$, we want to obtain $\trunc(\val[g_1]-\val[g_2])$. To achieve this, we first use the fact that
$$\phi\Bigg(\phi\left(-1 + \frac{4}{2+y},1-x\right),\frac{y}{2}\Bigg) = \phi\left(\frac{2-x}{2+y},\frac{y}{2}\right) = \frac{1}{2} + \frac{1}{8}(y-x).$$
In more detail, we first use a gate $g_3$ of type $G_{inv}$ with input $g_2$, then a gate $g_4$ of type $G_{1-}$ with input $g_1$, then a gate $g_5$ of type $G_\phi$ with inputs $g_3$ and $g_4$, then a gate $g_6$ of type $G_{/2}$ with input $g_2$, and finally a gate $g_7$ of type $G_\phi$ with inputs $g_5$ and $g_6$. We thus obtain that $\val[g_3] = -1 + 4/(2+\val[g_2]) \pm 8\varepsilon$, $\val[g_4] = 1 - \val[g_1] \pm \varepsilon$, and $\val[g_5] = (2-\val[g_1])(2+\val[g_2]) \pm 7\varepsilon$. Furthermore, it holds that $\val[g_6] = \val[g_2]/2 \pm 3\varepsilon$, and thus $\val[g_7] = 1/2 + (\val[g_2]-\val[g_1])/8 \pm 11\varepsilon$.

Next, we can obtain the subtraction operation from this by noting that
$$4\Bigg(1-\trunc\Bigg(2\left(\frac{1}{2} + \frac{1}{8}(y-x)\right)\Bigg)\Bigg) = 4\Bigg(1-\bigg(1 - \frac{1}{4}\trunc(x-y)\bigg)\Bigg) = 4 \frac{\trunc(x-y)}{4} = \trunc(x-y).$$
This is implemented by using a gate $g_8$ of type $G_{\times 2}$ with input $g_7$, then a gate $g_9$ of type $G_{1-}$ with input $g_8$, then a gate $g_{10}$ of type $G_{\times 2}$ with input $g_9$, and finally another gate $g_{11}$ of type $G_{\times 2}$ with input $g_{10}$. It holds that
$$\val[g_8] = \trunc(2 \cdot \val[g_7]) \pm \varepsilon = 1 - \trunc(\val[g_1]-\val[g_2])/4 \pm 23\varepsilon.$$
As a result, it then holds that $\val[g_9] = \trunc(\val[g_1]-\val[g_2])/4 \pm 24\varepsilon$, $\val[g{10}] = \trunc(\val[g_1]-\val[g_2])/2 \pm 49\varepsilon$, and finally $\val[g_{11}] = \trunc(\val[g_1]-\val[g_2]) \pm 99\varepsilon$. Thus, we can compute subtraction with error at most $99\varepsilon$.

\medskip

\noindent\textbf{$\bm{G_\times}$: Multiplication.} Given gates $g_1$ and $g_2$, we want to obtain $\val[g_1] \cdot \val[g_2]$. We only perform the construction for the case $\varepsilon = 0$, since we only need this gate for the \fixp-hardness. Note that we can multiply by $4$ using two consecutive $G_{\times 2}$ gates. Similarly, we can divide by $4$ using two consecutive $G_{/2}$ gadgets. To perform multiplication, we use the fact that
$$\phi(x,y) - \frac{1}{4} - \frac{x}{4} - \frac{y}{4} = \frac{xy}{4}.$$
In more detail, we first use a gate $g_3$ of type $G_\phi$ with input $g_1$ and $g_2$, then a gate $g_4$ of type $G_{/4}$ with input the constant $1$, then a gate $g_5$ of type $G_-$ with inputs $g_3$ and $g_4$, then a gate $g_6$ of type $G_{/4}$ with input $g_1$, then a gate $g_7$ of type $G_-$ with inputs $g_5$ and $g_6$, then a gate $g_8$ of type $G_{/4}$ with input $g_2$, then a gate $g_9$ of type $G_-$ with inputs $g_7$ and $g_8$, and finally a gate $g_{10}$ of type $G_{\times 4}$ with input $g_9$.
We have that
$$\val[g_3] = \phi(\val[g_1],\val[g_2]) = (\val[g_1]+\val[g_2]+\val[g_1]\cdot\val[g_2]+1)/4.$$
Then we obtain that $\val[g_5] = (\val[g_1]+\val[g_2]+\val[g_1]\cdot\val[g_2])/4$, $\val[g_7] = (\val[g_2]+\val[g_1]\cdot\val[g_2])/4$, $\val[g_9] = \val[g_1]\cdot\val[g_2]/4$, and finally $\val[g_{10}] = \val[g_1]\cdot\val[g_2]$. Thus, we can perform exact multiplication when $\varepsilon = 0$.

\medskip

\noindent\textbf{Hardness.} We have shown that we can simulate gates $G_1$ and $G_-$ with error at most $99\varepsilon$. Thus, by \cref{prop:gcircuit-ppad}, the \ppad-hardness of our restricted version follows. For the case $\varepsilon = 0$, we have shown that we can exactly simulate gates $G_{\times 2}$, $G_{1-}$ and $G_\times$. As a result, by \cref{prop:gcircuit-fixp}, the exact version of our restricted version is \fixp-hard.

\bibliographystyle{plainnat}
\bibliography{first_price}

\begin{thebibliography}{74}
\providecommand{\natexlab}[1]{#1}
\providecommand{\url}[1]{\texttt{#1}}
\expandafter\ifx\csname urlstyle\endcsname\relax
  \providecommand{\doi}[1]{doi: #1}\else
  \providecommand{\doi}{doi: \begingroup \urlstyle{rm}\Url}\fi

\bibitem[Athey(2001)]{Athey2001}
Susan Athey.
\newblock Single crossing properties and the existence of pure strategy
  equilibria in games of incomplete information.
\newblock \emph{Econometrica}, 69\penalty0 (4):\penalty0 861--889, July 2001.
\newblock \doi{10.1111/1468-0262.00223}.

\bibitem[Babichenko and Rubinstein(2021)]{babichenko2020settling}
Yakov Babichenko and Aviad Rubinstein.
\newblock Settling the complexity of {N}ash equilibrium in congestion games.
\newblock In \emph{Proceedings of the 53rd ACM Symposium on Theory of Computing
  (STOC)}, pages 1426--1437, 2021.
\newblock \doi{10.1145/3406325.3451039}.

\bibitem[Battigalli and Guaitoli(1997)]{battigalli1988conjectural}
Pierpaolo Battigalli and Danilo Guaitoli.
\newblock Conjectural equilibria and rationalizability in a game with
  incomplete information.
\newblock In Pierpaolo Battigalli, Aldo Montesano, and Fausto Panunzi, editors,
  \emph{Decisions, Games and Markets}, pages 97--124. Springer, 1997.
\newblock \doi{10.1007/978-1-4615-6337-2_4}.

\bibitem[Battigalli et~al.(1992)Battigalli, Gilli, and
  Molinari]{battigalli1992learning}
Pierpaolo Battigalli, Mario Gilli, and M.~Cristina Molinari.
\newblock Learning and convergence to equilibrium in repeated strategic
  interactions: An introductory survey.
\newblock \emph{Ricerche Economiche}, 46:\penalty0 335--378, 1992.

\bibitem[Bergemann et~al.(2017)Bergemann, Brooks, and
  Morris]{bergemann2017first}
Dirk Bergemann, Benjamin Brooks, and Stephen Morris.
\newblock First-price auctions with general information structures:
  Implications for bidding and revenue.
\newblock \emph{Econometrica}, 85\penalty0 (1):\penalty0 107--143, 2017.
\newblock \doi{10.3982/ecta13958}.

\bibitem[Bhawalkar and Roughgarden(2011)]{bhawalkar2011welfare}
Kshipra Bhawalkar and Tim Roughgarden.
\newblock Welfare guarantees for combinatorial auctions with item bidding.
\newblock In \emph{Proceedings of the 22nd Annual {ACM}-{SIAM} Symposium on
  Discrete Algorithms (SODA)}, pages 700--709, January 2011.
\newblock \doi{10.1137/1.9781611973082.55}.

\bibitem[Bitansky et~al.(2015)Bitansky, Paneth, and
  Rosen]{bitansky2015cryptographic}
Nir Bitansky, Omer Paneth, and Alon Rosen.
\newblock On the cryptographic hardness of finding a {Nash} equilibrium.
\newblock In \emph{Proceedings of the 56th Annual Symposium on Foundations of
  Computer Science (FOCS)}, pages 1480--1498, October 2015.
\newblock \doi{10.1109/focs.2015.94}.

\bibitem[Cai et~al.(2010)Cai, Wurman, and Gong]{cai2010note}
Gangshu Cai, Peter~R. Wurman, and Xiting Gong.
\newblock A note on discrete bid first-price auction with general value
  distribution.
\newblock \emph{International Game Theory Review}, 12\penalty0 (01):\penalty0
  75--81, 2010.
\newblock \doi{10.1142/s0219198910002520}.

\bibitem[Cai and Papadimitriou(2014)]{cai2014simultaneous}
Yang Cai and Christos Papadimitriou.
\newblock Simultaneous {Bayesian} auctions and computational complexity.
\newblock In \emph{Proceedings of the 15th {ACM} Conference on Economics and
  Computation (EC)}, pages 895--910, June 2014.
\newblock \doi{10.1145/2600057.2602877}.

\bibitem[Caragiannis et~al.(2015)Caragiannis, Kaklamanis, Kanellopoulos,
  Kyropoulou, Lucier, Paes~Leme, and Tardos]{caragiannis2015bounding}
Ioannis Caragiannis, Christos Kaklamanis, Panagiotis Kanellopoulos, Maria
  Kyropoulou, Brendan Lucier, Renato Paes~Leme, and {\'{E}}va Tardos.
\newblock Bounding the inefficiency of outcomes in generalized second price
  auctions.
\newblock \emph{Journal of Economic Theory}, 156:\penalty0 343--388, March
  2015.
\newblock \doi{10.1016/j.jet.2014.04.010}.

\bibitem[Chawla and Hartline(2013)]{chawla2013auctions}
Shuchi Chawla and Jason~D. Hartline.
\newblock Auctions with unique equilibria.
\newblock In \emph{Proceedings of the 14th {ACM} conference on Electronic
  Commerce (EC)}, pages 181--196, 2013.
\newblock \doi{10.1145/2492002.2483188}.

\bibitem[Chen et~al.(2009)Chen, Deng, and Teng]{chen2009settling}
Xi~Chen, Xiaotie Deng, and Shang-Hua Teng.
\newblock Settling the complexity of computing two-player {Nash} equilibria.
\newblock \emph{Journal of the {ACM}}, 56\penalty0 (3):\penalty0 14:1--14:57,
  May 2009.
\newblock \doi{10.1145/1516512.1516516}.

\bibitem[Chen et~al.(2017)Chen, Paparas, and Yannakakis]{chen2013complexity}
Xi~Chen, Dimitris Paparas, and Mihalis Yannakakis.
\newblock The complexity of non-monotone markets.
\newblock \emph{Journal of the {ACM}}, 64\penalty0 (3):\penalty0 1--56, June
  2017.
\newblock \doi{10.1145/3064810}.

\bibitem[Cheng(2006)]{Cheng2006}
Harrison Cheng.
\newblock Ranking sealed high-bid and open asymmetric auctions.
\newblock \emph{Journal of Mathematical Economics}, 42\penalty0 (4-5):\penalty0
  471--498, aug 2006.
\newblock \doi{10.1016/j.jmateco.2006.05.008}.

\bibitem[Choudhuri et~al.(2019)Choudhuri, Hub{\'{a}}{\v{c}}ek, Kamath,
  Pietrzak, Rosen, and Rothblum]{choudhuri2019finding}
Arka~Rai Choudhuri, Pavel Hub{\'{a}}{\v{c}}ek, Chethan Kamath, Krzysztof
  Pietrzak, Alon Rosen, and Guy~N. Rothblum.
\newblock Finding a {Nash} equilibrium is no easier than breaking
  {Fiat}-{Shamir}.
\newblock In \emph{Proceedings of the 51st Annual {ACM} Symposium on Theory of
  Computing (STOC)}, pages 1103--1114, June 2019.
\newblock \doi{10.1145/3313276.3316400}.

\bibitem[Christodoulou et~al.(2016)Christodoulou, Kov\'{a}cs, and
  Schapira]{christodoulou2008bayesian}
George Christodoulou, Annam\'{a}ria Kov\'{a}cs, and Michael Schapira.
\newblock Bayesian combinatorial auctions.
\newblock \emph{Journal of the {ACM}}, 63\penalty0 (2), April 2016.
\newblock \doi{10.1145/2835172}.

\bibitem[Chwe(1989)]{chwe1989discrete}
Michael Suk-Young Chwe.
\newblock The discrete bid first auction.
\newblock \emph{Economics Letters}, 31\penalty0 (4):\penalty0 303--306,
  December 1989.
\newblock \doi{10.1016/0165-1765(89)90019-0}.

\bibitem[Conitzer and Sandholm(2008)]{Conitzer:2008aa}
Vincent Conitzer and Tuomas Sandholm.
\newblock New complexity results about {Nash} equilibria.
\newblock \emph{Games and Economic Behavior}, 63\penalty0 (2):\penalty0
  621--641, 2008.
\newblock \doi{10.1016/j.geb.2008.02.015}.

\bibitem[Daskalakis et~al.(2009)Daskalakis, Goldberg, and
  Papadimitriou]{daskalakis2009complexity}
Constantinos Daskalakis, Paul~W. Goldberg, and Christos~H. Papadimitriou.
\newblock The complexity of computing a {Nash} equilibrium.
\newblock \emph{{SIAM} Journal on Computing}, 39\penalty0 (1):\penalty0
  195--259, 2009.
\newblock \doi{10.1137/070699652}.

\bibitem[Deligkas et~al.(2021)Deligkas, Fearnley, Melissourgos, and
  Spirakis]{deligkas2019computing}
Argyrios Deligkas, John Fearnley, Themistoklis Melissourgos, and Paul~G.
  Spirakis.
\newblock Computing exact solutions of consensus halving and the
  {Borsuk}-{Ulam} theorem.
\newblock \emph{J. Comput. Syst. Sci.}, 117:\penalty0 75--98, 2021.
\newblock \doi{10.1016/j.jcss.2020.10.006}.

\bibitem[Digiday.com(2019)]{adexchange}
Digiday.com.
\newblock What to know about {Google’s} implementation of first-price ad
  auctions, 2019.
\newblock URL
  \url{https://digiday.com/media/buyers-welcome-auction-standardization-as-google-finally-goes-all-in-on-first-price}.
\newblock Accessed: 2019-09-06.

\bibitem[Escamocher et~al.(2009)Escamocher, Miltersen, and
  Santillan~R.]{escamocher2009existence}
Guillaume Escamocher, Peter~Bro Miltersen, and Rocio Santillan~R.
\newblock Existence and computation of equilibria of first-price auctions with
  integral valuations and bids.
\newblock In \emph{Proceedings of The 8th International Conference on
  Autonomous Agents and Multiagent Systems (AAMAS)}, pages 1227--1228, 2009.
\newblock URL \url{https://dl.acm.org/doi/10.5555/1558109.1558225}.

\bibitem[Etessami and Yannakakis(2010)]{etessami2010complexity}
Kousha Etessami and Mihalis Yannakakis.
\newblock On the complexity of {Nash} equilibria and other fixed points.
\newblock \emph{{SIAM} Journal on Computing}, 39\penalty0 (6):\penalty0
  2531--2597, January 2010.
\newblock \doi{10.1137/080720826}.

\bibitem[Fabrikant et~al.(2004)Fabrikant, Papadimitriou, and
  Talwar]{fabrikant2004complexity}
Alex Fabrikant, Christos Papadimitriou, and Kunal Talwar.
\newblock The complexity of pure {Nash} equilibria.
\newblock In \emph{Proceedings of the 36th Annual ACM Symposium on Theory of
  Computing (STOC)}, pages 604--612, 2004.
\newblock \doi{10.1145/1007352.1007445}.

\bibitem[Fearnley et~al.(2021)Fearnley, Goldberg, Hollender, and
  Savani]{fearnley2021complexity}
John Fearnley, Paul~W. Goldberg, Alexandros Hollender, and Rahul Savani.
\newblock The complexity of gradient descent: {CLS} = {PPAD} $\cap$ {PLS}.
\newblock In \emph{Proceedings of the 53rd ACM Symposium on Theory of Computing
  (STOC)}, pages 46--59, 2021.
\newblock \doi{10.1145/3406325.3451052}.

\bibitem[Feldman et~al.(2020)Feldman, Fu, Gravin, and
  Lucier]{feldman2013simultaneous}
Michal Feldman, Hu~Fu, Nick Gravin, and Brendan Lucier.
\newblock Simultaneous auctions without complements are (almost) efficient.
\newblock \emph{Games and Economic Behavior}, 123:\penalty0 327--341, September
  2020.
\newblock \doi{10.1016/j.geb.2015.11.009}.

\bibitem[Filos-Ratsikas et~al.(2021)Filos-Ratsikas, Giannakopoulos, Hollender,
  Lazos, and Poças]{fghlp2021_ec}
Aris Filos-Ratsikas, Yiannis Giannakopoulos, Alexandros Hollender, Philip
  Lazos, and Diogo Poças.
\newblock On the complexity of equilibrium computation in first-price auctions.
\newblock In \emph{Proceedings of the 22nd ACM Conference on Economics and
  Computation (EC)}, pages 454--476, 2021.
\newblock \doi{10.1145/3465456.3467627}.

\bibitem[Frongillo and Witkowski(2016)]{frongillo2016geometric}
Rafael Frongillo and Jens Witkowski.
\newblock A geometric method to construct minimal peer prediction mechanisms.
\newblock In \emph{Proceedings of the 30th AAAI Conference on Artificial
  Intelligence}, pages 502--508, 2016.
\newblock \doi{10.5555/3015812.3015888}.

\bibitem[Fudenberg and Levine(1986)]{fudenberg1986limit}
Drew Fudenberg and David Levine.
\newblock Limit games and limit equilibria.
\newblock \emph{Journal of Economic Theory}, 38\penalty0 (2):\penalty0
  261--279, April 1986.
\newblock \doi{10.1016/0022-0531(86)90118-3}.

\bibitem[Garg et~al.(2016{\natexlab{a}})Garg, Mehta, and
  Vazirani]{garg2016dichotomies}
Jugal Garg, Ruta Mehta, and Vijay~V. Vazirani.
\newblock Dichotomies in equilibrium computation and membership of {PLC}
  markets in {FIXP}.
\newblock \emph{Theory of Computing}, 12\penalty0 (20):\penalty0 1--25,
  2016{\natexlab{a}}.
\newblock \doi{10.4086/toc.2016.v012a020}.

\bibitem[Garg et~al.(2016{\natexlab{b}})Garg, Pandey, and
  Srinivasan]{garg2016revisiting}
Sanjam Garg, Omkant Pandey, and Akshayaram Srinivasan.
\newblock Revisiting the cryptographic hardness of finding a {Nash}
  equilibrium.
\newblock In \emph{Proceedings of the 36th Annual International Cryptology
  Conference (CRYPTO)}, pages 579--604, 2016{\natexlab{b}}.
\newblock \doi{10.1007/978-3-662-53008-5_20}.

\bibitem[Goldberg(2011)]{goldberg2011survey}
Paul~W. Goldberg.
\newblock A survey of {PPAD}-completeness for computing {N}ash equilibria.
\newblock In Robin Chapman, editor, \emph{Surveys in Combinatorics 2011},
  London Mathematical Society Lecture Note Series, pages 51--82. Cambridge
  University Press, 2011.
\newblock \doi{10.1017/CBO9781139004114.003}.

\bibitem[Goldberg and Hollender(2021)]{goldberg2019hairy}
Paul~W. Goldberg and Alexandros Hollender.
\newblock The {H}airy {B}all problem is {PPAD}-complete.
\newblock \emph{Journal of Computer and System Sciences}, 122:\penalty0 34--62,
  2021.
\newblock \doi{10.1016/j.jcss.2021.05.004}.

\bibitem[Goldberg et~al.(2022)Goldberg, Hollender, Igarashi, Manurangsi, and
  Suksompong]{goldberg2020consensus}
Paul~W. Goldberg, Alexandros Hollender, Ayumi Igarashi, Pasin Manurangsi, and
  Warut Suksompong.
\newblock Consensus halving for sets of items.
\newblock \emph{Mathematics of Operations Research}, 2022.
\newblock \doi{10.1287/moor.2021.1249}.

\bibitem[Gottlob et~al.(2007)Gottlob, Greco, and Mancini]{Gottlob:2007aa}
Georg Gottlob, Gianluigi Greco, and Toni Mancini.
\newblock Complexity of pure equilibria in bayesian games.
\newblock In \emph{Proceedings of the 20th International Joint Conference on
  Artifical Intelligence (IJCAI)}, pages 1294--1299, 2007.
\newblock \doi{10.5555/1625275.1625485}.

\bibitem[Griesmer et~al.(1967)Griesmer, Levitan, and
  Shubik]{griesmer1967toward}
James~H. Griesmer, Richard~E. Levitan, and Martin Shubik.
\newblock Toward a study of bidding processes part {IV} -- games with unknown
  costs.
\newblock \emph{Naval Research Logistics}, 14\penalty0 (4):\penalty0 415--433,
  1967.
\newblock \doi{10.1002/nav.3800140402}.

\bibitem[Grigor'ev and Vorobjov(1988)]{grigor1988solving}
D.~Yu. Grigor'ev and N.N. Vorobjov.
\newblock Solving systems of polynomial inequalities in subexponential time.
\newblock \emph{Journal of Symbolic Computation}, 5\penalty0 (1-2):\penalty0
  37--64, 1988.
\newblock \doi{10.1016/s0747-7171(88)80005-1}.

\bibitem[Hahn(1973)]{hahn1973notion}
Frank Hahn.
\newblock \emph{On the Notion of Equilibrium in Economics: An Inaugural Lecture
  [By] F.H. Hahn}.
\newblock Cambridge University Press, 1973.

\bibitem[Harsanyi(1967)]{harsanyi1967games}
John~C. Harsanyi.
\newblock Games with incomplete information played by ``{Bayesian}'' players,
  {I--III}: Part {I.} {T}he basic model.
\newblock \emph{Management Science}, 14\penalty0 (3):\penalty0 159--182, 1967.
\newblock \doi{10.1287/mnsc.1040.0270}.

\bibitem[Hartline(2012)]{hartline2012bayesian}
Jason~D. Hartline.
\newblock Bayesian mechanism design.
\newblock \emph{Foundations and Trends in Theoretical Computer Science},
  8\penalty0 (3):\penalty0 143--263, 2012.
\newblock \doi{10.1561/0400000045}.

\bibitem[Jehle and Reny(2001)]{Jehle2001a}
Geoffrey~A. Jehle and Philip~J. Reny.
\newblock \emph{Advanced Microeconomic Theory}.
\newblock Financial Times/Prentice Hall, 2001.

\bibitem[Johnson et~al.(1988)Johnson, Papadimitriou, and
  Yannakakis]{johnson1988easy}
David~S. Johnson, Christos~H. Papadimitriou, and Mihalis Yannakakis.
\newblock How easy is local search?
\newblock \emph{Journal of Computer and System Sciences}, 37\penalty0
  (1):\penalty0 79--100, 1988.
\newblock \doi{10.1016/0022-0000(88)90046-3}.

\bibitem[Kalai and Lehrer(1993)]{kalai1993rational}
Ehud Kalai and Ehud Lehrer.
\newblock Rational learning leads to {Nash} equilibrium.
\newblock \emph{Econometrica}, 61\penalty0 (5):\penalty0 1019--1045, 1993.
\newblock \doi{10.2307/2951492}.

\bibitem[Kalai and Lehrer(1995)]{kalai1995subjective}
Ehud Kalai and Ehud Lehrer.
\newblock Subjective games and equilibria.
\newblock \emph{Games and Economic Behavior}, 8\penalty0 (1):\penalty0
  123--163, 1995.
\newblock \doi{10.1016/s0899-8256(05)80019-3}.

\bibitem[Krishna(2009)]{krishna2009auction}
Vijay Krishna.
\newblock \emph{Auction Theory}.
\newblock Academic Press, 2nd edition, 2009.

\bibitem[Lebrun(1996)]{lebrun1996existence}
Bernard Lebrun.
\newblock Existence of an equilibrium in first price auctions.
\newblock \emph{Economic Theory}, 7:\penalty0 421--443, 1996.
\newblock \doi{10.1007/BF01213659}.

\bibitem[Lebrun(1999)]{lebrun1999first}
Bernard Lebrun.
\newblock First price auctions in the asymmetric {N} bidder case.
\newblock \emph{International Economic Review}, 40\penalty0 (1):\penalty0
  125--142, 1999.
\newblock \doi{10.1111/1468-2354.00008}.

\bibitem[Lebrun(2006)]{lebrun2006uniqueness}
Bernard Lebrun.
\newblock Uniqueness of the equilibrium in first-price auctions.
\newblock \emph{Games and Economic Behavior}, 55\penalty0 (1):\penalty0
  131--151, April 2006.
\newblock \doi{10.1016/j.geb.2005.01.006}.

\bibitem[Lizzeri and Persico(2000)]{lizzeri2000uniqueness}
Alessandro Lizzeri and Nicola Persico.
\newblock Uniqueness and existence of equilibrium in auctions with a reserve
  price.
\newblock \emph{Games and Economic Behavior}, 30\penalty0 (1):\penalty0
  83--114, January 2000.
\newblock \doi{10.1006/game.1998.0704}.

\bibitem[Lucier and Borodin(2010)]{lucier2010price}
Brendan Lucier and Allan Borodin.
\newblock Price of anarchy for greedy auctions.
\newblock In \emph{Proceedings of the 21st Annual {ACM}-{SIAM} Symposium on
  Discrete Algorithms (SODA)}, pages 537--553, January 2010.
\newblock \doi{10.1137/1.9781611973075.46}.

\bibitem[Marshall et~al.(1994)Marshall, Meurer, Richard, and
  Stromquist]{marshall1994numerical}
Robert~C. Marshall, Michael~J. Meurer, Jean-Francois Richard, and Walter
  Stromquist.
\newblock Numerical analysis of asymmetric first price auctions.
\newblock \emph{Games and Economic Behavior}, 7\penalty0 (2):\penalty0
  193--220, 1994.
\newblock \doi{10.1006/game.1994.1045}.

\bibitem[Maskin and Riley(2000)]{maskin2000equilibrium}
Eric Maskin and John Riley.
\newblock Equilibrium in sealed high bid auctions.
\newblock \emph{The Review of Economic Studies}, 67\penalty0 (3):\penalty0
  439--454, 2000.
\newblock \doi{10.1111/1467-937X.00138}.

\bibitem[Maskin and Riley(2003)]{maskin2003uniqueness}
Eric Maskin and John Riley.
\newblock Uniqueness of equilibrium in sealed high-bid auctions.
\newblock \emph{Games and Economic Behavior}, 45\penalty0 (2):\penalty0
  395--409, 2003.
\newblock \doi{10.1016/S0899-8256(03)00150-7}.

\bibitem[Megiddo and Papadimitriou(1991)]{megiddo1991total}
Nimrod Megiddo and Christos~H. Papadimitriou.
\newblock On total functions, existence theorems and computational complexity.
\newblock \emph{Theoretical Computer Science}, 81\penalty0 (2):\penalty0
  317--324, 1991.
\newblock \doi{10.1016/0304-3975(91)90200-l}.

\bibitem[Mehta(2018)]{mehta2018constant}
Ruta Mehta.
\newblock Constant rank two-player games are {PPAD}-hard.
\newblock \emph{SIAM Journal on Computing}, 47\penalty0 (5):\penalty0
  1858--1887, 2018.
\newblock \doi{10.1137/15M1032338}.

\bibitem[Milgrom and Shannon(1994)]{milgrom1994monotone}
Paul Milgrom and Chris Shannon.
\newblock Monotone comparative statics.
\newblock \emph{Econometrica}, 62\penalty0 (1):\penalty0 157--180, 1994.
\newblock \doi{10.2307/2951479}.

\bibitem[Myerson(1981)]{myerson1981optimal}
Roger~B. Myerson.
\newblock Optimal auction design.
\newblock \emph{Mathematics of Operations Research}, 6\penalty0 (1):\penalty0
  58--73, 1981.
\newblock \doi{10.1287/moor.6.1.58}.

\bibitem[Myerson(1997)]{myerson2013game}
Roger~B. Myerson.
\newblock \emph{Game Theory: Analysis of Conflict}.
\newblock Harvard University Press, 1997.

\bibitem[Paes~Leme and Tardos(2010)]{leme2010pure}
Renato Paes~Leme and {\'{E}}va Tardos.
\newblock Pure and {Bayes}-{Nash} price of anarchy for generalized second price
  auction.
\newblock In \emph{Proceedings of the 51st Annual Symposium on Foundations of
  Computer Science (FOCS)}, pages 735--744, October 2010.
\newblock \doi{10.1109/focs.2010.75}.

\bibitem[Paes~Leme et~al.(2020)Paes~Leme, Sivan, and Teng]{paes2020competitive}
Renato Paes~Leme, Balasubramanian Sivan, and Yifeng Teng.
\newblock Why do competitive markets converge to first-price auctions?
\newblock In \emph{Proceedings of The World Wide Web Conference (WWW)}, pages
  596--605, 2020.
\newblock \doi{10.1145/3366423.3380142}.

\bibitem[Papadimitriou(1994)]{papadimitriou1994complexity}
Christos~H. Papadimitriou.
\newblock On the complexity of the parity argument and other inefficient proofs
  of existence.
\newblock \emph{Journal of Computer and System Sciences}, 48\penalty0
  (3):\penalty0 498--532, 1994.
\newblock \doi{10.1016/s0022-0000(05)80063-7}.

\bibitem[Plum(1992)]{plum1992characterization}
Michael Plum.
\newblock Characterization and computation of {Nash}-equilibria for auctions
  with incomplete information.
\newblock \emph{International Journal of Game Theory}, 20\penalty0
  (4):\penalty0 393--418, December 1992.
\newblock \doi{10.1007/bf01271133}.

\bibitem[Rasooly and Gavidia-Calderon(2020)]{rasooly2020importance}
Itzhak Rasooly and Carlos Gavidia-Calderon.
\newblock The importance of being discrete: on the inaccuracy of continuous
  approximations in auction theory.
\newblock \emph{arXiv:2006.03016}, 2020.
\newblock URL \url{http://arxiv.org/abs/2006.03016}.

\bibitem[Reny and Zamir(2004)]{reny2004existence}
Philip~J. Reny and Shmuel Zamir.
\newblock On the existence of pure strategy monotone equilibria in asymmetric
  first-price auctions.
\newblock \emph{Econometrica}, 72\penalty0 (4):\penalty0 1105--1125, July 2004.
\newblock URL
  \url{https://onlinelibrary.wiley.com/doi/abs/10.1111/j.1468-0262.2004.00527.x}.

\bibitem[Riley and Samuelson(1981)]{riley1981optimal}
John~G. Riley and William~F. Samuelson.
\newblock Optimal auctions.
\newblock \emph{The American Economic Review}, 71\penalty0 (3):\penalty0
  381--392, 1981.
\newblock URL \url{https://www.jstor.org/stable/1802786}.

\bibitem[Rosen et~al.(2021)Rosen, Segev, and Shahaf]{rosen2017can}
Alon Rosen, Gil Segev, and Ido Shahaf.
\newblock Can {PPAD} hardness be based on standard cryptographic assumptions?
\newblock \emph{Journal of Cryptology}, 34\penalty0 (1), 2021.
\newblock \doi{10.1007/s00145-020-09369-6}.

\bibitem[Rosenthal(1973)]{rosenthal1973class}
Robert~W. Rosenthal.
\newblock A class of games possessing pure-strategy {Nash} equilibria.
\newblock \emph{International Journal of Game Theory}, 2\penalty0 (1):\penalty0
  65--67, 1973.
\newblock \doi{10.1007/BF01737559}.

\bibitem[Rubinstein and Wolinsky(1994)]{rubinstein1994rationalizable}
Ariel Rubinstein and Asher Wolinsky.
\newblock Rationalizable conjectural equilibrium: Between {Nash} and
  rationalizability.
\newblock \emph{Games and Economic Behavior}, 6\penalty0 (2):\penalty0
  299--311, March 1994.
\newblock \doi{10.1006/game.1994.1016}.

\bibitem[Rubinstein(2018)]{rubinstein2018inapproximability}
Aviad Rubinstein.
\newblock Inapproximability of {Nash} equilibrium.
\newblock \emph{{SIAM} Journal on Computing}, 47\penalty0 (3):\penalty0
  917--959, 2018.
\newblock \doi{10.1137/15m1039274}.

\bibitem[Vazirani and Yannakakis(2011)]{vazirani2011market}
Vijay~V. Vazirani and Mihalis Yannakakis.
\newblock Market equilibrium under separable, piecewise-linear, concave
  utilities.
\newblock \emph{Journal of the {ACM}}, 58\penalty0 (3):\penalty0 1--25, May
  2011.
\newblock \doi{10.1145/1970392.1970394}.

\bibitem[Vickrey(1961)]{vickrey1961counterspeculation}
William Vickrey.
\newblock Counterspeculation, auctions and competitive sealed tenders.
\newblock \emph{Journal of Finance}, 16\penalty0 (1):\penalty0 8--37, March
  1961.
\newblock \doi{10.1111/j.1540-6261.1961.tb02789.x}.

\bibitem[Wang et~al.(2020)Wang, Shen, and Zuo]{wang2020bayesian}
Zihe Wang, Weiran Shen, and Song Zuo.
\newblock {Bayesian Nash} equilibrium in first-price auction with discrete
  value distributions.
\newblock In \emph{Proceedings of the 19th International Conference on
  Autonomous Agents and Multiagent Systems (AAMAS)}, pages 1458--1466, 2020.
\newblock URL \url{https://dl.acm.org/doi/abs/10.5555/3398761.3398929}.

\bibitem[Witkowski and Parkes(2012)]{witkowski2012peer}
Jens Witkowski and David~C. Parkes.
\newblock Peer prediction without a common prior.
\newblock In \emph{Proceedings of the 13th ACM Conference on Electronic
  Commerce}, pages 964--981, 2012.
\newblock \doi{10.1145/2229012.2229085}.

\bibitem[Yannakakis(2009)]{yannakakis2009equilibria}
Mihalis Yannakakis.
\newblock Equilibria, fixed points, and complexity classes.
\newblock \emph{Computer Science Review}, 3\penalty0 (2):\penalty0 71--85, May
  2009.
\newblock \doi{10.1016/j.cosrev.2009.03.004}.

\end{thebibliography}

\end{document}